\documentclass[letterpaper,11pt]{article}
\usepackage[margin=1in]{geometry}

\usepackage[switch, modulo]{lineno}

\usepackage{mathtools}
\usepackage[T1]{fontenc}
\usepackage{fnpct}
\usepackage{microtype}
\usepackage{paralist}
\usepackage[ttscale=.85]{libertine}
\usepackage{libertinust1math}
\usepackage{sectsty}
\usepackage[size=tiny]{todonotes}
\usepackage[normalem]{ulem}
\usepackage{hyperref}

\usepackage{enumerate}
\usepackage{xspace} 
\usepackage{tcolorbox}
\usepackage{algorithm}
\usepackage{algpseudocode}

\newif\ifdraft
\drafttrue
\usepackage[utf8]{inputenc}
\usepackage[maxbibnames=99,sortcites]{biblatex}
\usepackage{framed}
\usepackage{tcolorbox}

\usepackage{enumitem}
\usepackage{blindtext}

\newcommand{\E}[1]{\mathbf{E}\left[#1\right]}
\newcommand{\Exp}[1]{\mathbb{E}\left[#1\right]}
\newcommand{\Prob}[1]{\mathbf{P}\left(#1\right)}
\newcommand{\Pro}[1]{\mathbb{P}\left(#1\right)}

\newcommand{\poly}{\operatorname{\text{{\rm poly}}}}

\newcommand{\eps}{\varepsilon}
\newcommand{\whp}{w.h.p.\xspace}
\newcommand{\wehp}{w.e.h.p.\xspace}
\newcommand{\aas}{a.a.s.\xspace}

\newcommand{\dd}{\,\mathrm d}

\newcommand{\B}{\mathcal B}
\newcommand{\Hb}{\ensuremath{\mathbb H^2}\xspace}
\newcommand{\bH}{\Hb}
\newcommand{\disk}{\ensuremath{\mathcal D_R}\xspace}

\newcommand{\hrg}{\mathcal{G}(n, \alpha, C)}

\newcommand{\bigO}{\mathcal{O}}

\DeclareMathOperator{\polylog}{polylog}

\DeclareMathOperator{\dist}{d_h}

\newcommand{\intensity}{\lambda}
\newcommand{\func}{\rho}

\usepackage[capitalize,noabbrev]{cleveref}
\usepackage[dvipsnames]{xcolor}
\definecolor{customred}{RGB}{225, 90, 90}
\usepackage{xspace}
\usepackage{tikz}
\usepackage[framemethod=tikz]{mdframed}
\mdfsetup{frametitlealignment=\centering}

\usepackage{uniinput}

\usepackage{graphicx,tipa}
\usepackage{caption}
\usepackage{subcaption}

\usepackage{amsmath}
\usepackage{amssymb}

\usepackage{amsthm}
\usepackage{thmtools}
\usepackage{thm-restate}

\declaretheorem{theorem}
\newtheorem*{theorem*}{Theorem}
\declaretheorem{lemma, proposition, corollary, remark, observation, definition, claim}[
style=plain,
sibling=theorem
]

\makeatletter 
\DeclareFontFamily{U}{tipa}{}
\DeclareFontShape{U}{tipa}{m}{n}{<->tipa10}{}
\newcommand{\arc@char}{{\usefont{U}{tipa}{m}{n}\symbol{62}}}%
\newcommand{\arc}[1]{\mathpalette\arc@arc{#1}}
\newcommand{\arc@arc}[2]{%
  \sbox0{$\m@th#1#2$}%
  \vbox{
    \hbox{\resizebox{\wd0}{\height}{\arc@char}}
    \nointerlineskip
    \box0
  }%
}
\makeatother

\newcounter{ctr}
\loop
 \stepcounter{ctr}
 \expandafter\edef\csname c\Alph{ctr}\endcsname{\noexpand\mathcal{\Alph{ctr}}}
\ifnum\thectr<26
\repeat

\usepackage{authblk}
\usepackage{fontawesome5}

\newcommand{\authoremail}[1]{%
  \href{mailto:#1}{\tiny\raisebox{4pt}{\faEnvelope[regular]}}%
}

\newcommand{\CONGEST}{\ensuremath{\mathsf{CONGEST}}\xspace}
\newcommand{\LOCAL}{\ensuremath{\mathsf{LOCAL}}\xspace}

\newcommand{\Layer}[1]{\mathcal{L}_{#1}}
\newcommand{\tiletile}[2]{\mathcal{T}_{#1,#2}}
\newcommand{\block}[2]{\mathcal{A}_{#1,#2}}
\newcommand{\angulardist}[2]{\delta_\varphi(#1,#2)}
\newcommand{\rootlayer}[0]{\ell_0}
\newcommand{\boxbelow}[2]{\mathcal{B}_{#1, #2}}

\author[1]{Yannic Maus\thanks{This research was funded in whole or in part by the Austrian Science Fund (FWF) \url{https://doi.org/10.55776/P36280}, \url{https://doi.org/10.55776/I6915}. For open access purposes, the author has applied a CC BY public copyright license to any author-accepted manuscript version arising from this submission.}\authoremail{yannic.maus@tugraz.at}}
\author[2]{{Janosch Ruff\thanks{This research was partially funded by the German Research Foundation (Deutsche Forschungsgemeinschaft, DFG) – project number 390859508.}}\authoremail{Janosch.Ruff@hpi.de}}
\author[2]{{Sonia Simons}\authoremail{sonia.simons@student.hpi.uni-potsdam.de}}
\author[2]{{George Skretas}\authoremail{Georgios.Skretas@hpi.de}}

\affil[1]{TU Graz, Austria}
\affil[2]{Hasso Plattner Institute, University of Potsdam, Germany}

\title{Distributed Symmetry Breaking on Hyperbolic Random Graphs}
\date{}

\begin{document}
\allsectionsfont{\sffamily}
\maketitle
\thispagestyle{empty}

\begin{abstract}
Real-world networks like the internet share patterns like a power law degree distribution and a high clustering coefficient. Many of these properties are captured by the generative model of hyperbolic random graphs (HRGs), which provides a theoretical framework for studying such networks. Motivated by the observation that several algorithms perform better on real-world networks than their worst-case guarantees suggest, we design and analyse distributed algorithms under the assumption that the input graph is an HRG. 
Indeed, prior work has shown that the classical symmetry-breaking problem of $\Delta+1$ colouring, where $\Delta$ is the maximum degree of the graph, can be solved in 2 rounds on HRGs [Maus and Ruff; SODA'26]. 

In stark contrast to this 2-round algorithm for $\Delta+1$ colouring, we prove that the related symmetry-breaking problems of maximal independent set (MIS) and maximal matching (MM) are substantially harder: we establish a lower bound of $\Omega\left(\frac{\log\log n}{\log\log\log n}\right)$ for MIS and MM on HRGs. Our lower bound techniques rely on new structural insights that may be of independent interest: we show that HRGs contain $d$-ary trees with large height and degree which enables us to adapt and lift prior impossibility results for distributed algorithms to the setting of HRGs. 

We also show that these lower bounds are polynomial tight: we design algorithms tailored to HRGs that solve MIS and MM in $\tilde{\bigO}(\log^{5/3}\log n)$ rounds with high probability in the \LOCAL model, improving over the general worst-case lower bound of $\Omega\left(\min\left\{\log \Delta, \sqrt{\log n}\right\}\right)$ rounds [Khoury and Schild; FOCS'25].

Finally, we show that access to geometric information can significantly reduce the complexity of maximal matching: if vertices know their respective geometric embeddings, then MM can be solved in $\bigO(\log\log\log n)$ rounds on HRGs. This reveals a separation between the standard \LOCAL model and an embedding-aware variant on hyperbolic random graphs.
\end{abstract}

\clearpage
\thispagestyle{empty}
\tableofcontents

\newpage
\setcounter{page}{1}
\section{Introduction}
\emph{Symmetry breaking} problems play a key role in the theory of distributed computing. In this paper, we study the fundamental symmetry breaking problems \emph{maximal independent set} (MIS) and \emph{maximal matching} (MM). A maximal independent set is an independent set that cannot be extended by adding any vertex, while a maximal matching is a matching where no additional edge can be added. We study these problems through the lenses of the classic \LOCAL model and the \CONGEST model \cite{Linial-92,Peleg-00}: A communication network is represented by a graph $G=(V,E)$ with $|V|=n$. In synchronous rounds, vertices exchange messages with their neighbours, and the round complexity of an algorithm is the number of rounds until all vertices have produced their output. We consider both the \LOCAL model, where messages have unbounded size, and the \CONGEST model, where each message is limited to $\bigO(\log n)$ bits.

One of the classic results from the 1980's is \emph{Luby's Algorithm}, giving a randomised algorithm that finds an MIS within $\bigO(\log n)$ rounds with high probability \cite{luby86, Alon1986}. The same approach can then also be extended to an MM. Remarkably, despite much progress on distributed computing since then, $\bigO(\log n)$ remains the best-known round complexity as a function of $n$ for general graphs \cite{barenboimelkin_book, BEPS, ghaffari-soda-16, HarrisSS18, khoury2025round, ks-stoc-26, Barenboim2009, ghaffari2023faster, GG24, gghir-soda-23, ghaffariDISC19, lenzen2011mis, Mtivier2010, Schneider2010, RG20}. Meanwhile, Khoury and Schild recently showed a lower bound of $\Omega(\sqrt{\log n})$ for MIS and MM on worst-case general graphs \cite{khoury2025round}, improving upon an earlier bound of $\Omega(\sqrt{\log n/\log\log n})$ rounds by Kuhn, Moscibroda, and Wattenhofer~\cite{kmw-jacm-2016}. 

Worst-case graphs, however, may be a poor proxy for the communication networks that motivate many distributed problems. This raises the question of whether symmetry breaking becomes easier when the communication network is drawn from a model that captures structural features commonly observed in real-world networks. A recent step in this direction was taken by Maus and Ruff, who showed that a distributed colouring problem requiring $\Omega(\log^* n)$ rounds on worst-case graphs \cite{Linial-87, Naor91} can be solved in just two rounds on \emph{hyperbolic random graphs}, a common abstraction of real-world networks \cite{mr-soda-26}. Their simple algorithm significantly outperforms the state of the art for general graphs. 
 Their work raises the broader question of whether this dramatic speed-up is specific to colouring, or whether it reflects a more general phenomenon. 
Thus, in this paper, we further explore the following research question: 

\begin{tcolorbox}
How does the structure of real-world networks affect the complexity of symmetry breaking?
\end{tcolorbox}

While different real-world networks have different structures, e.g., the internet graph has a different topology than a road network, many real-world networks have been observed to share several topological features. For example, many networks exhibit a heterogeneous degree distribution, close to a power-law.\footnote{That is, the number of vertices that have degree $d$ is roughly $\approx n\cdot d^{-\tau}$ (where $\tau \in (2, 3)$ is the \emph{power-law exponent}).} Examples include, among others, social networks~\cite{vhhk-s-19}, communication networks like the internet \cite{faloutsos1999power}, and biological networks \cite{koonin2006power}. Another property observed in real-world networks is a high clustering coefficient \cite{PhysRevE.74.056114, Newman2003, Watts1998}. We follow the lines drawn in \cite{pkbv-info-2010, kpk-h-10, Bogu2010, Serrano2008} using hyperbolic geometry to encapsulate both of these properties. In particular, we use the model of \emph{hyperbolic random graph} (HRG) as introduced by Krioukov, Papadopoulos, Kitsak, Vahdat, and Boguñá \cite{kpk-h-10}. At this point, HRGs are well established as a model for complex real-world networks, and many different centralised algorithmic problems have been studied \cite{bffkmm-icalp-2022,bfk-chrg-18, bfk-tw-2016, katzmann-exactvc-2023, katzmann-approxvc-2023, Kiwi2024, bmrs-stacs-25}. Moreover, the excellent empirical work by Bläsius and Fischbeck suggests that the theoretical analysis of algorithms on HRGs provides a good predictor for performance on real-world networks \cite{bf-evacaga-22}, and HRGs also serve as the underlying model in the work of Maus and Ruff \cite{mr-soda-26}. A \emph{threshold hyperbolic random graph} is obtained by sampling $n$ vertices into a hyperbolic disk with radius $R \approx 2\log n$, and a pair of vertices share an edge if the hyperbolic distance between the two vertices is at most $R$. We provide a formal definition in \Cref{sec:prelims} and refer the interested reader to \cite{Krohmer2016, katz-diss-23} for an in-depth introduction to HRGs. 
\subsection{Our Contributions}
  \begin{figure}[t]
    \centering \includegraphics{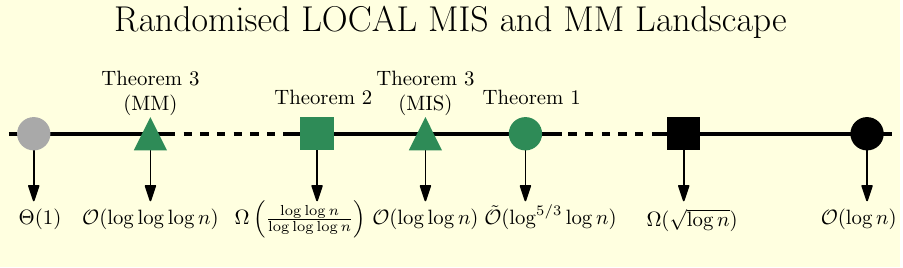} 
    \caption{Landscape of randomised \LOCAL Maximal Independent Set (MIS) and Maximal Matching (MM). Green: our results for MIS/MM on hyperbolic random graphs. Black: current state of the art for MIS/MM on general graphs. Grey: distributed colouring on HRGs. Disks: upper bounds (\Cref{thm:mainpolylog}). Squares: lower bounds (\Cref{thm:lowerbound}). Triangles: embedding-aware upper bounds (\Cref{thm:geometric}).}
    \label{fig:results}
  \end{figure}
We show that the complexity landscape of MIS and MM on HRGs differs dramatically from that of the distributed graph colouring problem studied in \cite{mr-soda-26}; see also \Cref{fig:results} for visualisiation of our results: On the algorithmic side, we show that MIS and MM can be solved exponentially faster on HRGs than on worst-case graphs by giving $\log^{\bigO(1)}\log n$-round algorithms (\Cref{thm:mainpolylog}). But, on the lower-bound side, we show that this speed-up has inherent limits. In contrast to graph colouring, MIS and MM do not collapse to constant time on HRGs; in fact, our lower bounds show that any further improvement to our upper bounds is limited to a small polynomial improvement (\Cref{thm:lowerbound}). 
This impossibility result relies on new structural insights into HRGs: we prove that their giant component contains substructures (trees) on which faster distributed algorithms are impossible (\Cref{pro:d-regular-trees}).
We believe the existence of these substructures is of independent interest, as it contributes to the structural understanding of HRGs as models of real-world networks.

\smallskip

The lower bound above applies in the standard \LOCAL model, where nodes do not know their position in the underlying hyperbolic geometry. This naturally raises the question of whether such geometric information can be exploited algorithmically. In the classic \LOCAL model, the initial knowledge of nodes is limited to their own ID and potentially some coarse upper bounds on global parameters such as the maximum degree or the number of nodes. However, in real-world networks, nodes may have additional information about the network that may be exploited algorithmically. For maximal matching, we show that such extra information can indeed be useful: we design an algorithm that beats the aforementioned lower bound when nodes have access to their geometric HRG coordinates, giving a doubly-exponential improvement over the complexity on general graphs (\Cref{thm:geometric}). 

We next present our results in detail and provide further background. 

\paragraph{Efficient Symmetry Breaking Algorithms.} 
We obtain the following algorithms for MIS and MM on HRGs that are exponentially faster than the best algorithms for general graphs~\cite{luby86,ghaffari-soda-16,BEPS}. 
\begin{restatable}[Upper bounds for MIS and MM]{theorem}{maintheorem}
\label{thm:mainpolylog}
    There are randomised distributed algorithms that, for threshold hyperbolic random graphs, compute Maximal Independent Set and Maximal Matching
    \begin{itemize}
        \item in $\tilde{\bigO}(\log^{5/3} \log n)$ rounds $\wehp$\footnote{We say an event $\mathcal{E}$ occurs \emph{with extremely high probability} (\wehp) if $\Pro{\mathcal{E}} \in 1 - n^{-\omega(1)}$.} of the \LOCAL model;
        \item  in $\tilde{\bigO}(\log^{3} \log n)$ rounds $\wehp$ of the \CONGEST model.\footnote{For MM, we obtain the slightly stronger bound of $\bigO(\log^3\log n)$.}
    \end{itemize}
\end{restatable}
All probabilistic statements in our theorems are with respect to both the random generation of the HRG and the random choices made by the algorithm.
See \Cref{fig:results} for a visual representation of our results in comparison to runtimes of existing work.
All algorithms of \Cref{thm:mainpolylog} use a two-step approach. For MIS, the first step consists of two iterations of a Luby-style algorithm in which certain nodes activate themselves, draw a random number, and nodes whose drawn numbers are local minima join the independent set. As a result, we obtain that the graph \emph{shatters} into exponentially smaller remaining unsolved connected components, on each of which we can then run the best deterministic algorithm exploiting the component's small size, namely the algorithm of Ghaffari and Grunau for MIS/MM in the \LOCAL model~\cite{GG24} and the fastest deterministic algorithms in the \CONGEST model~\cite{fggkr-soda-23}. Such a two-step approach is usually referred to as the \emph{shattering technique} and has been used in essentially all state-of-the-art symmetry breaking algorithms, e.g., \cite{BEPS,ghaffari-soda-16,CLP18}, and one also knows that the randomised complexity is lower bounded by the deterministic complexity on exponentially smaller instances \cite{CKP19}. A central difference in our work and the main contribution of \Cref{thm:mainpolylog}  is that we use a fundamentally different argument than all these existing works to show that such a shattering phenomenon emerges. This is interesting as recently discovered flaws and missing details in existing shattering approaches ask for more different shattering approaches \cite{GHMN26}.

One might hope that ordinary Luby iterations already yield such a shattering. We show that this is not the case: for any constant number of Luby iterations, both large components and large degrees remain; see \Cref{sec:hrg-aint-nuthin-to-fuck-with} for details.
\paragraph{Lower bounds.} \Cref{thm:mainpolylog} is not only faster than the known fastest algorithms for MIS and MM but also exponentially faster than the lower bound of $\Omega\left(\min\{\log \Delta, \sqrt{\log n}\}\right)$ by Khoury and Schild \cite{khoury2025round} for general graphs,  since the maximum degree $\Delta$ of an HRG is polynomial in $n$. Combined with the fact that the other central symmetry-breaking problem of $\Delta+1$-colouring, and even colourings with fewer colours, can be solved in constant time~\cite{mr-soda-26}, suggests that the algorithms in \Cref{thm:mainpolylog} may be far from optimal. We show that this is not the case. Any improvements to \Cref{thm:mainpolylog} can only be polynomial in its runtime. More precisely, we prove the following theorem, providing a strong separation between the runtime of MIS/MM and colouring problems on HRGs. 
\begin{restatable}[Lower bounds for MIS and MM]{theorem}{theoremlowerbound}
\label{thm:lowerbound}
 Asymptotically almost surely\footnote{With probability $1-o(1)$ over the draw of the hyperbolic random graph $G$.}, the randomised complexity\footnote{There exists an algorithm where the error probability is at most $1/n^{c}$ for some constant $c>0$.} for computing an MIS or an MM for the giant component of a hyperbolic random graph is 
$\Omega\left(\frac{\log\log n}{\log\log\log n}\right)$ rounds in the \LOCAL model.
\end{restatable}
\paragraph{On substructures of hyperbolic random graphs.} 
To prove \Cref{thm:lowerbound} we show that with high probability any HRG contains relatively large induced $d$-ary trees\footnote{A rooted tree where each vertex except for the leaves has branching factor $d$.} and then leverage the fact that distributed algorithms for MIS/MM cannot be fast on tree-like graphs~\cite{kmw-jacm-2016,khoury2025round,Coupette2021,bbhors-jacm-2019,Balliu2026New, bgko-podc-23,BBKO22}.
 \begin{theorem*}[Informal version of \Cref{pro:d-regular-trees}]\label{thm:trees}
In an HRG, asymptotically almost surely, there exist polynomially many induced $d$-ary trees with $n' \approx \sqrt{\log n}$ vertices, any degree $d$, depth $h \approx \log_d\log n$ and where the root of each tree has one additional edge to a vertex that connects it to the giant component.
\end{theorem*}
Then \Cref{thm:mainpolylog} implies that \Cref{pro:d-regular-trees} is in some sense almost-tight, i.e., significantly larger trees cannot exist.  For example, $d$-ary trees with degree $d \approx e^{\log^2\log n}$ and height $h \approx \log^2\log n$ cannot exist, as the resulting lower bound for computing MIS on HRGs would contradict the existence of the algorithm in \Cref{thm:mainpolylog}. In fact, we consider \Cref{pro:d-regular-trees} as a main technical contribution of our work that is of independent interest. More generally, these structural insights contribute to existing literature of analysing the structure of HRGs, like the use of ``dangling paths'' to obtain bounds on the cover and hitting times of random walks \cite{Kiwi2024}, bounds on the degeneracy to obtain efficient colouring approximations \cite{bmrs-stacs-25}, the construction of multi-commodity flows to bound the spectral gap \cite{hrg-spectral} or more broadly the connectivity of an HRG \cite{bfm-giant-15, fm-giant-18, km-slcrhg-19, ms-k-19, fk-dhrg-18, DBLP:journals/im/AbdullahFB17}.
\paragraph{Embedding-aware Algorithms.} The lower bounds of \Cref{thm:lowerbound} apply in the standard distributed setting, where nodes only interact through the graph and have no access to the underlying hyperbolic embedding.  This raises the question of whether this hidden structure can be exploited when it is made available to the algorithm. We therefore also study a coordinate-aware variant of the model, in which each node initially knows its position in the hyperbolic disk and can communicate this information to its neighbours in one round. This is a strong assumption, but geometric information of this kind is available or can be inferred in several network settings \cite{Bogu2010}. Similar assumptions have also been considered for distributed MIS in Euclidean unit disk graphs, where they lead to deterministic $\bigO(1)$-round algorithms in the \CONGEST model~\cite{Molla2019} and to efficient maximum independent set approximations for graphs embedded in the hyperbolic plane \cite{bfk-tw-2016, thomas-socg-25}.

Under this additional geometric information, the complexity landscape changes again. For maximal matching, we break the lower-bound barrier of \Cref{thm:lowerbound} and obtain a randomised $\bigO(\log\log\log n)$-round algorithm in the \CONGEST model.  This shows, similar to quantum algorithms which yield an advantage for some distributed problems~\cite{quantum-stoc, quantum-soda}, that one can beat lower bounds if additional computational resources like that of a geometric embedding are available. For maximal independent set, we obtain a $\bigO(\log\log n)$-round \CONGEST algorithm on HRGs with known geometric coordinates.
\begin{restatable}[Embedding-aware algorithms]{theorem}{theoremgeometric}
\label{thm:geometric}
    There are distributed algorithms that, for threshold hyperbolic random graphs given in their geometric representation, compute 
     \begin{itemize}
        \item Maximal Independent Set in ${\bigO}(\log\log n)$ rounds \whp\footnote{We say an event $\mathcal{E}$ occurs \emph{with high probability} (\whp) if $\Pro{\mathcal{E}} \in 1 - \bigO(1/n)$.} of the \CONGEST model;
        \item Maximal Matching in ${\bigO}(\log\log\log n)$ rounds \whp of the \CONGEST model.
    \end{itemize}
\end{restatable}
While we do not break the lower bound barrier of \Cref{thm:lowerbound} for MIS, we obtain a runtime of a ``flat'' $\log\log n$ for a deterministic algorithm. Moreover, the techniques do not rely on shattering as does \Cref{thm:mainpolylog} and, as such, provide a different approach for solving symmetry-breaking problems on real-world networks.

Finally, we remark that \Cref{thm:lowerbound} together with \Cref{thm:geometric} imply that a distributed computation of the embedding for an HRG requires $\Omega\left(\frac{\log\log n}{\log\log\log n}\right)$ rounds. 
\subsection{Further Related Work}
\paragraph{Hyperbolic random graphs.} A hyperbolic random graph combines the non-vanishing clustering coefficient of random geometric graphs (RGGs) \cite{p-rgg-03} and the power-law degree distribution exhibited by models like Chung--Lu graphs \cite{cl-ccrgg-02, cl-adrgged-02} or the Barabási–Albert model \cite{preferal-attatch}. Other properties like that of a giant component\footnote{A unique connected component of size $\Theta(n)$.} \cite{bfm-giant-15, fm-giant-18, bfkrz-esa-2023, Jorritsma2025} and the small-world phenomena\footnote{The largest distance between vertices is $\Theta(\log n)$.} \cite{DBLP:conf/analco/KiwiM15,fk-dhrg-18, ms-k-19, kostas-diameter, DBLP:journals/im/AbdullahFB17} follow. From the algorithmic side problems like Shortest Path \cite{bffkmm-icalp-2022}, Maximum Clique \cite{bfk-chrg-18}, Maximum Independent Set and Maximum Matching \cite{bfk-tw-2016}, Vertex Cover \cite{katzmann-exactvc-2023, katzmann-approxvc-2023}, Random Walks \cite{Kiwi2024} or Colouring \cite{bmrs-stacs-25, mr-soda-26} have been studied. 

A closely related model which captures the same properties is \emph{Geometric Inhomogeneous Random Graphs} (GIRGs) introduced by Bringmann, Keusch and Lengler \cite{Bringmann2019}. While HRGs and GIRGs are often regarded as roughly equivalent \cite{Komjthy2020}, several differences between the two models have been identified \cite{bmrs-stacs-25,cliques-scale-free-2024,bfk-eggihrg-22}. GIRGs have been used to analyse algorithms like routing~\cite{yannic-greedy}, shortest paths~\cite{johannes-2026}, rumour spreading \cite{kostas-rumour}, community testing \cite{Bet_Michielan_Stegehuis_2025}, geometry detection \cite{clara-21} or first passage percolation \cite{Komjthy2024}.

Another related model is the class of \emph{hyperbolic uniform disk graphs} (HUDGs) \cite{thomas-socg-25,bdhm-socg26,bks-hudg-23}; previously pioneered by Kisfaludi-Bak \cite{Kisfaludi-Bak-SODA20}. In this model, vertices are placed arbitrarily in a hyperbolic disk of radius $R$, and two vertices are adjacent whenever their hyperbolic distance is at most $R$. Unlike hyperbolic random graphs, no probability distribution is used to generate the vertex positions. When $R=2\log n$, the model can be viewed as a deterministic or worst-case analogue of a hyperbolic random graph: the same geometric threshold is used, but the vertex positions are chosen adversarially rather than sampled from the hyperbolic random graph distribution.
\paragraph{Distributed maximal independent set and matching.} 
While Luby's algorithm with $\bigO(\log n)$ rounds remains the fastest known algorithm on general graphs, substantial progress has been made on MIS and MM along other directions. From the algorithmic point of view, deterministic algorithms have been developed where the current state of the art for the \LOCAL model is $\tilde{\bigO}(\log^{5/3} n)$ due to Ghaffari and Grunau~\cite{GG24}, and in \CONGEST $\bigO(\log^2 \Delta\cdot \log\log \Delta \cdot \log n)$ rounds for MIS \cite{fggkr-soda-23} and $\bigO(\log^2\Delta \cdot \log n)$ for MM \cite{fischer2020improved}. Moreover, Barenboim, Elkin, Pettie and Schneider introduced a randomised algorithm that runs faster on graphs with small maximum degree $\Delta$ \cite{BEPS}. In his seminal paper, Ghaffari, improved on their MIS algorithm and accomplished a round complexity of $\bigO\left(\log \Delta + \poly\log\log n \right)$ in the \LOCAL model \cite{ghaffari-soda-16,RG20, GG24} and $\bigO\left(\log \Delta\log\log n + \poly\log\log n \right)$ in the \CONGEST model~\cite{G19}. Recently, it was shown by Khoury and Schild that these algorithms solve MIS on trees in $o\left(\sqrt{\log n}\right)$ rounds \cite{ks-stoc-26} while MM on trees requires $\Theta(\sqrt{\log n})$ rounds \cite{khoury2025round, BEPS, ghaffari-soda-16}. This shows that MM is harder than MIS on trees, while on general graphs, MIS is just as hard as MM.\footnote{This can be seen by taking the line graph $H$ of graph $G$ such that an MM on the line graph $H$ is equivalent to an MIS on $G$.} 

Within the broader context of lower bounds, due to a classic construction by Kuhn, Moscibroda and Wattenhofer there is an $\Omega\left(\min\left\{\frac{\log \Delta}{\log\log \Delta}, \sqrt{\frac{\log n}{\log\log n}} \right\}\right)$ lower bound for randomised MIS and MM on certain irregular trees \cite{kmw-jacm-2016} (see also \cite{Coupette2021} for a simplified version). For Maximal matching, the breakthrough work by Balliu, Brandt, Hirvonen, Olivetti, Rabie and Suomela introduced an $\Omega\left(\min\left\{\Delta, {\frac{\log n}{\log\log n}} \right\}\right)$ deterministic lower bound on regular trees\footnote{This work was honoured with the FOCS'19 best paper award.} \cite{bbhors-jacm-2019}. Recently  Khoury and Schild improved the randomised lower bound to $\Omega\left(\min\left\{\log \Delta, \sqrt{\log n} \right\}\right)$ for randomised MM on regular trees \cite{khoury2025round}. For MIS on trees, Balliu, Ghaffari, Kuhn and Olivetti established a lower bound of $\Omega\left(\min\left\{\frac{\log \Delta}{\log\log \Delta}, \sqrt{\frac{\log n}{\log\log n}} \right\}\right)$ for randomised algorithms \cite{bgko-podc-23}. As previously mentioned, MIS can be solved in $o(\sqrt{\log n})$ on trees, and thus, it is not a coincidence that this bound is weaker than that for MM on trees.

For more specialised graph classes, further algorithms have been developed. If the input graph is a growth-bounded graph, Kuhn, Moscibroda, Nieberg and Wattenhofer showed that MIS can be solved in deterministic $\bigO(\log \Delta \cdot \log^* n)$ rounds in \LOCAL \cite{kmnw-disc-05}. Moreover, Barenboim and Elkin developed a $o(\sqrt{\log n})$ rounds algorithm for MIS in the \LOCAL model if the input graph has bounded arboricity \cite{Barenboim2009,PR16}. MIS on bounded-independence graphs requires $\bigO(\log^* n)$ rounds for the \LOCAL model \cite{Schneider2010} and for unit disk graphs, MIS can even be solved in constant rounds by a deterministic \CONGEST algorithm \cite{Molla2019}. 

\section{Technical Overview}\label{sec:tech}
The distinguishing feature of hyperbolic random graphs as compared to general graphs is their underlying hyperbolic geometry that causes two nodes to be adjacent if and only if they are also geometrically close. Thus, our results are based on three geometric mechanisms. First, for the upper bounds, we use short randomised procedures to create geometric separators in the hyperbolic disk; these separators confine every remaining connected component to a narrow angular sector and hence to polylogarithmic size. Second, for the lower bounds, we show that HRGs contain many induced regular trees attached to the giant component by a cut edge, allowing lower bounds from tree-like instances to transfer to HRGs. Third, when geometric coordinates are available, we replace shattering by an explicit tiling of the disk, which lets nodes process well-separated tiles or annuli in parallel. We discuss these three ideas in turn.

\subsection{Efficient Algorithms for MIS and MM}
The algorithms for MIS and MM in \Cref{thm:mainpolylog} follow the same high-level strategy. We first run a constant-round randomised procedure that shatters the graph into connected components of size $\polylog n$. We then solve each remaining component independently using the best available deterministic algorithms. In the \LOCAL model, this gives the $\tilde \bigO(\log^{5/3}\log n)$-round bound via the algorithm of Ghaffari and Grunau~\cite{GG24}, and in the \CONGEST model, it gives the $\tilde \bigO(\log^3\log n)$-round bound via the deterministic \CONGEST algorithms of~\cite{fggkr-soda-23}.

\textit{Why don't we use off-the-shelf algorithms and their analysis?} Essentially all known sublogarithmic-time distributed algorithms for MIS and MM rely on such a shattering framework: an initial randomised step solves most of the graph, leaving only small connected components to be handled deterministically. Existing shattering-based algorithms are not directly suitable for our setting for two reasons. First, their randomised step inherently takes $\Theta(\log \Delta)$ rounds, which is too slow on HRGs, where the maximum degree $\Delta$ is polynomial in $n$. Second, in reality, the standard shattering guarantees do not leave components of $\polylog n$ size; rather, the residual components may have size $\poly(\Delta)\log n$, together with additional structure that can be exploited algorithmically. Relying on this additional structure would lead to weaker final round complexities in our setting. We therefore develop a different geometric argument that establishes the desired shattering phenomenon directly.

\smallskip

The key geometric target of the randomised step is the separator pattern shown in \Cref{fig:full_algo}. We want to remove all vertices in an inner disk (hatched area in \Cref{fig:full_algo}) and, in the remaining outer annulus, remove a collection of neighbourhoods ({green} areas in \Cref{fig:full_algo}) that cut the annulus into narrow angular sectors. If these remaining unsolved sectors have angular width at most $\phi$, for $\phi \approx \polylog(n)/n$, then each sector contains only $\polylog n$ vertices with high probability (blue regions in \Cref{fig:full_algo}b). Moreover, by the choice of the removed regions, any two vertices lying in different sectors are at hyperbolic distance larger than $R$, and hence are non-adjacent. Therefore, once this separator pattern is realised, every remaining connected component is contained in a single narrow sector and has size $\polylog n$.
\begin{figure}
    \centering
    \begin{subfigure}[t]{0.35\textwidth}
        \centering
        \includegraphics[width=\linewidth]{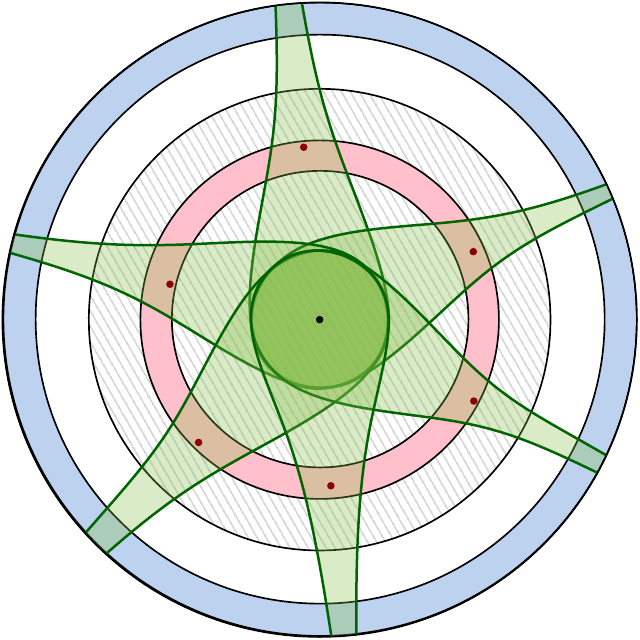}
        \caption{Overview of first two rounds.}
    \end{subfigure}\qquad\qquad
    \begin{subfigure}[t]{0.35\textwidth}
        \centering
        \includegraphics[width=\linewidth]{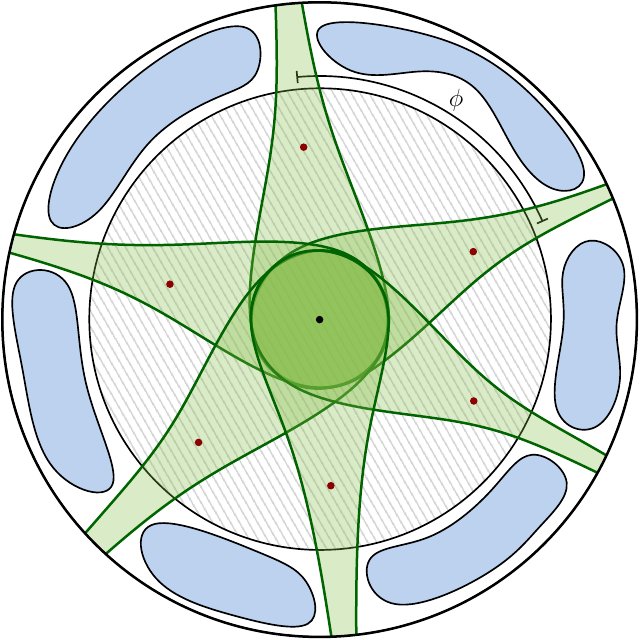}
        \caption{Result after the first two rounds.}
    \end{subfigure}

    \caption{Shattering: (a) Vertices in the blue and red annuli are active in steps 1 and 2, respectively. (b) The hatched area and the green area together form the separator.}
    \label{fig:full_algo}
\end{figure}

The main challenge is to realise this separator pattern while producing a locally valid partial solution for MIS or MM. The way this is done differs substantially between the two problems.

\paragraph{Geometric shattering for MIS.}
The main challenge in designing a fast shattering algorithm, ideally a constant-time algorithm, is dealing with the huge amount of dependencies that occur in algorithms for HRGs. On the positive side for MIS, a selected vertex removes all its neighbours, which makes it well suited for creating geometric separators: if we can make vertices in suitable annuli join the independent set, then their neighbourhoods carve out the desired removed regions, see the {green} disks in \Cref{fig:full_algo}a. 

The first step creates separators in the outer disk. 
We activate vertices in an annulus at radius roughly $\approx R-8\log\log n$ (see red annulus in \Cref{fig:full_algo}a), equivalently, vertices whose degrees lie in an appropriate polylogarithmic range. We then perform one Luby-style step on these active vertices: each active vertex draws a random value, and active vertices that are local winners join the independent set. The neighbourhoods of these winners remove angular intervals in the outer disk. With high probability, every angular interval of width about $\phi \approx \polylog n /n$ contains such a removed neighbourhood. Intuitively, this stems from the fact that active vertices have polylogarithmic degree and thus, given that within an interval of width $\phi$ there $\polylog n$ vertices located, in expectation there will be also a winner vertex. Consequently, the outer annulus is split into angular sectors of width at most $\phi$ (see \Cref{fig:full_algo}b).

The second step removes the inner disk (hatched region in \Cref{fig:full_algo}). Here we activate remaining low-degree vertices close to the boundary (blue annulus in \Cref{fig:full_algo}b). A vertex $u$ in the inner disk has many such active neighbours close to the boundary (that also survived the first step), and for the sake of analysis we can determine $\omega(\log )$ such neighbours for which the events of joining the independent set are independent. Each of these active neighbours has a constant probability of joining the independent set and thereby removing $u$. Since $u$ has $\omega(\log n)$ such independent chances, $u$ remains uncovered with probability $n^{-\omega(1)}$. A union bound then shows that the entire inner disk is removed with high probability.

After these two steps, the remaining vertices lie only in the narrow sectors of the outer annulus. As argued above, different sectors are disconnected, and each sector contains at most $\polylog n$ vertices with high probability. We then solve MIS independently on each remaining component using the deterministic algorithm for the corresponding models of \LOCAL and \CONGEST.

\paragraph{Geometric shattering for MM.}

For maximal matching, the same separator pattern is needed, but it is harder to realise. In MIS, one selected vertex removes its whole neighbourhood. In MM, a matched edge removes only its two endpoints, and a vertex can be matched only once. Thus, removing an entire region requires assigning distinct matching partners without creating conflicts.

We first remove the inner disk (\Cref{fig:full_algo} hatched area). Vertices close to the boundary and of sufficiently small degree are activated (blue annulus in \Cref{fig:full_algo}a). Each active vertex samples one neighbour uniformly at random and proposes to it. We then consider the vertices in the inner disk, which have a sufficiently large degree. With high probability, every such vertex receives at least one proposal from an active neighbour. Each vertex in the inner disk accepts one proposal, and these accepted proposals form matching edges. This matches all vertices in the inner disk and removes them from further consideration.

It remains to create separators in the outer annulus. As in the MIS case, neighbourhoods of vertices in a suitable annulus have the right geometry to separate the outer disk into narrow angular sectors. We activate such vertices (vertices in the red annulus in \Cref{fig:full_algo}a) and let them draw random values and select local maxima as in one iteration of Luby. By choosing the annulus carefully, we ensure that the relevant outer-disk portions of the $2$-hop neighbourhoods of activated nodes are pairwise disjoint. Thus, in the \LOCAL model, each activated vertex can learn its $2$-hop neighbourhood and compute a maximal matching. Removing the matched nodes from the graph disconnects the unsolved parts into small connected components (blue areas in \Cref{fig:full_algo}b). In the \CONGEST model, $2$-hop neighbourhoods cannot be learned efficiently; hence, we use additional structural information to compute the desired maximal matching efficiently. 
\subsection{Lower Bounds for MIS and MM}
The lower bounds are based on showing that HRGs contain hard substructures aka trees. The goal is not merely to find trees as abstract subgraphs, but to find induced regular trees that are attached to the giant component in a controlled way. More precisely, we show that, asymptotically almost surely, the giant component contains polynomially many induced $d$-ary trees (trees with degree $d$) of depth $h$. In particular, we use $d$-ary trees with parameters
\begin{align*}
d \approx {\log\log n}
\qquad\text{and}\qquad
h \approx \log \log n / \log\log\log n.
\end{align*}
Each such tree is attached to the giant component by a cut edge incident to its root, while all other vertices of the tree have no additional edges leaving the tree. Once such trees exist, known lower bounds for MIS and MM on regular trees can be extended and lifted to the stated lower bounds for HRGs.
 \begin{figure}[t]
    \centering \includegraphics{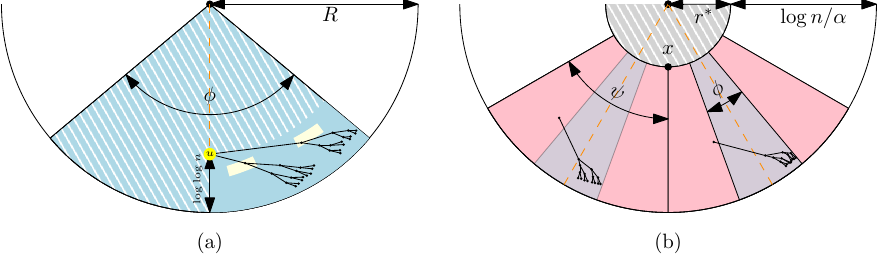} 
    \caption{Excerpt of a hyperbolic disk with trees. (a) Sketch of a nice sector (blue) with angle $\phi$. The hatched area is empty, and the other area contains a $d$-ary tree (here $d = 2$ $h = 3$). Yellow areas are the ``boxes'' of root $u$. (b) Two neighbouring $\Psi$-sectors (red areas) with angle $\psi$. Both embed a $\Phi$-sector with angle $\phi$ (blue areas); both $\Phi$-sectors are nice containing a tree with a "bridge edge" to a vertex in their respective "buffer regions" $\Psi\setminus\Phi$ connecting the tree to the giant component. If a $\Phi$-sector is nice, then any vertex in it has distance larger than $R$ to the point $x$. The grey hatched area contains no vertex a.a.s.}
    \label{fig:d-regular-trees}
  \end{figure}
\paragraph{Geometric construction.}
The trees are located close to the boundary of the hyperbolic disk; see \Cref{fig:d-regular-trees} for an illustration. We consider a narrow angular sector of width
$\phi = o\left({\log n}/{n}\right)$.
Inside such a sector, we prescribe small regions, called boxes (see yellow area in \Cref{fig:d-regular-trees}a for a sketch), where the vertices of the tree should appear. The root is placed at a distance about $\log\log n$ from the boundary (see vertex $u$ in \Cref{fig:d-regular-trees}a). Its children are placed slightly closer to the boundary, at a radial offset of about $2\log d$. This construction is then repeated recursively $h-1$ times: a vertex at level (i) has its $d$ children in boxes at the next radial level.

The boxes are positioned so that three geometric conditions hold. First, each parent is adjacent to every vertex in its child boxes, i.e., within hyperbolic distance at most $R$. Second, vertices in child boxes of the same parent are not adjacent to one another. Third, boxes belonging to different branches are sufficiently separated so that no unintended edges are created. These conditions follow from the hyperbolic distance formula and the choice of the radial and angular spacing between boxes.

Thus, if every prescribed box contains exactly one vertex and the remainder of the sector is empty, then the graph induced by the sector is precisely the desired $d$-ary tree of depth $h$. We call such a sector nice; later, we also ensure that the tree in a nice sector is also connected to the giant component in the desired way.

\paragraph{The small but sufficient probability of a sector being nice.}
We show that a single sector is nice with probability $n^{-o(1)}$: to give intuition for why this is, we first remark that each of the $n' \approx \sqrt{\log n}$ prescribed boxes in a sector has at least constant expected occupancy. So we categorise two different types of boxes, namely

\begin{enumerate}
    \item\label{item:constantbox} boxes where the expected number of vertices is constant, i.e., $\lambda = \Theta(1)$ and;
    \item\label{item:largerbox} boxes where the expected number of vertices is larger than constant, i.e., $\lambda = \omega(1)$.
\end{enumerate}

 For \Cref{item:constantbox}, observe that the number of vertices $X_i$ in the $i$-th box is known to be exactly distributed as a Poisson distribution, i.e., $P(X_i=k)=\lambda^k/k! \cdot e^{-\lambda}$. Here, $\lambda$ is the expected number of vertices in the box and scales with the size of the box. For boxes where  $\lambda =\Theta(1)$ we obtain $P(X_i=1)=\lambda\cdot e^{-\lambda}=p=\Omega(1)$. Thus, since we have at most $n' \approx \sqrt{\log n}$ boxes, the probability that all these boxes have exactly one vertex is given by $p^{n'} = n^{-o(1)}$ as desired.

For \Cref{item:largerbox}, recall that all boxes lie in a sector that has angle $\phi = o(\log n / n)$. Focussing on boxes with expected number of vertices larger than 1, fix any such box $\B_i$, and let $X_i$ be the random variable with which we count the number of vertices in $\B_i$. Using that we have a Poisson distribution, it follows that
$$\Prob{X_i = 1} = \lambda \cdot \exp(-\lambda) \geq \exp(-\lambda) = \Prob{X_i = 0},$$
since $\lambda > 1$. Hence, the probability that a box has exactly one vertex is basically lower bounded by the probability that a box is empty. Now, let $k$ be the number of boxes we consider, and we get from the equation above
$\prod_{i=1}^k \Prob{X_i = 1} \geq \prod_{i=1}^k\Prob{X_i = 0}.$
On the other hand, since all boxes lie within a sector, the probability that all boxes are empty is lower bounded by the probability of the event that the entire sector is empty. Using that the sector has angle $\phi = o(\log n / n)$, the expected number of vertices in a sector is $\lambda_{\text{sector}} \approx n\cdot \phi = o(\log n)$. Thus, the probability of the event that all boxes are empty is $\prod_{i=1}^k\Prob{X_i = 0} \geq \exp{(-\lambda_{\text{sector}})} = n^{-o(1)}$ by another application of a Poisson distribution. Putting everything together, it now follows that $\prod_{i=1}^k \Prob{X_i = 1}  \geq \prod_{i=1}^k\Prob{X_i = 0} = n^{-o(1)}$ as desired.

Of course, finding a nice sector is not enough: the resulting tree must be induced in the full HRG, and it must be attached to the giant component in a controlled way. To ensure this, we embed each candidate sector $\Phi$ into a larger buffer sector $\Psi$. The buffer is chosen large enough so that, conditioned on $\Phi$ being nice and the relevant part of $\Psi \setminus \Phi$ being empty, vertices of the tree have no additional neighbours outside the tree, except for a designated edge from the root to the giant component. We again show that this event holds with probability $n^{-o(1)}$ and we use the standard fact that HRGs contain no vertices too close to the origin asymptotically almost surely to rule out long-range interference from the central part of the disk (see hatched area in \Cref{fig:d-regular-trees}b).

There are polynomially many disjoint buffer sectors $\Psi$, and the events of the $\Psi$-sectors are independent of each other since vertices in two distinct ``nice sectors'' have disjoint neighbourhoods given that there is no vertex in the hatched area of \Cref{fig:d-regular-trees}b. Since each $\Psi$-sector succeeds with probability $n^{-o(1)}$, a Chernoff bound implies that polynomially many of them contain the desired tree structure with high probability. Since we conditioned on the event that the hatched area of \Cref{fig:d-regular-trees}b is empty, an event that holds \aas,  this proves the structural theorem and, via the lower-bound transfer from regular trees, yields \Cref{thm:lowerbound}.

\subsection{Embedding-Aware Algorithms}

We finally outline the algorithms of \Cref{thm:geometric}, where nodes are given access to geometric coordinates. This additional information changes the algorithmic approach. Instead of first producing random separators and then solving small residual components, the algorithms use the geometry directly.
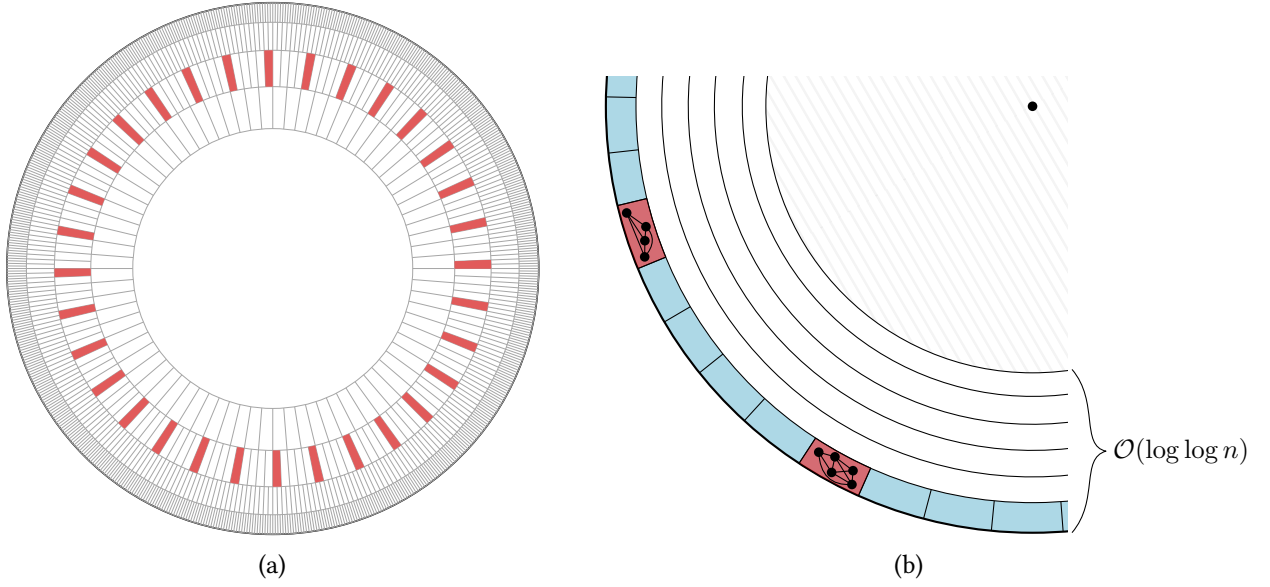
\begin{figure}[t]
\centering

\begin{subfigure}{0.49\linewidth}
\centering
\resizebox{7.2cm}{!}{
\begin{tikzpicture}[scale=5]

\draw[thick] (0,0) circle (0.95);

\def\rA{0.50}
\def\rB{0.65}
\def\rC{0.78}
\def\rD{0.88}
\def\rE{0.95}

\foreach \r in {\rA,\rB,\rC,\rD,\rE}
    \draw[gray!70] (0,0) circle (\r);

\foreach \a in {0,4.5,...,355.5}
    \draw[gray!70] (\a:\rA) -- (\a:\rB);

\foreach \a in {0,2.25,...,357.75}
    \draw[gray!70] (\a:\rB) -- (\a:\rC);

\foreach \a in {0,1.125,...,358.875}
    \draw[gray!70] (\a:\rC) -- (\a:\rD);

\foreach \a in {0,0.5625,...,359.4375}
    \draw[gray!70] (\a:\rD) -- (\a:\rE);

\foreach \k in {0,5,...,160}{
    \pgfmathsetmacro{\a}{2.25*\k}
    \pgfmathsetmacro{\b}{2.25*(\k+1)}

    \fill[customred]
        (\a:\rB)
        arc[start angle=\a,end angle=\b,radius=\rB]
        --
        (\b:\rC)
        arc[start angle=\b,end angle=\a,radius=\rC]
        -- cycle;
}

\end{tikzpicture}
}
\caption{}
\label{fig:a}
\end{subfigure}
\hfill
\begin{subfigure}{0.49\linewidth}
\centering
\includegraphics{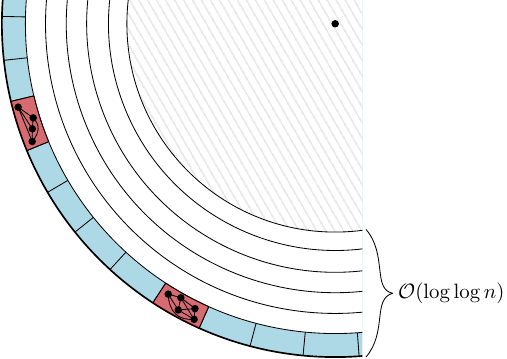}
\caption{}
\label{fig:b}
\end{subfigure}

\caption{Sketch for our tiling. (a) Any pair of points in two different red tiles has a distance larger than $R$. (b) The induced subgraph of a tile is a clique, and there are no edges between the cliques of the red tiles.}
\label{fig:onion}
\end{figure}
\paragraph{Embedding-aware deterministic algorithm for MIS.}
For the MIS algorithm, we first partition the outer part of the hyperbolic disk into tiles arranged in $\bigO(\log\log n)$ radial layers. The tiling is chosen to have two useful properties. First, every tile induces a clique (see red tiles \Cref{fig:onion}b). Second, within each fixed layer, the tiles can be processed according to a constant-size schedule: in each step of the schedule, all currently active tiles are mutually non-adjacent and can therefore be handled in parallel; for a sketch, see red tiles in \Cref{fig:onion}a.

The algorithm is then a simple layer-by-layer sweep. For each of the $\bigO(\log\log n)$ layers, we execute the constant-size schedule for that layer. Whenever a tile is processed, we solve it locally by selecting one still-eligible vertex, if such a vertex exists, and adding it to the independent set. Since a tile is a clique, this is sufficient to resolve all the vertices of the tile. Since each layer requires only $\bigO(1)$ rounds, the total runtime is $\bigO(\log\log n)$.

Correctness for the outer disk follows directly from the sweep: every processed tile is solved when it is considered. It remains to argue that the inner disk (hatched area in \Cref{fig:onion}b) is also dominated. Every inner-disk vertex has neighbours in many processed tiles; we show that by the tiling construction (with high probability over the draw of the HRG), any even adversarially chosen MIS (the algorithm is deterministic) within these tiles still contains at least one neighbour of the vertex. Hence, after the $\bigO(\log\log n)$-round sweep, every vertex is either selected or has a selected neighbour, and the resulting independent set is maximal.

\paragraph{Embedding-aware randomised algorithm for MM.}
For MM, a layer-by-layer sweep would only give an $\bigO(\log\log n)$-round algorithm, which is slower than our target runtime. Moreover, the MIS strategy does not transfer directly: matching one vertex in a tile only removes one neighbour, rather than dominating an entire neighbourhood. We therefore first remove the inner disk (hatched area in \Cref{fig:onion}b) with the same constant-round randomised proposal step as for \Cref{thm:mainpolylog}.

Afterwards, we use a bottom-up approach on the $\bigO(\log\log n)$ remaining outer layers. Initially, we compute a maximal matching inside each single layer in $\bigO(1)$ rounds using a tiling schedule that covers each intra-layer edge. We then repeatedly merge solved layers: when two neighbouring layers have already been solved separately, we may still need to add matching edges between them to obtain maximality for their union. Such merging steps can be performed in parallel for disjoint pairs of layers. The main technical work is to show that each individual merge can be implemented in $\bigO(1)$ rounds as sketched in \Cref{fig:annuli}.

After $t$ merging steps, the algorithm has solved blocks of $2^t$ consecutive layers. Since there are only $\bigO(\log\log n)$ layers, $\bigO(\log\log\log n)$ steps suffice to solve the entire outer disk. Including the constant-round preprocessing of the inner disk, this gives the desired $\bigO(\log\log\log n)$-round algorithm for MM.

\section{Conclusion}
In this work, we investigate classical symmetry-breaking problems through the lens of hyperbolic random graphs. In particular, we resolve an open question posed by \cite{mr-soda-26} by showing that maximal independent set and maximal matching can be computed in $\poly\log\log n$ rounds on hyperbolic random graphs. Thus, when the input graph is drawn from the hyperbolic random graph model, these problems admit substantially faster algorithms than the $\Omega(\sqrt{\log n})$ round complexity known for worst-case instances \cite{khoury2025round}. Our result for the \LOCAL model relies on a constant-round algorithm that shatters an HRG into polylogarithmic-size components. The runtime then follows from the state-of-the-art deterministic algorithm \cite{GG24}; any improvements for the deterministic algorithm in \cite{GG24} would also imply a faster algorithm for hyperbolic random graphs. 

Conversely, we show that MIS and MM are inherently harder than $(\Delta+1)$-colouring on hyperbolic random graphs. While $(\Delta+1)$-colouring can be solved in just 2 rounds \cite{mr-soda-26}, MIS and MM remain substantially more difficult, and we prove a lower bound of $\Omega\left(\frac{\log\log n}{\log\log\log n}\right)$ rounds for both MIS and MM. Our lower bounds rely on new structural insights of hyperbolic random graphs that may be of independent interest. Specifically, we prove that the giant component of a hyperbolic random graph is likely to contain polynomially many logarithmic-size $d$-ary trees for a wide range of different degrees $d$ and height $h$.

Finally, for maximal matching, we demonstrate that access to the underlying geometric coordinates enables an exponentially faster algorithm than implied by this lower bound.

\paragraph{Open Questions.} 
Our work opens up several research questions and  the identified structural insights may also improve the design and analysis of distributed algorithms beyond the HRG model.
\begin{itemize}
    \item How does Luby's algorithm perform on HRGs?
    \item Is there a separation between MIS and MM on HRGs, as is the case on trees?
    \item For embedding-aware graphs, can we overcome the lower bound barrier of \Cref{thm:lowerbound} for MIS?
    \item How can we establish lower bounds in the setting of embedding-aware graphs?
    \item What is the complexity of distributed approximation algorithms for Maximum Matching and Maximum Independent Set on HRGs?
    \item Is the complexity landscape of other symmetry-breaking problems, such as edge colouring or vertex cover, similar to that of MIS/MM or $(\Delta+1)$-colouring on HRGs?
\end{itemize}
\section*{Outline of the rest of the paper} The remainder of this article is structured as follows. In \Cref{sec:prelims}, we give a formal definition for the model of hyperbolic random graphs, and we collect important lemmas we make use of throughout this work. The analyses of our shattering algorithms can be found in \Cref{sec:shattering-algos}. Our findings of $d$-ary trees and the lower bound implications are stated in \Cref{sec:d-regular-trees}. We conclude with our approach for algorithms where the embedding of an HRG is given (\Cref{sec:geometric}). 
\section{Hyperbolic Random Graphs}\label{sec:prelims}
\subparagraph{Hyperbolic plane.}
In the following, we introduce the model of \emph{hyperbolic random graphs} and follow the lines drawn by Papadopoulos~et~al.~\cite{pkbv-info-2010}. Points are represented by their \emph{native representation} in the \emph{hyperbolic plane} $\bH = [0, \infty) \times[0,2\pi)$. Thus, a point $x \in \bH$ is identified by a radial coordinate $r(x)$ and an angular coordinate $\varphi(x)$. We use exclusively the curvature -1 by which the \emph{hyperbolic distance} $\dist(\cdot,\cdot)$ for two points $x,y \in \bH$ is given by
\begin{align}\label{eq:hyperbolic_distance}
    \cosh(\dist(x,y)) = \cosh(r(x))\cosh(r(y)) - \sinh(r(y))\sinh(r(x))\cos(\varphi(x) - \varphi(y)),
\end{align}
and consequently, we use for \bH the topology induced by $\dist$.

For a hyperbolic random graph, we consider a subregion of the hyperbolic plane \Hb: we operate on a disk with radius $R$  denoted by \(\disk = [0,R] \times [0,2\pi)\), centred at the \emph{origin} \((0,0)\). Note that the set of points with distance $r$ to $x\in\disk$ is given by the set $\B_x(r) = \{y\in\disk : \dist(x,y)\leq r\}$, as we restrict the hyperbolic plane to $\disk\subset \Hb$. 

In order to obtain a power-law degree distribution, we introduce the parameter $\alpha \in (1/2, 1)$ (see also \cite{gpp-rhg-12, pkbv-info-2010}). Then, the probability measure $\mu$ on \disk for measurable $\mathcal S\subseteq \disk$ is defined by
\begin{align*}
	\mu(\mathcal S) = \int_S \func(x) \dd x, 
    \qquad \func(x) = \frac{\alpha\sinh(\alpha x)}{2\pi(\cosh(\alpha R) - 1)}, 
\end{align*}
where $\rho$ is the density of $\mu$ with respect to the Lebesgue measure on \disk. 

\subparagraph{(Threshold) hyperbolic random graphs.} Vertices $V$ are identified by their point coordinates in $\disk$: for a subset of vertices that are identified by a set of points $\mathcal{S} \subseteq \disk$, we write $V \cap \mathcal{S}$. Throughout this work, we exclusively work on the \emph{Poissonised version} of \emph{(threshold) hyperbolic random graphs} (see also \cite{Kiwi2024, hrg-spectral, km-slcrhg-19}).  This model allows us to analyse the number of vertices in disjoint areas of the hyperbolic disk independently. Let $n \in \mathbb{Z}^+$ and $N$ be a Poisson random variable with expectation $\E{N} = n$. Then, let $R:= 2\log(n) + C$ for some constant $C \in \mathbb{R}$ and we use an inhomogeneous Poisson point process on $\disk$ where for any $\mathcal{S} \subseteq \disk$ it holds that the number of expected vertices is given by $n\cdot\mu(\mathcal{S}) = \E{|V \cap \mathcal{S}|}$. That is, the intensity function at polar coordinates $(r, \varphi)$ for $0\leq r < R$ is $\intensity(r, \varphi):=n\rho(x)$ and the set of vertices is a random variable $V = \{(r_1,\varphi_1), (r_2,\varphi_2),.., (r_N,\varphi_N)\}$. For a pair of vertices $u,v \in V \cap \disk$, there exists the edge $\{u,v\} \in E$ in a threshold hyperbolic random graph, if and only if it holds $\dist(u,v) \leq R$.

\smallskip

In the following, we collect some results we make use of throughout this work. It will be convenient to characterise the connection of vertices in terms of their \emph{angular distance} $\angulardist{\cdot}{\cdot}$. Since vertices are connected if and only if their distance is at most $R$, we define
\begin{align}\label{eq:angle-func}
    \theta_R(r_1,r_2) = \arccos\left(\frac{\cosh(r_1)\cosh(r_2) - \cosh(R)}
    {\sinh(r_1)\sinh(r_2)}\right),
\end{align}
which denotes by \Cref{eq:hyperbolic_distance} the angle distance for two points with radii $r_1$ and $r_2$ such that their hyperbolic distance is exactly $R$. Thus, two vertices with an angle distance of at most $\theta_R$ are connected. Throughout this work, we use for $\theta_R(\cdot,\cdot)$ the following bounds, which are due to \cite[Corollary 5]{bks-hudg-23}.
\begin{lemma}[Angle distance]\label{lem:max-angle}
Let $x,y \in \disk$ and $1 \leq r(x), r(y) \leq R$ and $r(x)+r(y) \geq R$. Then it holds
$$
\sqrt{e^{R-r(x)-r(y)}} \leq \theta_R(r(x), r(y)) \leq \pi \sqrt{e^{R-r(x)-r(y)}}.
$$
\end{lemma}
As remarked in \cite[Remark 4]{km-slcrhg-19}, the function $\theta_R(\cdot, \cdot)$ is decreasing in both arguments.
\begin{remark}\label{rmk:theta-monotonicity}
    $\theta_R(\cdot, \cdot)$ is strictly decreasing in both arguments.
\end{remark}
Next, we look into the measure for different areas in $\disk$. First, we consider a ball of radius $r$ with the origin as its centre point, as sketched by the red area in \Cref{fig:neighbourhood-and-layers}a (\cite[Lemma 3.2]{gpp-rhg-12}).
\begin{figure}[t]
    \centering \includegraphics[height=0.3\textheight]{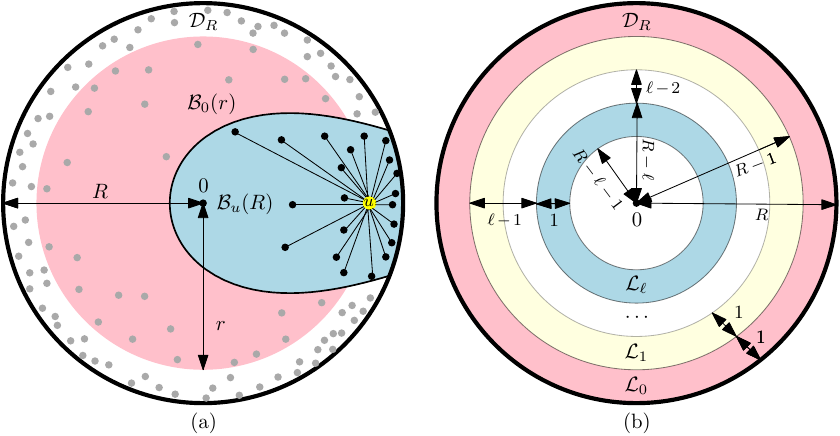}
    \caption{(a) Illustration of the neighbourhood $N(u)$ of a vertex $u$  given by $V \cap \B_u(R)$ (blue region) following the geometry of the hyperbolic disk. The red area is the ball $\B_0(r)$ centred around the origin for some $r$. (b)~Sketch of  layer $\mathcal{L}_0$ (red area), layer $\mathcal{L}_1$ (yellow area) and some layer $\mathcal{L}_\ell$ (blue area) for some $0< \ell<R$.}
    \label{fig:neighbourhood-and-layers}
\end{figure}
\begin{lemma}[Measure of inner disk]\label{lem:measure-inner-disk}
    For any $0 \leq r \leq R$ we have
$\mu\left(\mathcal{B}_0(r)\right) = (1 + o(1))e^{-\alpha(R- r)}~.$
\end{lemma}
We write $N(u) := \{v \in V : \{u,v\} \in E(G)\}$ for the \emph{neighbourhood} of $u$ and $\deg(u) := |N(u)|$ for the \emph{degree} of $u$. Since a vertex $u$ has an edge to every vertex within distance $R$, it is adjacent to all vertices in $\B_u(R) \cap \disk$ (see blue area in \Cref{fig:neighbourhood-and-layers}a). For the expected degree of a vertex, we use the following bounds (\cite[Lemma 7]{mr-soda-26}).
\begin{lemma}[Vertex Degree]\label{lem:vertex-degree}
    Let $u \in V \cap \disk$ be a vertex with radius $r\geq 1$. Then, the expected degree of $u$ in a threshold HRG is 
    $$
    n e^{-r/2}\cdot\frac{(1-o(1))\alpha }{\pi(\alpha-1/2)} \leq \E{\deg(u)} 
        \leq ne^{-r/2}\cdot \frac{ (1+o(1))\alpha}{\alpha-1/2}.
    $$
\end{lemma}
Another sub-area of the disk we use are \emph{layers} as used by Friedrich and Krohmer \cite{fk-dhrg-18}. For $\ell \in [\lfloor R \rfloor]$, a {layer} is defined by $\Layer{\ell}:= \B_0(R-\ell)\setminus \B_0(R-\ell-1)$ (see \Cref{fig:neighbourhood-and-layers}b for a sketch). Using \Cref{lem:measure-inner-disk} and \Cref{lem:vertex-degree}, we obtain the following.
\begin{lemma}[Layer properties] 
\label{lem:layer-properties} Let $\ell \in [\lfloor R \rfloor]$ and $u \in V \cap \Layer{\ell}$. Then
\begin{equation} \label{eq:layer_measure}
    \mu(\Layer{\ell}) =  (1-e^{-\alpha} +o(1)) e^{-\alpha \ell},
\end{equation}
and 
    \begin{equation} \label{eq:layer_expected_degree}
e^{\ell/2} \cdot \frac{ (1-o(1))\alpha e^{C/2}}{\pi(\alpha-1/2)} \leq \E{\deg(u)} \leq e^{\ell/2} \cdot \frac{ (1+o(1))\alpha e^{C/2}}{\alpha-1/2}.
\end{equation}
\end{lemma}
For notational convenience, we sometimes write $V_\ell := V \cap \Layer{\ell}$ for the set of vertices in layer $\ell$.

Finally, the following statement will be useful: it states that if two vertices share a neighbour with larger radius, then they form a connected triangle with that neighbour.
\begin{lemma}\label{lem:triangle-inequality}
    Let $x,y,z \in \disk$ such that $1 \leq r(x) \leq r(y) \leq r(z) - 2\log(2\pi) \leq R -  2\log(2\pi)$. Then, if $\dist(x,z) \leq R$ and $\dist(y,z) \leq R$, it holds that $\dist(x,y) \leq R$.
\end{lemma}
\begin{proof}
By our hypothesis that $\dist(x,z) \leq R$ and $\dist(y,z) \leq R$, it follows for the angular distance between $x$ and $y$ by \Cref{lem:max-angle} that
\begin{align*}
  \angulardist{x}{y} \leq \theta_R(r(x), r(z)) + \theta_R(r(y), r(z)) \leq 2\pi\sqrt{e^{R-r(x)-r(z)}},
\end{align*}
since $r(x) \leq r(y)$ and using \Cref{rmk:theta-monotonicity}. Consequently, using our hypothesis that $r(z) \geq r(y) +  2\log(2\pi)$, it follows 
\begin{align*}
  \angulardist{x}{y} \leq \sqrt{e^{R-r(x)-r(y)}} \leq \theta_R(r(x), r(y)),
\end{align*}
and thus, by \Cref{lem:max-angle} we have $\dist(x,y) \leq R$ as desired.
\end{proof}

\paragraph{Further Notation.}
We write $[k]$ (where $k \in \mathbb{Z}^+$) for the set $\{0,1,\dots, k-1\}$. For the randomness over the graph distribution of a hyperbolic random graph, we use $\Pro{\cdot}$ and  $\Exp{\cdot}$. Conversely, we write $\Prob{\cdot}$ and  $\E{\cdot}$ when we deal with the randomness of an algorithm on a sampled hyperbolic random graph $G$. A \emph{sector} $\Phi \subseteq \disk$ with angle $\phi$ and bisector $\varphi$, is defined by the set of points $\Phi := \{x \in \disk : \varphi - \phi / 2 \leq \varphi(x) \leq \varphi + \phi / 2  \}$. For a vertex $u \in V$, we write $E(u) := \{\{u,v\} \in E : v \in N(u)\}$ for the set of \emph{incident edges} to $u$.
\newcommand{\independentset}[1]{\text{IS}({#1})}
\newcommand{\notindependentset}[1]{\text{OUT}({#1})}
\newcommand{\constant}{\ensuremath{1000}}
\newcommand{\algo}{HRG-Shattering-MIS\xspace}
\newcommand{\algomis}{Embedding-aware-MIS\xspace}
\newcommand{\algos}{HRG-Shattering-MM\xspace}
\newcommand{\algosmm}{Embedding-aware-MM\xspace}

\section{Efficient Algorithms for MIS/MM on HRGs (Theorem \ref{thm:mainpolylog})}\label{sec:shattering-algos}
The goal of this section is to prove the following theorem. 
\maintheorem*
All algorithms follow a shattering-based approach that exploits structural properties of hyperbolic random graphs to identify separators that shatter the graph into small components. These can then be solved independently and in parallel using any deterministic algorithm. The overall running time is determined by the algorithm used to solve the remaining components, as the shattering itself takes constant rounds. The proof for the MIS part can be found in \Cref{sec:cheddar}. The MM part we show in \Cref{sec:mm_shatter}.
\subsection{Efficient Maximal Independent Set Algorithm (MIS part of Theorem~\ref{thm:mainpolylog})}\label{sec:cheddar}
We use Luby's algorithm as a subroutine. In Luby's algorithm, all (active) vertices draw a random ID. If a vertex $u$ has the largest ID among its neighbours $N(u)$, then $u$ is added to the independent set $I$. An iteration of Luby's algorithm has $2$ rounds: first, all active vertices send their random ID to all neighbours. Then any active vertex $u$ informs all its neighbours based on the IDs if $u$ is (a) in the independent set, (b) has a neighbour that is part of the independent set, or (c) neither of the two. For $t \in \mathbb{Z}^+$, we denote by $V_{(t)}$ the set of vertices that are neither part of the independent set nor have a neighbour in the independent set after the $t$-th step of \emph{\algo}, where by \algo we refer to the following process.
\begin{tcolorbox}
\begin{itemize}
    \item \textbf{Step 1:} Activate all vertices of the set $U_{(1)} : = \{v \in V : \left\lfloor\frac{\log^{4}(n)}{\constant}\right\rfloor\ \leq \deg(v) \leq \lceil\constant \log^{4}(n) \rceil\}$. Execute one iteration of Luby's algorithm on active vertices $U_{(1)}$.
    \item \textbf{Step 2:} Activate all vertices of the set $U_{(2)} := \{v \in V_{(1)} : \deg(v) \leq \lceil \log^{3/2}(n)\rceil\}$. Execute one iteration of Luby's algorithm on active vertices $U_{(2)}$.
\end{itemize}
\end{tcolorbox}
We analyse step 1 in \Cref{lem:cheddar1} and step 2 is addressed in \Cref{lem:cheddar2}. We then use the two lemmas to show that {\algo} shatters a hyperbolic random graph into components that are all of size at most $\poly\log n$; see also \Cref{pro:shattering-mis}.

 The following lemma tells us that when we consider the independent set $I$ generated after the first step of \algo, that (1) all vertices of $I$ are contained in an annulus $\mathcal{A}$ with constant thickness and (2) for any sector with a large enough angle, we find a vertex that is part of the independent set $I$; see also \Cref{fig:shattering-mis}a for a sketch.
\begin{figure}[t]
    \centering \includegraphics[height=0.2\textheight]{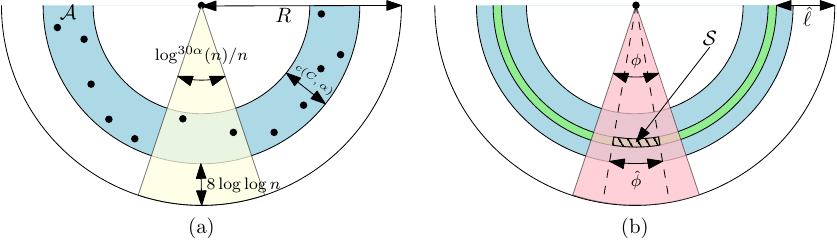}
    \caption{(a) Sketch of \Cref{lem:cheddar1}. All active vertices that join the independent set after the first step are in $\mathcal{A}$ (blue annulus) \wehp, and if a sector is large enough, there is a vertex in the sector that joined the independent set (yellow area). (b)~Sketch of the area $\mathcal{S}$ (hatched region). The sector $\Phi$ (red) is active if area $\mathcal{S}$ contains a vertex that joins the independent set after the first step.}
    \label{fig:shattering-mis}
\end{figure}
\begin{lemma}[Step 1 of \algo] \label{lem:cheddar1}
    Let $G$ be a threshold hyperbolic random graph and consider the set of vertices $I_{(1)}$ that is part of the independent set after step 1 of \algo. Moreover, let $c(C, \alpha)$ be a constant large enough and consider annulus $\mathcal{A} \coloneq \B_0(R - {8\log\log n} + c(C, \alpha))\setminus \B_0(R - {8\log\log n} - c(C, \alpha))$. Then the following holds for $I_{(1)}$ with probability $1 - n^{-\omega(1)}$:
    \begin{enumerate}
        \item\label{item:ring} All vertices contained in the independent set after the first step are in $\mathcal{A}$; $I_{(1)} \subseteq U_{(1)} \subseteq V \cap \mathcal{A}$.

        \item\label{item:active-sector} In any sector with angle $\frac{\log^{30\alpha}n}{n}$ there exists a vertex $v$ such that $v \in I_{(1)}$.
    \end{enumerate}
\end{lemma}
\begin{proof} 
 We prove the first part of the lemma by showing that, \wehp, no vertex outside the annulus has the same degree as an active vertex $U_{(1)}$. Consequently, no such vertex joins the independent set in the first step.
    
    To this end, let $c(C, \alpha) \coloneqq 100 + |C| + 2\log(1/(\alpha - 1/2))$ chosen with hindsight, and consider any vertex $u \in V \cap \B_0(R - {8\log\log n} - c(C, \alpha))$. Then, for $u$ it holds via \Cref{lem:vertex-degree} that 
    \begin{align*}
        \E{\deg(u)} 
        &\geq \frac{(1-o(1))\alpha}{\pi(\alpha - 1/2)} e^{(c(C, \alpha) - C)/2} \log^{4}n\\
        &\geq 10000 \log^{4}n,
    \end{align*}
    by our choice of $c$. Thus, a Chernoff-bound (\Cref{lem:Poisson-Chernoff}) and union bound yield that for any $u \in V \cap \B_0(R - {8\log\log n} - c(C, \alpha))$, that $\deg(u) > \lceil\constant \log^{4} n \rceil$ w.e.h.p.
    
    Analogously, we get for any vertex $u \in V \cap \mathcal{D}_R \setminus \B_0(R - {8\log\log n} + c(C, \alpha))$ that 
     \begin{align*}
        \E{\deg(u)} 
        & \leq \frac{(1+o(1))\alpha}{(\alpha - 1/2)} e^{-(c(C, \alpha) + C)/2} \log^{4}n\\
        & \leq \frac{\log^{4} n}{10000} .
    \end{align*}
    Another combination of Chernoff bound with a union bound then reveals that for any $u \in V \cap \mathcal{D}_R \setminus \B_0(R - {8\log\log n} + c(C, \alpha))$ that $\deg(u) < \lfloor\frac{\log^{4} n}{\constant}\rfloor$ \wehp This concludes the proof for \Cref{item:ring}.
    
    We now turn to the second point of the lemma. Consider a sector $\Phi$ with angle $\phi = \log ^{16\alpha }(n) /n $. W.l.o.g. let $\Phi$ have bisector $0$. Then, let $\hat{\ell} = \lceil 8\log\log n + 2\log(\alpha - 1/2) - C \rceil$, $\hat{\phi} = \log^{10\alpha}(n)/n$ and consider the area $\mathcal{S} := \{x \in \Layer{\hat{\ell}} : 2\pi - \hat{\phi} /2 \leq \varphi(x) \leq \hat{\phi}/2 \}$ (see also \Cref{fig:shattering-mis}b for a sketch). We say $\Phi$ is \emph{active} if $I_{(1)} \cap \mathcal{S} \neq \emptyset$, i.e., there is a vertex in the area $\mathcal{S}$ that joins the independent set after the first step of \algo.
\begin{claim}\label{claim:active-sector}
    A sector $\Phi$ with angle $\phi = \log ^{16\alpha }(n) /n $ is active with probability $\Omega(1/\log^{4} n)$.
\end{claim}
{
\renewcommand{\qedsymbol}{$\blacksquare$} 
\begin{proof}[Proof of claim]
    We lower bound the probability that $\Phi$ is active as follows: first, we show that any vertex in $\mathcal{S}$ is active in the first step \wehp Then, we show that $\mathcal{S}$ contains at least one vertex \wehp The bound then follows using a union bound over the complementary of these two events and since any active vertex $u$ in step 1 of \algo has degree $\deg(u) \in \Theta(\log^{4} n)$.

    By our choice of $\hat{\ell}$ and $\mathcal{S} \subset \Layer{\hat{\ell}}$, we get for any $u \in V \cap \mathcal{S}$ by using \Cref{lem:layer-properties} that
\begin{equation*} 
\frac{ (1-o(1))\alpha}{\pi} \cdot \log^4 n  \leq \E{\deg(u)} \leq (1+o(1))\alpha \cdot \log^4 n.
\end{equation*}
Since $\alpha \in (1/2, 1)$, using Chernoff bounds then yields that for $u \in V \cap \mathcal{S}$ it holds that  $\left\lfloor\frac{\log^{4}(n)}{\constant}\right\rfloor\ \leq \deg(u) \leq \lceil\constant \log^{4}(n) \rceil$ \wehp Hence, $u$ is active for the first step of Luby in \algo and a union bound shows that any vertex in $\mathcal{S}$ is active \wehp

Next we use that $\mathcal{S}$ spans an angle of $\hat{\phi}=\log^{10\alpha}(n)/n$. This in conjunction with $\mathcal{S} \subset \Layer{\hat{\ell}}$ and $\hat{\ell} \in 8\log\log n + \bigO(1)$ then yields via \Cref{lem:layer-properties} that
$$
\mu(\mathcal{S}) = \hat{\phi}/2\pi \cdot \mu(\Layer{\hat{\ell}}) =\Theta(1)\frac{\log^{10\alpha}n}{n} \cdot e^{-8\alpha \log\log n}  \in  \omega\left(\frac{\log n}{n}\right),
$$
since $\alpha > 1/2$. Consequently, by another Chernoff bound, $\mathcal{S}$ is non-empty \wehp The claim now follows since any vertex in by our outline of the proof for Claim \ref{claim:active-sector} since any active vertex $u$ joins $I$ with probability at least $1/(\deg(u) +1) \in \Omega(1/\log^4 n)$.
\end{proof}}

Now, partition the disk $\disk$ into $\lceil 2\pi n \cdot \log ^{-28\alpha } n\rceil =: k$ sectors such that each sector has angle $\psi \in \Theta( \log ^{28\alpha } (n)/n)$. Note that if each such sector contains at least one vertex that joins the independent set $I_{(1)}$, then it also holds for any sector with angle at least $\frac{\log^{30\alpha}n}{n}$ that there exists a vertex that is in $I_{(1)}$ as desired for \Cref{item:active-sector}.

To show that this holds with the desired probabilistic guarantee, fix any sector $\Psi$ with angle $\psi$ and partition $\Psi$ into $k' \in \Theta(\log ^{12\alpha }n)$ sectors, such that each sector has angle $\phi = \log ^{16\alpha }(n) /n $. For the $i$-th sector with angle $\phi$, we write $\Phi_i$ and let $X_i$ be the indicator random variable that is $1$ if the $i$-th sector is active. Moreover, let $X = \sum_{i}^{k'}{X_i}$ so that by linearity of expectation the expected number of active sectors with angle $\phi$ in $\Psi$ is $$\Exp{X} = k' \cdot\Exp{X_i} \in \Omega(\log n^{12\alpha - 4}n) \in \omega(\log n),$$ by applying Claim \ref{claim:active-sector} and using that $\alpha > 1/2$. Next, for any pair of sectors $\Phi_i$ and $\Phi_j$, we show that $X_i$ and $X_j$ are independent \wehp

To see this, recall that $\Phi_i$ is active if the area $\mathcal{S}_i \subset \Phi_i$ with angle $\hat{\phi} = \log^{10\alpha}(n)/n$ contains a vertex that draws the largest ID among active neighbours. W.l.o.g. let $\Phi_i$ have bisector $0$. Moreover, using \Cref{item:ring}, any active vertex $v$ has radius $r(v) \in R - 8\log\log n - \Theta(1)$ \wehp Hence, using \Cref{lem:max-angle} in conjunction with \Cref{rmk:theta-monotonicity}, any active vertex $v$ that is a neighbour of $u \in V \cap \mathcal{S}_i$ has angle 
$$
\varphi(v) \leq \hat{\phi} + \theta_R(r(u) , r(v)) \in \bigO\left(\frac{\log^8 n +\log^{10\alpha} n }{n} \right) \in o(\phi) \text{ \wehp},
$$
where the last step follows since $\psi = \log ^{16\alpha }(n) /n $ and $\alpha > 1/2$. Consequently the event only depends on the randomness in $\Phi_i$ \wehp A union bound over $k'$ sectors with angle $\phi$ shows that this also holds for all $k'$ $\Phi$-sectors in sector $\Psi$ \wehp

Thus, using that for any pair of sectors $\Phi_i, \Phi_j \subset \Psi$ the respective random variables $X_i$ and $X_j$ are independent \wehp and using $\Exp{X} \in \omega(\log n)$, we obtain via a Chernoff bound that $\Psi$ contains at least $1$ vertex that is in $I_{(1)}$ \wehp A union bound over the $k \in o(n)$ $\Psi$-sectors with angle $\psi$ wraps up the proof since every sector with angle $\psi \in \Theta(\log^{28\alpha}(n)/n) \in o(\log^{30\alpha}(n)/n)$ contains a vertex of $I_{(1)}$ \wehp
\end{proof}
Next, we show that after step two of \algo, all vertices of degree larger than $\approx \log^{7/2} n$ are removed by having a neighbour in the independent set $I_{(2)}$ (except for the vertices we included in the independent set after step one). For the area of these vertices, see also the hatched area in \Cref{fig:shattering-mis2}a. 
\begin{figure}[t]
    \centering \includegraphics[height=0.2\textheight]{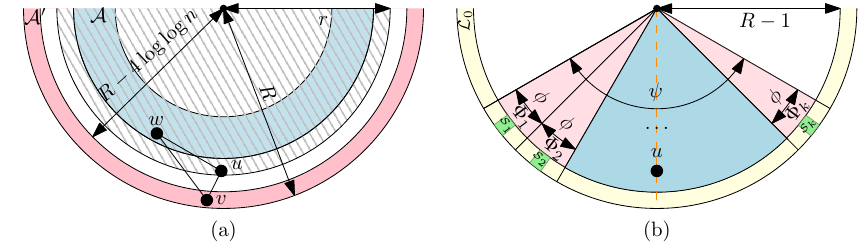}
    \caption{(a) The hatched area is removed after the second step of \algo (\Cref{lem:cheddar2}). Case 1: Vertex $w$ joined the independent set in the first round. If there exist the edges $\{u,v\}$ and $\{v,w\}$, then there exists also the edge $\{u,w\}$. (b) Case 2: Partition of a sector with angle $\psi$ and bisector $\varphi(u)$ into sectors with angle $\phi$. A sector $\Phi_i$ (red area) is active if there is a vertex in $\mathcal{S}_i$ (green area) that joins the independent set.}
    \label{fig:shattering-mis2}
\end{figure}
\begin{lemma}[Step 2 of \algo] \label{lem:cheddar2}
    Let $r \coloneqq R - {7\log\log n}$. Then after the second step of \algo, for any vertex $u \in V_{(1)} \cap \mathcal{B}_0(r)$ it holds with probability $1 - n^{-\omega(1)}$ that $u \not\in V_{(2)}$. 
\end{lemma}
\begin{proof}
Consider the annulus $\mathcal{A}' := \disk\setminus \B_0(R - 4 \log\log n - C)$. Then, by \Cref{lem:vertex-degree}, any vertex $v \in V \cap (\disk\setminus\mathcal{A}' )$ has expected degree $\E{\deg(v)} \in \Omega(\log^2 n)$. Using a Chernoff- and a union bound, all vertices outside the annulus $\mathcal{A}' $ have a degree larger than $\lceil \log^{3/2}(n)\rceil$ \wehp Hence, all vertices that are active in the second step of \algo are within the annulus $\mathcal{A}' $, i.e., $U_{(2)} \subseteq V_{(1)} \cap \mathcal{A}' $ \wehp

Now, using $r= R- \log\log n$ as defined in our lemma statement, fix a vertex $u \in V_{(1)} \cap \B_0(r)$, i.e., a vertex in the inner disk that was not removed in the first round of \algo. We prove our desired statement by the following case distinction: (1) there exists at least one vertex $v$ in the neighbourhood of $u$, such that $v$ lies in the annulus $\mathcal{A}' $ and $v$ was removed after the first step of \algo, i.e., $(N(u) \cap \mathcal{A}' )\setminus (V \setminus V_{(1)}) \neq \emptyset$ (see also \Cref{fig:shattering-mis2}a), and (2) non of the neighbours of $u$, that lie in annulus $\mathcal{A}' $ was removed after the first step of \algo, i.e., $(N(u) \cap \mathcal{A}' )\setminus (V \setminus V_{(1)}) = \emptyset$.

\smallskip

\textbf{Case 1}~[$(N(u) \cap \mathcal{A}' )\setminus (V \setminus V_{(1)}) \neq \emptyset$]: Consider a vertex $v \in (N(u) \cap \mathcal{A}' )\setminus (V \setminus V_{(1)})$. Since $v \notin V_{(1)}$, either (a) $v$ is in the independent set after the first step, $v \in I_{(1)}$, or (b) $v$ has a neighbour $w \in I_{(1)}$ that is in the independent set after the first step. Using \Cref{lem:cheddar1}, no vertex in the set $I_{(1)}$ is in the annulus $\mathcal{A}'$ \wehp and thus, case (a) does not occur \wehp Hence, only case (b) remains  (see \Cref{fig:shattering-mis2}a for a sketch). Using again \Cref{lem:cheddar1}, it follows that vertex $w \in N(v) \cap I_{(1)}$ has radius $r(w) \leq R - 8\log\log n + c$ \wehp Then, since $v \in V \cap \mathcal{A}'$, it holds $r(v) \geq R - 4\log\log n$ - C. Thus, using \Cref{lem:triangle-inequality}, $u \in N(w)$ since $r(u) \leq R - 7\log\log n$. It follows that $u \not\in V_{(1)}$ \wehp since $w \in I_{(1)}$; $u$ was removed after the first step of \algo, and thus, we conclude by (a) and (b) that case~1 does not occur \wehp

\smallskip

\textbf{Case 2}~[$(N(u) \cap \mathcal{A})\setminus (V \setminus V_{(1)}) = \emptyset$]: Consider the sector $\Psi$ with angle $\psi = e^{C/2} \cdot \log^{7/2} n /n$ and bisector $\varphi(u)$. Note that, since $r(u) \leq R - 7\log\log n$, we obtain via \Cref{lem:max-angle} in conjunction with \Cref{rmk:theta-monotonicity} that $\Psi \subset \B_u(R)$, and thus, if a vertex $v \in V \cap \Psi$ joins the independent set, $u$ is removed since $v \in N(u)$. We partition $\Psi$ into $k:= e^{C}\cdot\log^{3/2}n / 100$ sectors, such that for $i \in [k]$ each sector $\Phi_i \subset \Psi$ has angle $\phi = 100\cdot e^{-C/2}\log^2 n /n$. Let $\varphi_i$ be the bisector of $\Phi_i$ and let $\mathcal{S}_i := \{x \in \Phi_i \cap \Layer{0} : \varphi_i - 1/n \leq \varphi(x) \leq \varphi_i + 1/n \}$, i.e., a sub-area of sector $\Phi_i$ in layer $\Layer{0}$, with the same bisector as $\Phi_i$ that spans an angle of $2/n$ (see also green area in \Cref{fig:shattering-mis2}b). We say that $\Phi_i$ is \emph{active} if there exists a vertex $v \in V \cap \mathcal{S}_i$ such that $v$ joins the independent set in the second step of \algo, i.e., $I_{(2)} \cap \mathcal{S}_i \neq \emptyset$.
\begin{claim}\label{claim:active-sector2}
    A sector $\Phi_i$ is active with non-vanishing probability.
\end{claim}
{
\renewcommand{\qedsymbol}{$\blacksquare$} 
\begin{proof}[Proof of claim]
Recall that due to case~2, we condition on the event that no neighbour $v \in N(u) \cap \mathcal{A}'$ has a neighbour $w \in N(w) \cap I_{(1)}$ in the independent set after the first step of \algo. 

Thus, using that $|V \cap \mathcal{S}_i|$ follows a Poisson distribution and $$\mu(\mathcal{S}_i) =  \frac{\mu(\Layer{0})}{\pi n} \in \Omega(1/n),$$ by \Cref{lem:layer-properties} and $\mathcal{S}_i$ spanning an angle of $2/n$, it holds $|v \in \mathcal{S}_i \cap V_{(1)}| = 1$ with non vanishing probability. We write for this event $\mathcal{E}_i := \{|V \cap \mathcal{S}_i| =1\}$.

Moreover, since $v \in V \cap \Layer{0}$ it holds $\E{\deg(v)} \in \bigO(1)$ by \Cref{lem:layer-properties}. Clearly, the number of active neighbours of $v$ is upper bounded by $\deg(v)$. Thus, again it follows from a Poisson distribution that the number of active neighbours of $v$ is $\bigO(1)$ with non-vanishing probability when conditioned on event $\mathcal{E}_i$ (i.e., $\mathcal{S}_i$ is empty except for $v$).

The claim then follows since, given that $v$ has at most constant active neighbours, $v$ joins the independent set with non-vanishing probability.
\end{proof}}
To finish the proof, let $X_i$ be the indicator random variable that is $1$ if sector $\Phi_i$ is active and let $X = \sum_{i=0}^{k-1}X_i$ be the number of ``active $\Phi$-sectors in sector $\Psi$''. By linearity of expectation and Claim \ref{claim:active-sector2} in conjunction with $k \in \Omega(\log n^{3/2})$ it holds $\Exp{X} \in \omega(\log n)$. Thus, if we can show for any $X_i$ and $X_j$ that they are independent, our desired statement follows from a Chernoff bound in conjunction with a union bound over the set of vertices $V \cap \B_0(R - 7\log\log n)$. This is since an active $\Phi$-sector implies that a neighbour $v \in N(u)$ joins the independent set which removes $u$.

We show that this independence holds for any two indicator random variables $X_i$ and $X_j$ \wehp Recall that $X_i$ is $1$ if and only if there is a vertex $v \in V \cap \mathcal{S}_i$ that joins the independent set. Moreover, recall that \wehp only the area $\mathcal{A}' = \disk \setminus \B_0(R-4\log\log n -C)$ contain active vertices. W.l.o.g., let sector $\Phi_i$ have bisector $0$. Thus, considering $v \in V \cap \mathcal{S}_i$ with radius $r(v) \geq R - 1$, for any active neighbour $w \in N(v)$ we have $r(w) \geq R - 4\log\log n$ \wehp, and it holds
$$
\varphi(w) \leq 1/n + \theta_R(r(v), r(w)) \leq \frac{\pi \cdot e^{(1-C)/2} \log^2 n}{n} < \phi \text{ \wehp,}
$$
using \Cref{lem:max-angle} and \Cref{rmk:theta-monotonicity} in conjunction with $\phi = 100 \cdot e^{-C/2} \log ^2 n /n$ and $\mathcal{S}_i$ spanning at most an angle of $1/n$. Thus, all active neighbours of the vertex $v \in V \cap \mathcal{S}_i$ are contained in sector $\Phi_i$ \wehp, yielding the independence \wehp as desired. A Chernoff bound yields $X \in \omega(\log n)$ \wehp such that $u$ is removed \wehp and a union bound over all the set of vertices $|V \cap \B_0(R - 7\log\log n)| \in o(n)$ \wehp wraps up the proof.
\end{proof}
We now have all ingredients to show that \algo shatters a hyperbolic random graph into components of size at most $\poly\log n$.
\begin{proposition}[MIS Shattering]\label{pro:shattering-mis}
    For a threshold hyperbolic random graph $G$, it holds with probability $1 - n^{-\omega(1)}$ that, after the second step of \algo, the largest connected component in $G[V_{(2)}]$ is of size at most $\log^{\bigO(1)}(n)$.\footnote{By $V_{(2)}$ we refer to the set of vertices that are neither in the independent set nor have a neighbour in the independent set after step 2 of \algo.} The round complexity of \algo for $\CONGEST$ is $\bigO(1)$.
\end{proposition}
\begin{proof}
The runtime $\bigO(1)$ directly follows from the fact that one iteration of Luby's algorithm requires $\bigO(1)$ rounds in \CONGEST. Now, consider the set of vertices $I_{(1)}$, i.e., the set of vertices in the independent set after the first step of \algo. Let $k := |I_{(1)}|$ and consider for $I_{(1)}$ an ascending ordering by angular coordinates $v_0, v_1,\dots, v_{k-1}$. For $i \in [k]$, let $\Phi_i$ be the sector defined by $\{x \in \disk : \varphi (v_{i}) \leq \varphi(x) \leq \varphi(v_{i+1})\}$ where $v_k = v_0$. Note that by \Cref{lem:cheddar1}, for $i \in [k]$ it holds for any consecutive pair $v_i, v_{i+1}$ that $\angulardist{v_{i}}{v_{i+1}} \leq 2 \log^{30\alpha}(n)/n$ \wehp and thus, the angle for any sector $\Phi_i$ is at most $\bigO(\log^{30\alpha}(n)/n)$ \wehp Thus, it holds for any $i \in [k]$ that $\mu(\Phi_i) \in \bigO(\log^{30\alpha}(n)/n)$. Via Chernoff bound for $|V \cap \Phi_i|$ and a union bound over $k \in o(n)$ sectors, we then obtain for all $i \in [k]$ that $|V \cap \Phi_i| \in \bigO(\log^{30\alpha}n)$ \wehp Clearly, this implies that $|V_{(2)} \cap \Phi_i| \in \bigO(\log^{30\alpha}n)$ \wehp Thus, if we can show for any pair $i \neq j \in [k]$ that for all pairs $u \in V_{(2)} \cap \Phi_i$ and $v \in V_{(2)} \cap \Phi_j$ it holds $\dist(u,v) > R$ \wehp, this finishes the proof.

 To this end, fix any pair of sectors $\Phi_i$ and $\Phi_{j}$, and w.l.o.g. let $j = i+1$ since non-adjacent sectors are separated by an even larger angular interval. Consider any pair $u \in V_{(2)} \cap \Phi_i$ and $v \in V_{(2)} \cap \Phi_{i+1}$. Using \Cref{lem:cheddar2}, it holds \wehp for both vertices $u, v \in V_{(2)}$ that $r(u), r(v) \geq R - 7 \log\log n$. Moreover, let $s = v_{i+1}$ be the vertex that lies on the intersecting ray of $\Phi_i$ and $\Phi_j$. Then, by \Cref{lem:max-angle} in conjunction with \Cref{rmk:theta-monotonicity}, it holds for the angular distance between $u$ and $v$ that
$$
\angulardist{u}{v} \geq \theta_R(r(v), r(s)) + \theta_R(r(u), r(s)) > \theta_R(r(u), r(v)) \text{ \wehp,}
$$
since $u, v \not\in N(s)$ and $r(v), r(u) \in r(s) + \Omega(\log \log n)$ \wehp by \Cref{lem:cheddar1}. Thus, for all pairs $u \in V_{(2)} \cap \Phi_i$ and $v \in V_{(2)} \cap \Phi_j$ it holds $\dist(u,v) > R$ \wehp; taking a union bound over all pairs of sectors, the desired result follows: every connected component of $G[V_{(2)}]$ is contained in a single sector $\Phi_i$. Since each sector contains at most $\bigO(\log
^{30\alpha}n)$ vertices \wehp, every connected component has poly-logarithmic size. This proves the proposition.
\end{proof}
\begin{proof}[Proof of \Cref{thm:mainpolylog} (MIS part)]
    \textbf{\LOCAL:} For \LOCAL MIS, we obtain a shattering where each connected component is of size at most $\polylog n$ \wehp after $\bigO(1)$ using \Cref{pro:shattering-mis}. Applying in parallel for each component \cite[Theorem 3.1]{GG24}, we obtain an MIS after $\tilde{\bigO}(\log^{5/3}\log n)$ rounds \wehp 

\smallskip

    \textbf{\CONGEST:} For \CONGEST MIS, \Cref{pro:shattering-mis} implies that after $\bigO(1)$ rounds each connected component induced by undecided nodes is of at most poly-logarithmic size \wehp  Consequently, applying in parallel for each connected component \cite[Theorem 1.1]{fggkr-soda-23}, we obtain an MIS after $\tilde{\bigO}(\log^{3}\log n)$ rounds \wehp 
\end{proof}
\subsection{Efficient Maximal Matching Algorithm (MM part of Theorem~\ref{thm:mainpolylog})}\label{sec:mm_shatter}
In this section, we turn to our shattering algorithm, which we use to obtain a maximal matching $M \subseteq E$ in $\polylog\log n$ rounds. {For $t \in \mathbb{Z}^+$, let $M_t$ be the set of edges included in our matching after step $t$ of our algorithm \algos (see description below). We denote by $V_{(t)}$ the set of \emph{unmatched vertices} after step $t$: a vertex is unmatched if it has at least one incident edge that can be included without violating the matching condition after step $t$ of \emph{\algos},} i.e., $V_{(t)} :=\{u \in V : \exists v \in N(u) \text{ such that } E(v) \cap M_t = \emptyset \text{ and } E(u) \cap M_t = \emptyset\}$. The following is a brief description of our algorithm, which we refer to as \algos:
\begin{tcolorbox}
\begin{itemize}
       \item \textbf{Step 1 (Inner disk removal):} Activate all vertices of the set \(U_{(1)} := \{v \in V_{(1)} : \deg(v) \leq \lceil \log^{3/2}(n)\rceil\}\). In parallel, each active vertex \(u \in U_{(1)}\) marks one incident edge uniformly at random. Then, in parallel, each vertex $u \in V$ chooses a marked edge $\{u,v\} \in E$ uniform at random from its incident marked edges, and we add $\{u,v\}$ to the matching $M$.
       
  \item \textbf{Step 2 (Separating the outer disk):} Activate all vertices of the set $U_{(2)} : = \{v \in V : \left\lfloor\frac{\log^{4}(n)}{\constant}\right\rfloor\ \leq \deg(v) \leq \lceil\constant \log^{4}(n) \rceil\}$. In parallel, each vertex $u \in U_{(2)}$ draws a random ID, and we add $u$ to the set $S$ if $u$ has the largest ID among its neighbours in $U_{(2)}$.
 
If \LOCAL: In parallel, each vertex $s \in S$ collects the adjacency matrix of the induced subgraph $G_s := G[N_2(s) \cap V_{(1)}]$.\footnote{We write $N_2(u)$ for the two-hop neighbourhood $N_2(u) := \{v \in V : d_G(u,v) \leq 2\}$.} In parallel, each vertex $s \in S$ computes a maximal matching $M(s) \subseteq E(G_s)$ and we add $M(s)$ to the matching $M$.
 
If \CONGEST: In parallel, each vertex $s \in S$ activates the set of vertices $U(s) := N_2(s) \cap V_{(1)}$. In parallel, for each $s \in S$, compute a maximal matching $M(s)$ for $G[U(s)]$ via \cite[Theorem 1.2]{fischer2020improved} and we add $M(s)$ to the matching $M$.
\end{itemize}
\end{tcolorbox}
The rest of this section is dedicated to proving that \algos shatters an HRG into polylogarithmic connected components; see also \Cref{pro:shattering-mm}. We start by showing in \Cref{lem:match-inner-disk} that after step 1, in constant rounds, all vertices of the inner disk $\B_0(R - 6\log\log n)$ are matched \wehp We then continue by analysing step 2 and the properties it yields for set $S$ in \Cref{lem:cheddar20}.
\begin{lemma}[Step 1 of \algos] \label{lem:match-inner-disk}
    Let $G$ be a threshold hyperbolic random graph. Then in \CONGEST after step 1 of \algos, the set of vertices $V \cap \B_0(R- 6\log\log n)$ is matched in constant rounds \wehp
\end{lemma}
\begin{proof}
    In the step 1 of \algos, we first activate all vertices with degree at most $\lceil\log^{3/2} n\rceil$. In parallel, all active vertices mark one incident edge uniformly at random.
    
    Fix a vertex $u \in V \cap \B_0(R- 6\log\log n)$ and let $X_u$ be the random variable with which we count the number of marked edges incident to $u$. Note that if we can show that $X_u \geq 1$ for all $u  \in V \cap \B_0(R- 6\log\log n)$ with probability $1 - n^{-\omega(1)}$, then this proves the claim since then any vertex in $u \in V \cap \B_0(R- 6\log\log n)$ can simply select one of the marked edges so that $u$ is matched.
    
To show this, consider the set of vertices which form an $e$ edge with $u$ such that $e$ is potentially marked. Let us write $N'(u) := N(u) \cap \{v \in V : \deg(v) \leq \lceil\log^{3/2} n\rceil\}$ for this set and note that for $v \in N'(u)$, the edge $\{u,v\}$ is marked with probability $1/\deg(v)$, as $v$ samples the edge uniform at random among its neighbours. Hence, it holds
\begin{align}\label{eq:prob-bad-event}
    \Pro{X_u = 0} = \prod_{v \in N'(u)}(1-1/\deg(v)).
\end{align}
In the following, we lower bound $|N'(u)|$. Since for any vertex $v \in V \cap \Layer{0}$, it holds via \Cref{lem:vertex-degree} that $\E{\deg(v)} \in \bigO(1)$, such that a Chernoff-bound with a subsequent union bound over at most $\bigO(n)$ vertices yields that $\deg(v) \leq \log^{3/2}n$ and all vertices in layer $\Layer{0}$ are active \wehp and participate in our marking process. Using this bound and that $r(u) \leq R - 6\log\log n$ in conjunction with the fact that $\theta_R(\cdot, \cdot)$ is a monotonic decreasing function in both arguments (\Cref{rmk:theta-monotonicity}), it holds via law of total expectation that
\begin{align*}
    \E{|N'(u)|} \geq (1-o(1)) n\cdot \theta_R(R - 6\log\log n, R)\cdot \mu(\Layer{0}) \in \Omega(\log^3(n)),
\end{align*}
where in the last step we used \Cref{eq:layer_measure} and  \Cref{lem:max-angle}. Another combination of Chernoff and union bounds reveals that for any $u \in V \cap \B_0(R- 6\log\log n)$ it holds $|N'(u)| \in \Omega(\log^3 n)$ with probability $1 - n^{-\omega(1)}$.

Let $\mathcal{E}$ be the event that for all $v \in V \cap \Layer{0}$ it holds $\deg(v) \in \bigO(\log^{3/2}n)$ and that for all $u \in V \cap \B_0(R- 6\log\log n)$ it holds $|N'(u)| \in \Omega(\log^3(n))$. By previous discussion it follows via union bound that $\Pro{\mathcal{E}} \in 1 - n^{-\omega(1)}$. Note that for our fixed $u \in V \cap \B_0(R- 6\log\log n)$ it holds $\Pro{X_u = 0 |\mathcal{E}} \in e^{-\Omega(\log^{3/2}n)}$ by \Cref{eq:prob-bad-event}. Hence, we obtain using a union bound
\begin{align*}
 \Pro{\not\exists u  \in V \cap \B_0(R- 6\log\log n) : X_u = 0} \geq \Pro{\mathcal{E}}\cdot(1 - \sum_{u \in V \cap \B_0(R- 6\log\log n) }\Pro{X_u = 0|\mathcal{E}}) \in 1-n^{-\omega(1)},
\end{align*}
since $|V \cap \B_0(R- 6\log\log n)| \in \bigO(n)$ \wehp That is, we showed that $X_u \geq 1$ for all $u  \in V \cap \B_0(R- 6\log\log n)$ with probability $1 - n^{-\omega(1)}$ as desired. This finishes the proofs since the desired runtime $\bigO(1)$ is immediate.
\end{proof}
The following is an analogue statement of \Cref{lem:cheddar1}. Additionally, we show that any pair of vertices in the set $S$ has a minimum angular distance (\Cref{item:min_angle}).
\begin{lemma}[Step 2 of \algos] \label{lem:cheddar20}
    Let $G$ be a threshold hyperbolic random graph and consider the set of vertices $S$ in step 2 of \algos. Moreover, let $c(C, \alpha)$ be a constant large enough and consider annulus $\mathcal{A} \coloneq \B_0(R - {8\log\log n} + c(C, \alpha))\setminus \B_0(R - {8\log\log n} - c(C, \alpha))$. Then the following holds for $S$ with probability $1 - n^{-\omega(1)}$:

    \begin{enumerate}
        \item\label{item:ring2} All vertices contained in the set $S$ in the second step of \algos are in $\mathcal{A}$; $S \subseteq U_{(2)} \subseteq V \cap \mathcal{A}$.

        \item\label{item:active-sector2} In any sector with angle $\frac{\log^{30\alpha}n}{n}$ there exists a vertex $v$ such that $v \in S$.
        \item\label{item:min_angle} For any pair of vertices $u,v \in S$ it holds that the angular distance is $\angulardist{u}{v} \in \Omega(\log^8 n / n)$ and $\{u,v\}\not\in E.$
       
    \end{enumerate}
\end{lemma}
\begin{proof}
    Items~\ref{item:ring2} and~\ref{item:active-sector2} follow directly from
\Cref{lem:cheddar1}, since step~2 of \algos activates exactly the same set of
vertices as step~1 of \algos and applies the same local selection rule.
Therefore, the arguments used in the proof of \Cref{lem:cheddar1} carry over
verbatim: the set $S$ is equivalent to $I_{(1)}$.

To show that \Cref{item:min_angle} holds, suppose for contradiction that $\{u,v\}\in E$. Since both $u$ and $v$ belong
to $S$, each has the largest ID in its closed neighbourhood. However, as
$u\in N(v)$ and $v\in N(u)$, at most one of the two vertices can have the
largest ID among the vertices of its closed neighbourhood, a contradiction.
Hence $u$ and $v$ are not adjacent. Then, using that $r(u), r(v) \in R - 8 \log\log n +\bigO(1)$ \wehp by \Cref{item:ring2}, it holds \wehp using \Cref{lem:max-angle} in conjunction with \Cref{rmk:theta-monotonicity} that 
    $$
    \angulardist{u}{v} > \theta_R(r(u),r(v)) \in  \Omega(\log^8 n / n).
    $$
    Hence the claimed lower bound on the angular distance holds for every pair of vertices in $S$ whenever the event of
Item~\ref{item:ring2} occurs. Since this event holds with probability
$1-n^{-\omega(1)}$, the proof is complete.
\end{proof}
We now show that after the second step of \algos, a hyperbolic random graph is shattered into components of size at most $\poly\log n$ for maximal matching.
\begin{proposition}[MM Shattering]\label{pro:shattering-mm}
    For a threshold hyperbolic random graph $G$, it holds with probability $1 - n^{-\omega(1)}$ that, after the second step of \algos, the largest connected component in $G[V_{(2)}]$ is of size at most $\log^{\bigO(1)}(n)$.\footnote{By $V_{(2)}$ we refer to the set of unmatched vertices after step 2 of \algos.} The round complexity of \algos is $\bigO(1)$ for $\LOCAL$ and $\bigO(\log^3\log n)$ for $\CONGEST$.
\end{proposition}
\begin{proof}
\textbf{\LOCAL:}  We start by proving the \LOCAL part. Consider any pair of vertices $s, s' \in S$ in the second step of \algos. We show for any of the two sets of vertices $N_2(s) \cap V_{(1)} =: U(s) $ and  $N_2(s') \cap V_{(1)} =: U(s')$, that for any pair $v \in U(s)$ and $v' \in U(s')$ it holds $\{v,v'\} \not \in E$ \wehp This ensures that for any pair $s, s' \in S$,  there are no conflicts when $s$ and $s'$ compute in parallel $M(s)$ and $M(s')$ in the \LOCAL computation of the second step in \algos. To see this, note first that $\{s,s'\} \not\in E$ \wehp holds immediately by \Cref{lem:cheddar20}. Thus, it suffices to show that for any $v \in U(s)$ and $v' \in U(s')\setminus\{s'\}$ that $\{v,v'\} \not \in E$ \wehp Now, recall that by \Cref{lem:match-inner-disk} it holds \wehp for any $v,u \in V_{(1)}$ that $r(v),r(u) \geq R - 6\log\log n$. Moreover, by \Cref{lem:cheddar20} it holds for some constant $c$ that, $r(s) \geq R - 8\log\log n -c$ \wehp Thus, for any vertices $v, u \in U(s)$ it holds \wehp that 
  \begin{align}\label{eq:two-hop-angle}
      \angulardist{s}{v} \leq \theta_R(r(s), r(u)) + \theta_R(r(v), r(u)) < \Theta(1) \log^{7} n / n, 
  \end{align} 
by \Cref{lem:max-angle}. Moreover, using \Cref{lem:cheddar20}~\cref{item:min_angle}, it holds \wehp for the pair $s, s' \in S$ that
\begin{align}\label{eq:separator-angle}
     \angulardist{s}{s'}  \in  \Omega(\log^8 n / n).
\end{align}
Consequently, we obtain \wehp for any pair $v \in U(s)$ and $v' \in U(s')\setminus\{s'\}$ by combining \Cref{eq:two-hop-angle} and \Cref{eq:separator-angle} that
\begin{align}\label{eq:two-hop-intersection}
    \angulardist{v}{v'} = \angulardist{s}{s'} - \angulardist{s}{v} - \angulardist{s'}{v'} \in \Omega(\log^8 n / n) \in \omega(\theta_R(r(v), r(v'))) \text{ \wehp,}
\end{align}
using \Cref{lem:max-angle} in conjunction with $r(v') \geq R - 6\log\log n$ \wehp by \Cref{lem:match-inner-disk} and $r(v) \geq R - 8\log\log n -c$ \wehp by \Cref{lem:cheddar20}. It follows via \Cref{eq:two-hop-intersection} that no two-hop neighbourhoods $U(s') \cap U(s)$ share any edge \wehp as desired, which, by union bound then also holds for any pair $s, s' \in S$ \wehp We conclude that \wehp there are no conflicts for any matchings $M(s)$ and $M(s')$ and it is left to show that (a) we obtain the desired shattering and (b) the computation in \algos requires $\bigO(1)$ rounds of the \LOCAL model.

(a)~The shattering property follows by the same argument as in the proof of
\Cref{pro:shattering-mis} after replacing the set $I_{(1)}$ with $S$.
Indeed, by \Cref{lem:match-inner-disk}, every vertex in
$V \cap \B_0(R-6\log\log n)$ is already matched after Step~1 of \algos \wehp
Moreover, for every $s \in S$, the algorithm computes a maximal matching on
the induced subgraph $G[U(s)]$, where $U(s)=N_2(s)\cap V_{(1)}$.
Since $N(s)\subseteq U(s)$, every neighbour of $s$ is either matched itself
or matched to another vertex in $U(s)$. Consequently,
$N(s)\cap V_{(2)}=\emptyset$ for every $s\in S$.
Hence every unmatched vertex lies in the same outer annulus as in
\Cref{pro:shattering-mis}, and the remainder of the shattering argument
applies verbatim.

(b)~By \Cref{lem:match-inner-disk} step 1 of \algos requires constant rounds. The communication in step 2 of \algos also requires constant rounds so that active vertices are informed if they participate in the set $S$. Moreover, each vertex $s \in S$ in parallel collects its two-hop neighbourhood of unmatched vertices in constant rounds and then, in further constant rounds, informs the vertices of the set $U(s)$ about the matching $M(s) \subseteq E(G_2)$. Hence, at most $\bigO(1)$ communication rounds are required.

\smallskip

\textbf{\CONGEST:} The correctness of the \CONGEST part follows from the same arguments as for the \LOCAL model. For the round complexity, we consider the following: first, for all $s \in S$, in parallel, each induced subgraph $G(U(s))$ can be informed by $s$ to be active within $2$ rounds in \CONGEST. Next, fix a vertex $s \in S$. Since there exists a constant $c$ such that $r(s) \geq R - 8\log\log n -c$ by \Cref{lem:cheddar20}, it follows by \Cref{lem:vertex-degree} that $\E{\deg(s)} \in \Theta(\log^4 n)$. Thus, by a Chernoff and union bound it holds for any $s \in S$ that $\deg(s) \in \Theta(\log^4 n)$ \wehp Moreover, it holds for any $v \in V_{(1)}$ that $r(v) \geq R - 6\log\log n$ \wehp by \Cref{lem:match-inner-disk}. Subsequently, using \Cref{lem:vertex-degree} with a Chernoff and a union bound, we have for each $v \in V_{(1)}$ that $\deg(v) \in \bigO(\log^3 n)$. Hence, for all $U(s)$ it holds $|U(s)|\in \bigO(\log^7 n)$ \wehp Then, recall that for any pair of vertices $s, s' \in S$ we showed that there are no edges between $U(s)$ and $U(s')$ \wehp Thus, using in parallel for each $s \in S$ on the induced subgraph $G[U(s)]$ \cite[Theorem 1.2]{fischer2020improved}, we obtain a maximal matching for all such induced subgraphs $G[U(s)]$ in $\bigO(\log^3\log n)$ rounds without conflicts \wehp
\end{proof}
\begin{proof}[Proof of \Cref{thm:mainpolylog} (MM part)]
    \textbf{\LOCAL:} For MM, we obtain a shattering where each connected component is of size at most $\polylog n$ \wehp after $\bigO(1)$ using \Cref{pro:shattering-mm}. Applying in parallel for each component \cite[Theorem 3.1]{GG24}, yields  $\tilde{\bigO}(\log^{5/3}\log n)$ rounds \wehp.

\smallskip

    \textbf{\CONGEST:} The shattering into poly-logarithmic components in  $\bigO(\log^3\log n)$ rounds follows similarly from \Cref{pro:shattering-mm} \wehp Consequently, by applying in parallel for each connected component the maximal matching algorithm from \cite[Theorem 1.2]{fischer2020improved}, resulting in a round complexity of $\bigO(\log^3\log n)$.
\end{proof}
\section{Lower Bounds for MIS/MM \& Substructures in HRGs (Theorem \ref{thm:lowerbound})}\label{sec:lower-bound}
In this section, we prove our lower bounds for MIS and MM on HRGs. We show the following statement.

\theoremlowerbound*

To accomplish our lower bounds, we first establish that there exist many $d$-ary trees in HRGs with relatively large degree $d$ and height $h$ (\Cref{pro:d-regular-trees} in \Cref{sec:d-regular-trees}). We then use these substructures to obtain our lower bounds of \Cref{thm:lowerbound} in \Cref{sec:lower-bound-corollary}.
\subsection{On \texorpdfstring{$d$}{d}-ary Trees in Hyperbolic Random Graphs (Proof of Theorem \ref{pro:d-regular-trees})} \label{sec:d-regular-trees}
In this section, we prove the existence of (polynomially many) $d$-ary trees with $n' \approx \sqrt{\log n}$ vertices, of any degree $d  \leq n'$ growing in $n$ and height $h \approx \log(n')/\log(d)$ in an HRG. Specifically, we show the following.
\begin{restatable}[$d$-ary trees in hyperbolic random graphs]{theorem}{theoremregulartree}
\label{pro:d-regular-trees}
Let $G$ be a threshold hyperbolic random graph and let $m:= m(n) \leq \sqrt{\log n}$ be a function growing in $n$. Then, \aas, there exist $n^{\Omega(1)}$ disjoint sectors $\Phi_i$ such that for $T_i := G[V \cap \Phi_i]$, $T_i$ is a balanced $d$-ary tree\footnote{A rooted tree where each vertex, except for the leaves, has $d$ children.} with $n' \in \Omega(\sqrt{\log n})$ vertices, degree $d \in \Omega(m)$ and height $h \in \Omega(\log_m\log n')$. In particular, for each such tree $T_i$ with root $u$\footnote{The root node $u$ is the vertex with graph distance $d_G(u,w) = h$ to each leaf $w \in V(T_i)$.}, there exists a vertex $v \in V(H)$ in the giant component $H$ of $G$ such that $\{u,v\} \in E(H)$ is the unique edge between $T_i$ and $G \setminus V(T_i)$.
\end{restatable}
The parameter $m$ of \Cref{pro:d-regular-trees} can be tuned to adjust the degree $d$ and also implicitly the height $h$ for a desired $d$-ary tree. For example, setting $m = \sqrt{\log n}$ gives a "star graph" with degree $\Omega(\sqrt{\log n})$ and height $1$. For our purposes, we will set $m = \log\log n$ later on in~\Cref{sec:lower-bound-corollary}. 

To prove \Cref{pro:d-regular-trees}, we partition the hyperbolic disk into $k$-many sectors where $k$ is polynomial in $n$. For polynomially many such sectors, we reveal that they contain a $d$-ary tree with the parameters $n', d,$ and $h$. The rough idea of how we find this structure goes as follows. We fix a sector $\Phi$ with angle $\phi$ and aim to find a root node $u$ located in layer $\Layer{\ell_0}$ intersecting sector $\Phi$ (where $\ell_0 \approx \log\log n$; see also \Cref{fig:d-regular-trees}a). Next we consider $d$ equally sized boxes in a layer $\Layer{\ell_1}$ with $\ell_1 < \ell_0$, where each box has width $\phi_{\ell_1}$ and height $1$. We choose these parameters such that they fulfil the following property: if a vertex $v$ lies in a box in layer $\Layer{\ell_1}$, then it is adjacent to the root node $u$ and $v$ has distance at least $R$ to any point in layer $\Layer{\ell_1}$ if the angular distance is larger than $\phi_{\ell_1}$ (see \Cref{fig:boxes} for a visualisation). 

For each box in layer $\Layer{\ell_1}$, we find exactly one vertex such that the degree of $u$ is $d$ ($u$ has $d$ children). Moreover, for any pair of boxes in $\Layer{\ell_1}$, by our parameter choice $\phi_{\ell_1}$, any pair of ``children'' of $u$ has distance at least $R$ and thus does not share an edge. Then we recursively build the tree by ``assigning'' any vertex $v$ in layer $\Layer{\ell_1}$ a set of $d$ boxes in a layer $\Layer{\ell_2}$, such that each box contains a vertex. We repeat this for each vertex in the boxes in $\Layer{\ell_2}$, obtaining boxes in layer $\Layer{\ell_3}$ and so on until we have boxes in layer $\Layer{\ell_h}$, yielding a tree of height $h$. To prove that this indeed yields the sought-after $d$-ary tree structure, we first show that the graph resulting from the geometric embedding, given that we have exactly one vertex in each assigned box while the rest of the sector is empty, indeed results in a $d$-ary tree \Cref{lem:tree-geometry}. Afterwards, we show that this event occurs with a small but not polynomial-decaying probability (\Cref{lem:tree-prob}). This, in conjunction with the event that all vertices used to build the tree have no further neighbours, occurs for a fixed, slightly larger sector with probability $n^{-o(1)}$. Since we have polynomial many sectors, there are, in expectation, $ n^{\Omega(1)}$ many sectors where the desired $d$-ary tree occurs. This, the bridge of the root to the giant component and the concentration result are addressed in the proof of \Cref{pro:d-regular-trees}.

To make this formal, we first define for any vertex $u$ the boxes where we wish to find exactly one vertex each, representing the children of $u$ (see also \Cref{fig:boxes}b for a sketch of the definition).
\begin{figure}[t]
    \centering \includegraphics[height=0.2\textheight]{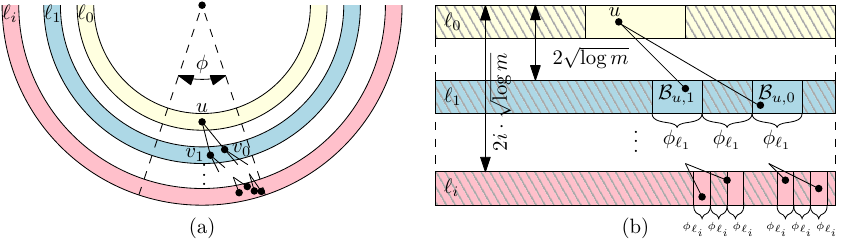}
    \caption{Sketch of our tree construction with degree $d=2$. (a) Layer $\Layer{\ell_0}$ contains the root vertex $u$ which has children $v_0$ and $v_1$ in layer $\Layer{\ell_1}$. (b) Illustration of the boxes $\boxbelow{u}{0}$ and $\boxbelow{u}{1}$ of $u$. A sector with angle $\phi$ is nice if the hatched area is empty and each box contains exactly one vertex.}
    \label{fig:boxes}
\end{figure}
\begin{definition}[Boxes]\label{def:boxes}
    Let $\mathcal{D}_R$ be a hyperbolic disk with radius $R = 2\log n + C$ and $10000\leq m \leq \sqrt{\log n}$. Moreover, define
    \begin{itemize}
    
        \item (root layer) $\rootlayer := \lceil \log\log n \rceil$,
        \item (layer of tree level $i$) $\ell_i :=  \rootlayer - \lceil 2i \cdot {\log m} \rceil$,
        \item (box width in tree level $i$) $\phi_{\ell_i} := \frac{100e^{\ell_{i} + C/2}}{n}$,
        \item (height of tree) $h := \lfloor{\log\log n}/(2\log m) \rfloor$ and
        \item (degree in tree) $d := \lfloor m/10000\rfloor$.
        
    \end{itemize}
    Consider any $i \in [h]$ and a vertex $u \in V_{\ell_i}$. Then for every $j \in [d]$, the $j$-th box of vertex $u$ is defined by the set of points
    \begin{align*}
        \boxbelow{u}{j} := \left\{x \in \Layer{\ell_{i+1}} : \varphi(u) + \frac{e^{\ell_i + C/2}}{10 n\cdot m} - (2j+1) \cdot \phi_{\ell_{i+1}} \leq\varphi(x) \leq \varphi(u) +\frac{e^{\ell_i + C/2}}{10 n\cdot m} - 2j \cdot \phi_{\ell_{i+1}} \right\}.
    \end{align*}
\end{definition}
Before embarking on the proof of our desired geometric properties, let us give some further intuition for our choice of parameters. Placing the "first" box of $u$ at angular distance $\approx \frac{e^{\ell_i}}{n\cdot m}$ ensures by \Cref{lem:max-angle} that $u$ "barely" has an edge to a vertex in this box. Then, "moving" the subsequent boxes of $u$ by angular difference $\approx \frac{e^{\ell_i}}{n}$ ensures that the hyperbolic distance between $u$ and any vertex contained in one of its boxes is at most $R$ as well. Moreover, by ``skipping'' every second box and by our choice $\phi_{\ell_i}$ we achieve that any pair of vertices in two different boxes on the same level does not have an edge. The multiplicative constant factors of $\phi_{\ell_i}$ and $d$ are chosen for convenience so that these properties are fulfilled. This gives us the desired property of a tree that $u$ is a ``parent'' to all vertices in its ``own'' boxes $\boxbelow{u}{j}$ while vertices that lie in boxes of $u$ are the ``children'' of $u$ and do not have an edge among each other.

To formalise our construction, we consider \Cref{alg:nice-sector}, which defines a \emph{nice} sector. We shall prove that the construction indeed gives a $d$-ary tree in a sector $\Phi$\footnote{The value for $m$ in \Cref{alg:nice-sector} is here the same as the parameter $m$ in \Cref{pro:d-regular-trees}.} (see also \Cref{fig:boxes}).
\begin{algorithm}
\caption{Nice sector $\Phi$}
\begin{algorithmic}[1]\label{alg:nice-sector}

\Require Sector $\Phi \subseteq \mathcal{D}_R$ with bisector $\varphi$ and angular width $\phi = \frac{5 \log n}{n\log m}$.

\State \textbf{assert} $|V_{\ell_0} \cap \Phi'| = 1$ where $\Phi' = \{\Phi \in \disk : \varphi - \phi/5 \leq \varphi(x) \leq \varphi +  \phi/5\} $ //Check for root node

\State $U \gets \{V_{\ell_0} \cap \Phi'\}$ //Add root node
\State $U_{all} \gets \emptyset$ //All vertices of the tree

\For{$i = 0$ to $h-1$} //Iterate through the levels of the tree
\State $U_{\text{all}} \gets U_{\text{all}} \cup U$ //Update the set of all vertices in
    \State $U_{\text{next}} \gets \emptyset$ //Next level of the tree
    \For{each $v \in U$} //Iterate through the vertices of a level
        \For{$j = 0$ to $d-1$} //Iterate through the boxes of a vertex
        \State \textbf{assert} $|V \cap \boxbelow{v}{j}| = 1$ //Check that there is exactly one vertex in each box of $v$
            \State $U_{\text{next}} \gets U_{\text{next}} \cup (V \cap \boxbelow{v}{j})$ //Add vertex of the box
        \EndFor
    \EndFor
    \State $U \gets U_{\text{next}}$
\EndFor

\State \textbf{assert} $(V \cap \Phi) \setminus U_{\text{all}} = \emptyset$ //Check if sector is empty except for the tree

\end{algorithmic}
\end{algorithm}
\begin{lemma}[$d$-ary tree geometry]\label{lem:tree-geometry}
    Let $m \in \omega(1)$ and $\Phi \subseteq \mathcal{D}_R$ be a sector with angular width $\phi = \frac{5 \log n}{n\cdot \log m}$. Then, if $\Phi$ is nice according to \Cref{alg:nice-sector}, the induced sub-graph $G[V \cap \Phi]$ is a $d$-ary tree with degree $d \in \Omega(m)$, height $h \in \Omega(\log\log n/ \log m)$ and $n' \in \Omega(\sqrt{\log n})$ vertices.
\end{lemma}
\begin{proof}
    We show the desired statement as follows: first, we show that all boxes are within the sector $\Phi$. Then, we establish for any vertex $u \in V \cap \Phi$ that any point in a box of $u$ has distance at most $R$ to $u$. Note that this already suffices to prove that all desired edges for a tree exist such that every vertex has $d$ children (except for the leaves of the tree). Afterwards, we show that any vertex $u$ has a distance larger than $R$ to any point of a box that is not the box of $u$ according to \Cref{def:boxes}. Since the rest of the sector is empty, this ensures that we also have all non-edges. Using that we have $h$ layers through which we iterate in line 4 of \Cref{alg:nice-sector}, this yields a tree of depth $h$ where each vertex has $d$ children because of the loop in line 8. For the desired properties of our tree, $d$ and $h$ follow directly from \Cref{def:boxes} and the number of vertices is given by $n' \geq d^{h} \in \Omega(\sqrt{\log n})$.

    \paragraph{All boxes are contained in the sector.} By \Cref{def:boxes}, a vertex $u$ in layer $\Layer{\ell_i}$ has boxes with angular distance at most $\frac{e^{\ell_i + C/2}}{10 n\cdot m}$ since $d\cdot \phi_{\ell_{i+1}} \leq \exp(C/2 + \ell_i)/(100n\cdot m)$. Hence, the maximal angle that is spanned by any pair of points $x,y$ in two different boxes is at most 
    $$
    \angulardist{x}{y} \leq \sum_{i=0}^{h-1} \frac{e^{\ell_i + C/2}}{10 n\cdot m} \leq \Theta(1)\cdot \frac{e^{\ell_0}}{n\cdot m} \in o(\phi),  
    $$
where we used that $\ell_i =  \lceil \log\log n \rceil - \lceil 2i \cdot {\log m} \rceil$, $\phi = 5\cdot \log n / (n\cdot \log m)$ and $m \in \omega(1)$ in the last step. Thus, all boxes are contained in sector $\Phi$ with angular width $\phi$.

    \paragraph{Edges.} Recall that $R = 2\log n + C$. To show that all desired edges exist, consider any vertex $u \in V_{\ell_{i}}$. Note that any box $\boxbelow{u}{j}$ is in layer $\Layer{\ell_{i+1}}$. Hence, using that $d\cdot \phi_{\ell_{i+1}} \leq \exp(C/2 + \ell_i)/(100n\cdot m)$ it follows by the definition of $\boxbelow{u}{j}$ that the angular distance between $u$ and $x \in \boxbelow{u}{j}$ is at most 
    \begin{align*}
    \angulardist{u}{x} \leq \frac{e^{\ell_i + C/2}}{10 n\cdot m} \leq e^{(R - r(u) - r(x))/2} \leq \theta_R(r(u), r(x)),
    \end{align*}
using that $r(u) \leq R - \ell_i$, $r(v) \leq R - \ell_i + 2\log m + 1$ and \Cref{lem:max-angle} in conjunction with $\theta_R(r(u), r(x))$ being monotonically decreasing by \Cref{rmk:theta-monotonicity}. Thus, all desired edges exist.

\paragraph{Non-edges.} Next, let $u \in V_{\ell_i}$ and $v \in V_{\ell_j}$. W.l.o.g. let $\ell_i \geq \ell_j$. Then, if $v$ is not in a box of $u$, we have to show that $d_h(u,v) > R$ to finish the proof. To this end, we distinguish two cases.

\smallskip

\textbf{Case 1}~[Vertex $u$ is an ancestor of $v$]: Let $w$ be the vertex that is in the $(d-1)$-th box of $u$, i.e., $\boxbelow{u}{d-1}$ and note that, since $d\cdot \phi_{\ell_{i+1}} \leq \exp(C/2 + \ell_i)/(100n\cdot m)$,
\begin{align}\label{eq:lower-bound-angle}
    \angulardist{u}{v} \geq \angulardist{u}{w} > (1/10 -1/100)\frac{e^{C/2 + \ell_i}}{m\cdot n} ,
\end{align}
using the definition for the box $\boxbelow{u}{j}$ (\Cref{def:boxes}).
On the other hand, using that $u$ is an ancestor of $v$ but $v$ is not in a box of $u$, we observe that $\ell_i \geq \ell_j + 4\log m$. Thus, we obtain
 \begin{align}\label{eq:lower-bound-angle2}
     \theta_R(r(u), r(v)) \leq \Theta(1) \frac{e^{C/2 + \ell_i}}{m^2 \cdot n},
 \end{align}
by \Cref{lem:max-angle}. Hence, by combining \Cref{eq:lower-bound-angle} and \Cref{eq:lower-bound-angle2} we have by $m \in \omega(1)$
$$
\angulardist{u}{v} > \Theta(1)\frac{e^{C/2 + \ell_i}}{m\cdot n} > \Theta(1) \frac{e^{C/2 + \ell_i}}{m^2 \cdot n} \in \omega(\theta_R(r(u), r(v))),
$$
and there is no edge between $u$ and $v$ by \Cref{lem:max-angle}. This concludes our first case.

\smallskip

\textbf{Case 2}~[Vertex $u$ is not an ancestor of $v$]: Let $w$ be the vertex that is the lowest common ancestor of $u$ and $v$. First, consider the case that $w$ is the parent of both $u$ and $v$. In this case, $w \in V \cap \Layer{\ell_{i-1}}$ and as such, the minimal angular distance between $u$ and $v$, by \Cref{def:boxes}, is 
$$
\angulardist{u}{v} \geq \phi_{\ell_i} = \frac{100e^{\ell_{i} + C/2}}{n} > \theta_R(r(u), r(v)),
$$
where we applied \Cref{lem:max-angle} and $r(u), r(v) \geq R - \ell_i -2$. Hence, there is no edge between $u$ and $v$.

To finalise the case, consider now that $u$ and $v$ do not share the same parent such that $\ell_i > \ell_j$.
Let the lowest common ancestor $w$ be in layer $\Layer{\ell_k}$ and consider the two children $u'$ and $v'$ of $w$, (where both $u'$ and $v'$ are in layer $\Layer{\ell_{k+1}}$), that are ancestors of $u$ and $v$ respectively (possibly $u' = u$). Note that by $d\cdot \phi_{\ell_{i+1}} \leq \exp(C/2 + \ell_i)/(100n\cdot m)$ and using \Cref{def:boxes} it follows that
\begin{align}\label{eq:case2-part1}
    \angulardist{u'}{v'} \geq \phi_{\ell_{k+1}} = \frac{100 \cdot e^{\ell_k + C/2 }}{m^2 \cdot n}.
\end{align}
Then, using that $\varphi(u)\geq \varphi(u')$ and that $\varphi(v) \leq \varphi(v') + \sum_{t = j}^{k-1} \frac{e^{(C+ \ell_t + \ell_{t -1})/2}}{n}$ by \Cref{def:boxes}, it follows
\begin{align}\label{eq:case2-part2}
    \angulardist{u}{v} \geq \angulardist{u'}{v'} - \underbrace{\sum_{t = j}^{k-1} \frac{e^{(C+ \ell_t + \ell_{t -1})/2}}{n}}_{\bigO\left(\frac{e^{\ell_k}}{m^3 n}\right)}.
\end{align}
Plugging \Cref{eq:case2-part1} into \Cref{eq:case2-part2} and upper bounding the sum of exponentials by $\bigO(\exp{(\ell_k - 3\log m)} /n)$ we obtain via $m \in \omega(1)$
\begin{align*}
         \angulardist{u}{v} \geq \frac{(100 - o(1)) e^{\ell_k + C/2}}{m^2 \cdot n} >  \theta_R(r(u), r(v)),
\end{align*}
using that $\ell_k \geq \ell_i + 2\log m \geq \ell_j + 4\log m $ and $r(u) \geq R - \ell_i - 2, r(v) \geq R - \ell_j - 2$ in conjunction with \Cref{lem:max-angle}. Hence, there is no edge between $u$ and $v$ as desired. This concludes the case and the proof.
\end{proof}
We now prove that a fixed sector is nice with a small but not polynomial decaying probability.
\begin{lemma}[$d$-ary tree probability]\label{lem:tree-prob}
    Let $m \in \omega(1)$ and $m \leq \sqrt{\log n}$. Moreover, let $\Phi \subseteq \mathcal{D}_R$ be a sector with angular width $\phi = \frac{5 \log n}{n\log m}$. Then, $\Phi$ is nice according to \Cref{alg:nice-sector} with probability $n^{-o(1)}$.  
\end{lemma}
\begin{proof}
 We show that all boxes of a nice sector according to \Cref{alg:nice-sector} have exactly one vertex with the desired probability $n^{-o(1)}$. To this end, fix any $i \in [h+1]\setminus\{0\}$ and consider any "non-root" layer $\Layer{\ell_i}$. Then, to fulfil the property of a nice sector according to \Cref{alg:nice-sector}, all $d^i$ boxes include exactly one vertex. Let $u$ be a vertex in layer $\Layer{\ell_{i-1}}$. Using \Cref{def:boxes} it holds for any single box in layer $\Layer{\ell_i}$ 
    \begin{align*}
        \mu(\boxbelow{u}{j}) = \frac{\phi_{\ell}}{2\pi} \cdot \mu(\Layer{\ell_i}) = \Theta(1)\frac{e^{\ell_i(1-\alpha)}}{n} \in \Omega(1/n),
    \end{align*}
    where we used \Cref{eq:layer_measure} and for the last step that $\ell_i \geq 0$. Thus, using the fact that the number of vertices in a box follows a Poisson point distribution, the probability that there is exactly one vertex in a box in layer $\Layer{\ell_i}$ is
    \begin{align}\label{eq:non-root-prob}
        p_i = n\cdot \mu(\boxbelow{u}{j}) \cdot e^{-n\cdot\mu(\boxbelow{u}{j})}\geq e^{\ell_i(1-\alpha)} \cdot \exp\left(-\Theta(1){e^{\ell_i(1-\alpha)}}\right).
    \end{align}
Hence, using that any pair of boxes is disjoint, the probability that all $d^i$ boxes in layer $\Layer{\ell_i}$ have exactly one vertex is given by
$$
\Prob{\mathcal{E}_{i}} = (p_i)^{d^i} \geq \exp\left(-\Theta(1){e^{\ell_i(1-\alpha)}}\right)^{d^i}.
$$
Next, we bound the probability that the desired "root" vertex in layer $\Layer{\ell_0}$ exists. Since the "box" of the "root" in \Cref{alg:nice-sector} has angle $\phi/5 = \log n /n\log(m)$ the probability that exactly one root vertex exists is by a Poisson distribution and \Cref{eq:layer_measure}
\begin{align}\label{eq:root-prob}
    p_0 = \frac{n\cdot\phi}{2\pi} \cdot \mu(\Layer{\ell_0})\cdot \exp{\left(-\frac{n\cdot\phi}{2\pi} \cdot \mu(\Layer{\ell_0})\right)} \geq \exp\left(-\Theta(1){e^{\ell_0(1-\alpha)}}\right).
\end{align}
since $m \leq \sqrt{\log n}$ and $\alpha < 1$. This implies that sector $\Phi$ has vertices that are required to be nice, i.e., there is the root vertex and every box has exactly one vertex as required in \Cref{alg:nice-sector}, with probability 
\begin{align*}
  \Prob{\mathcal{E}} \geq p_0 \cdot \prod_{i=1}^{h} \Pro{\mathcal{E}_{\ell_i}} \geq \prod_{i=0}^{h} \left((p_i)^{d^i}\right) \geq \prod_{i=0}^{h} \exp\left(-\Theta(1){e^{\ell_i(1-\alpha)}}\right)^{d^i},
\end{align*}
where we used \Cref{eq:non-root-prob} and \Cref{eq:root-prob} in the last step. Next, we use that $\alpha > 1/2$ which yields
\begin{align}\label{eq:probability-product}
  \Prob{\mathcal{E}}  \geq \prod_{i=0}^{h} \exp\left(-\Theta(1){d^i}{e^{\ell_i/2}}\right).
\end{align}
Since $\ell_i \leq \log\log n - 2i\cdot\log m + 1$, we obtain
\begin{align*}\
  \Prob{\mathcal{E}}  \geq \prod_{i=0}^{h} \exp\left(-\Theta(1){d^i}\cdot \sqrt{\log n}/m^i\right) \geq \prod_{i=0}^{h} \exp\left(-\Theta(1)\cdot \sqrt{\log n}/10000^i\right),
\end{align*}
where we used $d^i \leq (m/10000)^i$. Applying a geometric series, the above yields
\begin{align*}\
  \Prob{\mathcal{E}}  \geq \exp\left(-\Theta(1)\cdot\sum_{i=0}^{h} \sqrt{\log n}/10000^i\right) \geq e^{-\Theta(1)\sqrt{\log n}} \in n^{-o(1)},
\end{align*}
and we conclude that $\Phi$ has all vertices that are required to be nice with a probability that is decaying more slowly than polynomial in $n$. 

Finally, we lift this to the statement that the sector is also nice with probability $n^{-o(1)}$: we obtain a nice sector by showing that the sector, except for the boxes, is empty with essentially the same probability. To see this, note that by our choice of $\phi \in \bigO(\log n /(m n))$ and using that the number of vertices in sector $\Phi$ follow a Poisson distribution with expectation $\E{|V \cap \Phi|} = n \phi /2\pi$, the probability that the sector is empty is given by 
$$
\Prob{V \cap \Phi = \emptyset} = e^{-\E{|V \cap \Phi|}} \in e^{-\bigO(\log n /m)} \in n^{-o(1)},
$$
using $m \in \omega(1)$. Hence, given any area $\mathcal{S} \subseteq \disk$, the probability that $\Phi\setminus \mathcal{S}$ is empty is at least $n^{-o(1)}$ (using the fact that the probability of an area being empty is monotonically increasing if the considered area is shrinking). Then, let $\mathcal{S}$ be the area that we need to reveal for event $\mathcal{E}$ and it follows that a sector $\Phi$ is nice  with probability
$$\Prob{\Phi \text{ is nice}} \geq \Prob{\mathcal{E}}\cdot\Prob{V \cap (\Phi\setminus \mathcal{S}) = \emptyset | \mathcal{E}} \geq \Prob{\mathcal{E}}\cdot\Prob{V \cap \Phi = \emptyset}  \in n^{-o(1)},$$
as claimed.
\end{proof}
Next, we ``embed'' a nice sector in a larger ``buffer'' sector (see \Cref{fig:d-regular-trees}b for a sketch). This allows us, given that the ``buffer sector'' is large enough, to get independent probabilities among nice sectors to be disconnected from any other vertex in the disk. We use this to show that (polynomially) many $d$-ary trees exist and that they basically form their own component (except for a ``cut edge'' that goes from the root vertex of the tree to the giant, which is addressed in the ``in particular'' part of the statement ).
\theoremregulartree*
\begin{proof}
Throughout the proof, we partition the disk $\mathcal{D}_R$ into $k = \big\lceil \frac{n^{1-\frac{1}{2\alpha}}}{\log n}\rceil$ sectors such that for $i \in [k]$, the angle $\psi$ of any such sector $\Psi_i$ is $\psi \in \Theta\left(\frac{n^{\frac{1}{2\alpha}}\log n}{n}\right)$; see red sector in~\Cref{fig:d-regular-trees}b for an illustration. Moreover, let $r^* := R - \log n/\alpha - \log\log n /2$ and throughout the proof, we consider the area $\B_0(r^*)$; see hatched area in \Cref{fig:d-regular-trees}.

We prove our theorem in two steps. In the first step of the proof, we show that any sector $\Psi_i\setminus\B_0(r^*)$ contains a desired $d$-ary tree $T_i$ with probability $n^{-o(1)}$ if it contains a nice sector $\Phi_i$; blue sector in \Cref{fig:d-regular-trees}. This is addressed in Claim~\ref{claim:indicator}. 

Then, in a second step, we show that we have $n^{\Omega(1)}$ "buffer sectors" $\Psi_i$ \wehp that contain a tree $T_i$ and that if the area $\B_0(r^*)$ is empty, then the neighbourhood of any $d$-ary tree $T_i$ does not contain a vertex of any vertex of any other "buffer sector" $\Psi_j$. Showing that $\B_0(r^*)$ is empty \aas then gives the desired probabilistic guarantee of any $T_i$ not having any undesired edges. Finally, we finish the proof by showing that the edge of the "root" in $T_i$  which exists due to Claim~\ref{claim:indicator} connects $T_i$ to the giant in the desired fashion of the "in particular" statement of our theorem.

Now, for Claim~\ref{claim:indicator}, we consider any $\Psi_i$ and let $\Phi_i \subset \Psi_i$ be the sector of width $\phi = \frac{5\log n}{n\log m} \in o(\log n / n)$ as demanded in \Cref{lem:tree-prob} that has the same bisector as $\Psi_i$. In particular we consider a sector $\Psi_i$ where the corresponding sector $\Phi_i$ is nice. We introduce the following event $\mathcal{E}_i$ which says that no vertex in a nice sector $\Phi_i$ has an edge to a vertex in $\Psi_i \setminus \Phi_i$ except for the ``root'' vertex $u \in V_{\ell_0} \cap \Phi_i$ which has exactly one edge to a "special" vertex $v$ in $\Psi_i\setminus \Phi_i$ where $v$ lies in layer $\overline{\ell} := \lceil 2\log\log n/(1-\alpha)\rceil$. We later show that $v$ is part of the giant component.
\begin{enumerate}
\item[\textbf{Event $\mathcal{E}_i$}:]\label{item:eventful}
 let $U: = V \cap (\Phi_i \setminus \Layer{\ell_0})$, i.e., all vertices in $\Phi_i$ except for the ``root'' vertex in layer $\ell_0$, and let $\mathcal{E}_U:= \{$No vertex in $U$ has any neighbour in $\Psi_i \setminus (\Phi_i \cup \B_0(r^*))\}$. Moreover, let $u \in V \cap (\Phi_i \cap \Layer{\ell_0})$, i.e., the root vertex of $\Phi_i$, and let $\mathcal{E}_{\text{root}}:= \{\text{Vertex } u \text{ has exactly one neighbour $v$ in  } \Psi_i \setminus (\Phi_i \cup \B_0(r^*)), \text{ and } v \text{ is in layer } \Layer{\overline{\ell}}\}$.
\end{enumerate}
We define $\mathcal{E}_i := \mathcal{E}_U \cap \mathcal{E}_{\text{root}}$ and we will show that this event holds with probability $ n^{-o(1)}$. 
\begin{claim}\label{claim:indicator}
  $\Prob{\mathcal{E}_i \text{ }|\text{ } \Phi_i \text{ is nice}} \in n^{-o(1)}$. 
\end{claim}
{
\renewcommand{\qedsymbol}{$\blacksquare$} 
\begin{proof}[Proof of claim]
\textbf{Event $\mathcal{E}_U$:} We start by showing the desired properties for all "non-root" vertices. Conditioning on the property that $\Phi_i$ is nice according to \Cref{alg:nice-sector}, the probability for a vertex $v \in V_{\ell_j}$ in a layer $\ell_j < \ell_0$ to have no neighbour in $\Psi_i \setminus \Phi_i$ is at least
\begin{align}\label{eq:nice1}
\Prob{N(v) \cap  (\Psi_i \setminus \Phi_i) = \emptyset \text{ }|\text{ } v \in V_{\ell_j}, \Phi_i \text{ is nice}} \geq e^{-n\cdot\mu(\B_v(R) \cap \B_0(R))} \geq \exp{\left(-\Theta(1)e^{\ell_j /2}\right)},    
\end{align}
using that we have a Poisson distribution in conjunction with \Cref{eq:layer_expected_degree}.

Now, given that $\Phi_i$ is nice, let $\{v_1, v_2, \dots v_{n' - 1}\}$ be our set $ V \cap (\Phi_i \setminus \Layer{\ell_0})=:U$, i.e., vertices in a nice sector except for the ``root'', in no particular order. Moreover, let $\mathcal{E}_t$ be the event that the set of vertices  $\{v_1, v_2,\dots v_{t-1}\}$ have no neighbour in $\Psi_i \setminus \Phi_i$. Since we use a Poisson distribution, the probability of an area being empty is monotonically increasing if the considered area is shrinking and it holds for any $t$ that
\begin{align}\label{eq:nice2}
\Prob{v_{t} \text{ has no neighbour in $\Psi_i \setminus \Phi_i$ }| \mathcal{E}_t, \Phi_i \text{ is nice}} \geq \Prob{N(v_t) \cap  (\Psi_i \setminus \Phi_i) = \emptyset \text{ }|\text{ } v \in V_{\ell_j}, \Phi_i \text{ is nice}}. 
\end{align}
Then,  we obtain via \Cref{eq:nice1} and \Cref{eq:nice2} for the event $\mathcal{E}_U= \{$No vertex in $U$ has any neighbour in $\Psi_i \setminus (\Phi_i \cup \B_0(r^*))\}$ that 
\begin{align*}
\Prob{\mathcal{E}_U} &\geq \prod_{t = 1}^{n'-1}\Prob{v_{t} \text{ has no neighbour in $\Psi_i \setminus \Phi_i$ }| \mathcal{E}_t, \Phi_i \text{ is nice}}\\
&\geq \prod_{j = 1}^{h}\prod_{v \in V_{\ell_j} \cap \Phi_i} \Prob{N(v) \cap  (\Psi_i \setminus \Phi_i) = \emptyset \text{ }|\text{ } v \in V_{\ell_j}, \Phi_i \text{ is nice}}  \geq \prod_{j = 1}^{h} \exp{(-\Theta(1)e^{\ell_j /2}})^{d^j},     
\end{align*}
where we used in the last step that the number of vertices in layer $\ell_j$ of a nice sector $\Phi_i$ is $d^j$. By the same arguments we applied for \Cref{eq:probability-product}, we then get that 
\begin{align}\label{eq:prob-of-U}
  \Prob{\mathcal{E}_U} \in n^{-o(1)}.  
\end{align}
This concludes the part for "non-root" vertices. 

\smallskip

\textbf{Event $\mathcal{E}_{\text{root}}$:} Next, we consider the root vertex $u$ and establish our ``cut edge'' $\{u,v\}$ which we later show, connects $u$ to the giant. Let $u$ be the vertex in $\Phi_i$ that lies in layer $\Layer{\ell_0}$, i.e., the root of our tree in the nice sector $\Phi_i$. W.l.o.g. let $\varphi(u) = 0$ and consider for $\overline{\ell} = \lceil 2\log\log n/(1-\alpha)\rceil$ the set of points
$$
\mathcal{S}= \left\{x \in \Layer{\overline{\ell}} : - \frac{e^{C/2}\log^{\frac{1}{2} + \frac{1}{1-\alpha}}(n)}{n} \leq \varphi(x) \leq - \frac{e^{C/2}\log^{\frac{1}{2} + \frac{1}{1-\alpha}}(n)}{2n}\right\} \subset \Psi_i.
$$
We proceed by setting out the following goals.
\begin{enumerate}
    \item[(G1)] We show that any point in $x$ has distance at most $R$ to $u$ while any other vertex in $\Phi_i$, i.e., $v \in (V \cap \Phi_i)\setminus \{u\}$ , has distance at least $R$ to any point in $\mathcal{S}$.
    
    \item[(G2)] We show that there is exactly one vertex in $\mathcal{S}$ with probability $n^{-o(1)}$.
        
    \item[(G3)] We show that, under the condition of event $\mathcal{E}_U$, the probability of the event that $u$ has no neighbour in $\Psi_i \setminus (\Phi_i \cup \mathcal{S})$ is $n^{-o(1)}$.
\end{enumerate}
Note that if all goals are accomplished, then root $u$ has all desired properties for event $\mathcal{E}_{\text{root}}$. Moreover, if event $\mathcal{E}_U$ occurs, there are no undesired edges of ``non-root'' vertices in $\Psi_i \setminus \Phi_i$. Thus, since the events of (G2) and (G3) are holding independently with probability $n^{-o(1)}$ while the event of (G1) is deterministic, we get that the intersection of the event of all goals combined occur with probability $n^{-o(1)}$ such that
\begin{align*}
\Prob{\mathcal{E}_i \text{ }|\text{ } \Phi_i \text{ is nice}} = \Prob{\mathcal{E}_{\text{root}} \cap \mathcal{E}_{U} \text{ }|\text{ } \Phi_i \text{ is nice}} \geq \Prob{\mathcal{E}_U}\cdot \Prob{\mathcal{E}_{\text{root}} | \Phi_i \text{ is nice and event } \mathcal{E}_U \text{ occurs.}} \in n^{-o(1)},   
\end{align*}
as desired by using that event $\mathcal{E}_U$ and $\mathcal{E}_{\text{root}}$ are positively correlated and we previously established $\Prob{\mathcal{E}_U} \in n^{-o(1)}$ in~\Cref{eq:prob-of-U}. We continue by accomplishing each goal.
\paragraph{Goal 1.} To tackle that $u$ has distance at most $R$ to $x \in \mathcal{S}$, note that
$$
\angulardist{u}{x} \leq \frac{e^{C/2}\log^{\frac{1}{2} + \frac{1}{1-\alpha}}(n)}{n} \leq \frac{e^{(C + \ell_0 + \overline{\ell})/2}}{n}< \theta_R(r(u), r(x)),
$$
using that $\ell_0 = \lceil \log\log n \rceil$ and $\overline{\ell} = \lceil 2\log\log n/(1-\alpha) \rceil$ in conjunction with \Cref{lem:max-angle}. Thus, $u$ would have an edge to any potential vertex in $\mathcal{S}$.

Next, to see that any other vertex in $\Phi_i$ does not have an edge to a potential $v \in V \cap \mathcal{S}$, note that any vertex $w \in V \cap \Phi_i$ has angle $\varphi(w) > \varphi(u) = 0$, since $\Phi_i$ is nice. Hence, it holds for any $x \in \mathcal{S}$, (by $\ell_j < \ell_0$), that for any $w \in V \cap (\Phi_i \cap \Layer{\ell_j})$ the angular distance is 
$$
\angulardist{w}{x} \geq \frac{e^{C/2}\log^{\frac{1}{2} + \frac{1}{1-\alpha}}(n)}{2n} > \frac{1000 e^{(C + \ell_j + \overline{\ell})/2}}{n}> \theta_R(r(w), r(x)),
$$
since $r(w) \geq R - \ell_j - 1 \geq R - \log\log n + 2\log m - 2$, $m \in \omega(1)$ and $r(x) \geq R - \overline{\ell} - 1 \geq R - 2\log\log/(1-\alpha) - 2$. Thus, any vertex in $\Phi_i$, other than the ``root'' $u \in V \cap \Layer{\ell_0}$, would not have an edge to any vertex in $\mathcal{S}$ by \Cref{lem:max-angle}. This concludes the first goal.

\paragraph{Goal 2.} For the second goal, we calculate the measure for area $\mathcal{S}$ and obtain
$$
\mu(\mathcal{S}) = \mu(\mathcal{\Layer{\overline{\ell}}})\cdot \frac{e^{C/2}\log^{\frac{1}{2} + \frac{1}{1-\alpha}}(n)}{4\pi n} \in 1/(n\log^{\Theta(1)} n),
$$ 
using the angle spanned by area $\mathcal{S}$ and using \Cref{eq:layer_measure} in conjunction with $\overline{\ell}  \in \Theta(\log\log n)$. Hence, the expected number of vertices in $\mathcal{S}$ is 
$$
\E{|V \cap \mathcal{S}|} \in 1/\log^{\Theta(1)} n.
$$
Then, we apply the Poisson distribution and obtain for the probability that the number of vertices in $\mathcal{S}$ is exactly $1$ that 
$$
\Prob{|V \cap \mathcal{S}| = 1} = \E{|V \cap \mathcal{S}|} \cdot e^{-\E{|V \cap \mathcal{S}|}} \in n^{-o(1)},
$$
as set out by our second goal.

\paragraph{Goal 3.} For our third and final goal, note that we have
\begin{align}
    \mu(\B_u(R) \cap (\Psi_i \setminus (\Phi_i \cup \mathcal{S}))) \leq \mu(\B_u(R) \cap \B_0(R)) \leq \Theta(1)e^{\ell_0/2}/n,
\end{align}
using that $u \in V \cap \Layer{\ell_0}$ and \Cref{eq:layer_expected_degree}. Using the fact that the random variable of the number of vertices in $\B_u(R) \cap \B_0(R)$ follows a Poisson distribution and that $\ell_0 = \lceil \log\log n \rceil$, we then conclude that
$$
\Prob{ V \cap \B_u(R) \cap (\Psi_i \setminus (\Phi_i \cup \mathcal{S})) = \emptyset|\mathcal{E}_U} \geq e^{-n\cdot \mu(\B_u(R) \cap \B_0(R))} \in n^{-o(1)},
$$
as we sought to show for the third goal, and in conclusion, it follows that $\Prob{\mathcal{E}_i = 1 \text{ }|\text{ } \Phi_i \text{ is nice}} \in n^{-o(1)}$ as desired.
\end{proof}}
To finish the proof of our statement, recall that $\Psi_i$ has angle $\psi \in \Theta\left(\frac{n^{\frac{1}{2\alpha}}\log n}{n}\right)$ and thus, we have $k \in n^{\Omega(1)}$ such sectors in our entire disk. Now, let $X_i$ be the indicator random variable that is $1$ if the event $\{\mathcal{E}_i \cap \Phi_i$ is nice$\}$ occurs. Moreover, let $X:= \sum_{i=1}^k X_i$, i.e., the number of "buffer sectors" $\Psi_i$ that have the desired induced subgraph tree $T_i = G[V \cap \Phi_i]$. In the following, we obtain a lower bound for $X$. Recall that $r^* = R - \log n / \alpha - \log\log n /2$. Then, using linearity of expectation, \Cref{lem:tree-prob} and Claim~\ref{claim:indicator}  we have
\begin{align*}
\E{X\text{ }|\text{ }V \cap \B_0(r^*) = \emptyset} &= \sum_{i=1}^k \Prob{\mathcal{E}_i \cap \Phi_i \text{ is nice} \text{ }|\text{ }V \cap \B_0(r^*) = \emptyset} \\
&\geq k\cdot \Prob{\Phi_i \text{ is nice}\text{ }|\text{ }V \cap \B_0(r^*) = \emptyset}\cdot\Prob{\mathcal{E}_i \text{ }|\text{ } \Phi_i \text{ is nice}, V \cap \B_0(r^*) = \emptyset}\in n^{\Omega(1)}
\end{align*}
using that a nice sector according to \Cref{alg:nice-sector} has only vertices outside the area $\B_0(r^*)$ and $\mathcal{E}_i$ is independent of the event $V \cap \B_0(r^*) = \emptyset$. To obtain concentration for $X$, note that the sectors $\Psi_i$ are disjoint and the events that they are nice are mutually independent: determining whether $\Phi_i$ is nice only requires revealing the randomness within $\Phi_i$. In similar fashion, event $\mathcal{E}_i$ also only depends on the randomness in sector $\Psi_i \setminus \B_0(r^*)$. Therefore, a Chernoff bound applies and we obtain
\begin{align}\label{eq:X-concentration}
\Prob{X \in n^{\Omega(1)} \text{ }|\text{ }V \cap \B_0(r^*) = \emptyset} \in 1 - n^{-\omega(1)}.    
\end{align}
Next, we show for a sector $\Psi_i$ where event $\mathcal{E}_i$ occurs, that all neighbours of the nice sector $\Phi_i$ are contained in $\Psi_i$ if we condition on the event that $V \cap \B_0(r^*) = \emptyset$. Note that if $\Phi_i$ is nice, it holds for any vertex $u \in V \cap \Phi_i$ that $r(u) \geq R - \log\log n -2$. As such, the largest possible angle spanned among a vertex $u \in V \cap \Phi_i$ and any vertex $v \in V \cap \mathcal{D}_R\setminus\B_0(r^*)$ is
\begin{align}\label{eq:angle-neighbour-bound}
\theta_R(r(u), r(v)) \leq \theta_R(R - \log\log n - 2, \underbrace{R - \log n/\alpha - \log\log n /2}_{r^*}) \in \bigO\left(\frac{n^{\frac{1}{2\alpha}}\cdot \log^{3/4}n}{n}\right) \in o(\psi),   
\end{align}
by \Cref{lem:max-angle}. Hence, for any $u \in V \cap \Phi_i$ it holds that $N(u) \subseteq V \cap (\Psi_i\setminus \B_0(r^*))$.

In the next step, we bound the probability of the event that the area $\B_0(r^*)$ is empty (see hatched area in \Cref{fig:d-regular-trees}b). Note that by \Cref{lem:measure-inner-disk} the expected number of vertices in this area is $n\cdot\mu(B_0(r^*)) \in o(1)$. Hence, it holds via Poisson distribution that
\begin{align}\label{eq:empty-inner-disk}
   \Prob{V \cap \B_0(r^*) = \emptyset} \in e^{-o(1)} \in 1 - o(1).
\end{align}
Thus, using \Cref{eq:X-concentration}, \Cref{eq:empty-inner-disk}, and law of total probability we then get that
\begin{align}\label{eq:X-concentration2}
\Prob{X \in n^{o(1)}} \leq \Prob{V \cap \B_0(r^*) \neq \emptyset} + \Prob{V \cap \B_0(r^*) = \emptyset}\cdot \Prob{X \in n^{o(1)} \text{ }|\text{ }V \cap \B_0(r^*) = \emptyset} \in o(1).    
\end{align}
It follows that, \aas, non of the induced subgraph tree $T_i = G[V \cap \Phi_i]$ has any edge outside $\Psi_i$ using \Cref{eq:angle-neighbour-bound}.

To wrap things up, consider the event that $X \in n^{\Omega(1)}$ and let $\mathcal{E}_{\text{giant}}$ be the event that any vertex in layer $\overline{\ell} =\lceil 2\log\log n/(1-\alpha)\rceil$ is connected to the giant component of $G$. Note that the intersection of the two events implies our theorem if also $\B_0(r^*)$ contains no vertex, since event $\mathcal{E}_i$ ensures that $T_i$ is connected to the giant component in the desired way of the "in particular" statement of \Cref{pro:d-regular-trees}.

To bound the probability of event $\mathcal{E}_{\text{giant}}$ we observe the following: using \cite[Lemma 5.3]{Krohmer2016} any vertex $v \in V \cap \Layer{\overline{\ell}}$ is connected to the set of vertices in $V \cap \B_0(R/2)$ with probability $1-o(1)$. Moreover, by \cite[Theorem 1.4]{bfm-giant-15}, all vertices of the set $V \cap \B_0(R/2)$ are part of the giant component with probability $1-o(1)$. Thus, we conclude that also any $v \in V \cap \Layer{\overline{\ell}}$ is part of the giant component with probability  
\begin{align}\label{eq:giant}
\Prob{\mathcal{E}_{\text{giant}}} \in 1-o(1).    
\end{align}
Though, the three events $X \in n^{\Omega(1)}$, $\mathcal{E}_{\text{giant}}$ and $V\cap \B_0(r^*) = \emptyset$ are not independent, a union bound of the complements yields via \Cref{eq:empty-inner-disk}, \Cref{eq:X-concentration2} and \Cref{eq:giant}
\begin{align}\label{eq:wrap-up}
\Prob{X \in n^{\Omega(1)} \cap \mathcal{E}_{\text{giant}} \cap V \cap \B_0(r^*) = \emptyset} \in 1 - o(1).   
\end{align}
 Using \Cref{eq:wrap-up} then finishes the proof as the desired structural properties of a $d$-ary tree in a nice sector now follow from \Cref{lem:tree-geometry}.
\end{proof}
\subsection{Lower Bounds for MM and MIS (Proof of Theorem \ref{thm:lowerbound})}\label{sec:lower-bound-corollary}
\newcommand{\algohrg}[0]{A_{\text{hrg}}}
\newcommand{\algotree}[0]{A_{\text{tree}}}
In this section, we present a lower bound on the runtime of any distributed algorithm for solving MM and MIS on HRGs in the form of \Cref{thm:lowerbound}. Our proof strategy is as follows: first we show that an $r$-round randomised algorithm for MIS/MM on HRGs would imply an $r$-round randomised algorithm for MIS/MM on $d$-regular trees; this is addressed in \Cref{lem:algotree}. Then, we state a lemma for an established lower bound for MIS/MM on $d$-regular trees; see \Cref{lem:treeMISMMLowerbound} for the statement. By stacking these two lemmas together, we obtain the desired lower bound.

To accomplish our goal, we make use of the following tree substructure in HRGs, which follows from \Cref{pro:d-regular-trees} by setting the parameter $m(n) = \log\log n$.
\begin{corollary}[MIS and MM obstruction]\label{cor:d-reg-tree}
Let $G$ be a threshold hyperbolic random graph. Then, \aas there exists induced subgraph $T_{\text{hrg}}$ that is a balanced $d$-ary tree with degree $d_{\text{hrg}} \in \Theta(\log\log n)$, height $h_{\text{hrg}} \in \Theta(\log_d\log n)$ and $n'_{\text{hrg}} \in \Theta(\sqrt{\log n})$ vertices. Additionally, for each such tree $T_{\text{hrg}}$ there exists exactly one vertex $u \in V(T_{\text{hrg}})$ with a neighbour in $v \in V(G)\setminus V(T_{\text{hrg}})$ where $v$ is a vertex of the giant component of $G$. In particular, vertex $u$ has graph distance $h_{\text{hrg}}$ to all leaves of $T_{\text{hrg}}$.
\end{corollary}
The following lemma shows that if there existed a randomised algorithm that solves MIS/MM in $o(\log\log n / \log\log\log n)$ rounds on hyperbolic random graphs, then this would imply a randomised algorithm that solves MIS/MM in $o(\log\log n / \log\log\log n)$ rounds on balanced $d$-regular trees with degree $d \approx \log\log n$ and height $h \approx \log\log n /\log\log\log n$. 
\begin{lemma}[HRG-to-tree coupling]\label{lem:algotree}
Let $G\sim\hrg$ be a threshold hyperbolic random graph with the properties of induced trees $T_{\text{hrg}}$ as stated in \Cref{cor:d-reg-tree} and let $\algohrg$ be a \LOCAL algorithm that solves MIS (MM) in $r < h_{\text{hrg}}/200$ rounds with error probability $p(n)$ on the giant component of $G$. 

Then, there exist an $r$-round \LOCAL algorithm $\algotree$, which solves MIS (MM) on a balanced $d$-regular tree $T$ with degree  $d=d_{\text{hrg}} + 1$ and height $h = \lfloor h_{\text{hrg}}/100 \rfloor$ with error probability at most $2p(n)$.
\end{lemma}
\begin{proof} 
We aim to design an algorithm $\algotree$ for our infinite family of balanced $d$-regular trees with degree $d$ and height $h$, such that $\algotree$ has error probability at most $2p(n)$. 
To do so, we run algorithm $\algohrg$ on the vertices of $T$, pretending that $T$ is an HRG $G\sim\hrg$ where $n$ is suitably chosen such that $n'_{\text{hrg}}$, $d_{\text{hrg}}$, and $h_{\text{hrg}}$ satisfy the conditions of \Cref{cor:d-reg-tree}.
To prove that algorithm~$\algotree$ is well-defined, we consider its behaviour on a $d$-regular tree $T$ with $n'$ vertices, and also consider the giant component of an HRG $G$ that contains an induced  tree substructure $T_{\text{hrg}}$ as stated in \Cref{cor:d-reg-tree}. Note that $T_{\text{hrg}}$ contains $T$ as an induced subgraph, and to show that the execution of $\algotree$ is well-defined, we next consider a mapping of the vertices from $T$ to $T_{\text{hrg}}$.

Let $U\subseteq V(T)$ be the set of vertices that are at a distance at most $r$ from the leaves of $T$, and let $I=V(T)\setminus U$. Let $N_{r}[U],N_{r}[I]$ be the $r$-hop neighbourhood of $U,I$, respectively. For $r < h$, we will map each of these subgraphs from $T$ to $T_{\text{hrg}}$ such that the $r$-hop local view of each vertex $u\in V(T)$ is the same as the $r$-hop local view of a vertex $v\in V(T_{\text{hrg}})$. 

Since $N_{r}[U]$ contains the vertices up to distance $r$ from the leaves, if we assume that  $N_{r}[U]$ contains $k$ leaves,  we map the $k$ leaves of $N_{r}[U]$ to the $k$ ``leftmost'' leaves of $T_{\text{hrg}}$. After that, we map the rest of the vertices $N_{r}[U]$ to vertices $T_{\text{hrg}}$ so that the mapping preserves the $r$-hop neighbourhoods of the vertices in $U$. Let $M(U)$ be the set of vertices in $T_{\text{hrg}}$ that $N_{r}[U]$ was mapped to. With this partial mapping, we have guaranteed that if we run $\algotree$ on $T$, every vertex $u\in U$ has a copy vertex $v\in V(T_{\text{hrg}})$ with the same view. Additionally, by the choice of $r<h/2\leq h_{\text{hrg}}/200$, no vertex in $U$ can identify that they are not in an HRG $G$.

We will now map the vertices from $N_{r}[I]$. First, note that the diameter of the induced subgraph $T':= N_{r}[I]$ is the same as $T$, if we first remove the leaves of $T$. Therefore, the induced subgraph of $N_{r}[I]$ has height $h-1$. Now let $w\in V(T')$ be the unique vertex that is at distance $h-1$ from every leaf of $T'$. Also, among the vertices of $T_{\text{hrg}}$ that are part of the partial mapping of $N_{r}[U]$, let $v\in M(U)$ be one of the vertices that are furthest away from the leaves in $T_{\text{hrg}}$. Starting from $v$, we traverse $h-1$ hops towards the ``root'' of the $T_{\text{hrg}}$, and let $w'\in V(T_{\text{hrg})}$ be the vertex we arrive at. We map $w\in V(T')$ to $w'\in V(T_{\text{hrg}})$. After that, we map the rest of the vertices $N_{r}[U]$ to vertices $T_{\text{hrg}}$ so that the mapping preserves the $r$-hop neighbourhoods of the vertices in $I$. Let $M(I)$ be the set of vertices in $T_{\text{hrg}}$ that $N_{r}[I]$ was mapped to.
Note that by our choice of $w$, $w'$, and $r$, every vertex in $M(I)$ is neither a leaf nor a root. This guarantees that the vertices in $I \subset V(T)$ cannot distinguish that they are not on an HRG $G$ after the execution of $\algotree$. Additionally, every vertex in $N_{r}[I]$ has an identical $r$-hop local view to its mapped vertex in $M(I)$.

We now proceed by bounding the error probability of $\algotree$. To this end, we utilise the mapping we laid out above. Since MIS and MM are locally checkable problems, if $\algotree$ fails on $T$, then there is a local witness of failure centred either at a vertex of $U$ or at a vertex of $I$. Denote these two events by $F_U$ and $F_I$. Thus, by union bound, we have
\begin{align*}
\Pro{\algotree \text{ fails on } T}
\le \Pro{F_U}+\Pro{F_I}.
\end{align*}
We first bound $\Pro{F_U}$. To achieve this, we couple the randomness of the algorithm $\algohrg$ on the vertices $N_r[U]$ with the randomness of the algorithm $\algotree$ on the mapped vertices in $M(U)$. Since $r\le h(n)/200$, no vertex in $M(U)$ ``sees'' the unique attachment vertex of $T_{\text{hrg}}$ to $G\setminus T_{\text{hrg}}$. Hence, the radius-$r$ views used by $\algotree$ in $T_{\text{hrg}}$ and the radius-$r$ views seen by $\algohrg$ in $G$ are identical for all vertices relevant to a failure witness centred in $U$. Therefore, under this coupling, every occurrence of $F_U$ yields a local MIS, respectively, MM, violation in the output of $\algohrg$ on $G$. Consequently, by the hypothesis of our lemma statement, it holds
\begin{align*}
\Pro{F_U}\le \Pro{\algohrg \text{ fails on } G}\le p(n).
\end{align*}
It remains to bound $\Pro{F_I}$. For vertices in $I$, the relevant $N_r[I]$ are far from the leaves. We couple the algorithm $\algotree$ for $N_r[I]$ with the algorithm $\algohrg$ for the mapped vertices in $M(I)$ that are also far from both the leaves and the unique attachment point to $G\setminus T_{\text{hrg}}$. Thus, all radius-$r$ views used by $\algotree$ around a potential failure witness in $I$ are identical to the corresponding views of $\algohrg$ in an execution on $G$. Hence every occurrence of $F_I$ also yields a local violation of the output of $\algohrg$ on $G$, and therefore again
\begin{align*}
\Pro{F_I}\le \Pro{\algohrg \text{ fails on } G}\le p(n).
\end{align*}
Combining the two bounds gives
$\Pro{\algotree \text{ fails on } T}
\le \Pro{F_U}+\Pro{F_I}
\le 2p(n).$
Equivalently, if $\algotree$ failed with probability greater than $2p(n)$, then either $F_U$ or $F_I$ would occur with probability greater than $p(n)$, contradicting the correctness guarantee of $\algohrg$ under the corresponding coupling.
\end{proof}
Next, we show a statement which will imply that any randomised algorithm on a balanced $d$-regular tree with degree $d \approx \log\log n$ and height $h \approx \log\log n / \log\log\log n$ requires $\Omega(\log\log n / \log\log\log n)$ rounds to solve MIS and MM respectively.
\begin{lemma}[\cite{BBKO22}]
\label{lem:treeMISMMLowerbound}
The randomised complexity with error probability $p$ of MIS/MM on a balanced $d$-regular tree with $n$ vertices is $\Omega(\min\{d, \log_d n, \log_d \log {1/p}-\bigO(1)\})$.
\end{lemma}
\begin{proof}[Proof sketch:]
The result directly follows with \Cref{thm:sevenOne} (see Appendix~\Cref{sec:sevenOne} for a discussion), given that \cite{BBKO22} show that the respective sequences to apply the theorem exist for MIS and MM. More detailed,  \cite{BBKO22} shows that such a sequence with  $t=\Theta(d)$ and $f(d)=2^{d+1}$ exists for MIS. The main idea of the proof is that they construct a sequence and each problem in their sequence come with a fingerprint vector $z$. For MIS the vector is of length $2$ (formally $len(z)=1$ in their terminology) and the fingerprint vector of the original MIS problem is $[1,0]$. Then, fingerprint vectors can be obtained iteratively. The $j$-entry of the fingerprint for the $i$-th problem in the sequence is the $j$-prefix sum of the $i-1$-st fingerprint. Thus, the fingerprint vectors evolve like $[1,0]$, $[1,1]$,$[1,2]$, until $[1,t]$. They show that the sequence can be extended by yet another problem whenever the $l1$-norm of the vector is $\leq \Delta$. So, for MIS we can do $\Delta-1$ such steps. The precise interpretation of the fingerprint vector and the actual description of the problems in the sequence is technical and we refer to their paper for more details. We obtain that \cite[Lemma 6.1]{BBKO22} shows that the sequence for MIS is of length $\Theta(d)$ with $f(d)\leq 2^d(1+len(z))$ where $len(z)=1$ for the MIS problem (in their paper $len(z)=\beta$ which equals $1$ for MIS, see Section 5.1 in the arxiv version of their work.  
Thus, via \Cref{thm:sevenOne} the runtime bound for error probability $p$ for MIS simplifies to $\Omega\left(\min\left\{d,\log_{d}n, \log_{d}\log 1/p-\log_{d}\log 2^{d+1}\right\}\right)$. In particular, the term $\log_{d}\log 2^{d+1} = \log_{d}({d+1}) \in \bigO(1)$.

For MM \cite{BO20} (see e.g. Theorem 4.5) shows that the respective sequence with length $t=2d-1$ and $f(d)=5$ exists. For the latter, we cite from their work: \emph{Crucially, all the problems of the family are described using only $5$ labels, and while the result of the round elimination technique may contain more than $5$ labels, we will provide relaxations that allow us to map these problems back to this family} exists, and hence, \Cref{thm:sevenOne} yields a lower bound of  $\Omega\left(\min\left\{2d-1,\log_{d}n, \log_{d}\log 1/p-\log_{d}\log 5\right\}\right) \in \Omega\left(\min\left\{d,\log_{d}n, \log_{d}\log 1/p-\bigO(1)\right\}\right)$.
\end{proof}
We now put \Cref{cor:d-reg-tree}, \Cref{lem:algotree} and \Cref{lem:treeMISMMLowerbound} together to obtain our lower bounds of \Cref{thm:lowerbound}.
\theoremlowerbound*
\begin{proof}
Using \Cref{cor:d-reg-tree} in conjunction with \Cref{lem:algotree}, it follows \aas over the draw of $G\sim\hrg$ that any $o(\log\log n / \log\log\log n)$-round randomised algorithm for MIS or MM on $G$ that has an error probability of at most $p$ would imply the existence of a $o(\log\log n / \log\log\log n)$-round randomised algorithm for MIS or MM on balanced $d$-regular trees, where the degree is $d \in \Theta(\log\log n)$ and the tree has height $h\in \Theta(\log\log n / \log\log\log n)$, with error probability at most $2p$. Therefore, it suffices to show that no such algorithm exists on any such $d$-regular tree with running time $o(\log\log n / \log\log\log n)$ and error probability at most $2/n^{c}$, for every constant $c>0$. Indeed, this immediately implies that \aas over the draw of $G$, there does not exist an $o(\log\log n / \log\log\log n)$-round randomised algorithm on $G$ for MIS or MM with error probability $p = 1/n^{c}$. Consequently, to establish our main result, it remains to prove that no such algorithm exists for balanced $d$-regular trees.

We show this by using \Cref{lem:treeMISMMLowerbound}. Fix any constants $c, c_d, c_h > 0$ and consider a balanced $d$-regular tree $T$ with degree $d = c_d \cdot \log\log n$ and height $h = c_h \cdot \log_d\log n$ giving a $d$-regular tree of size $n' \geq (d-1)^{h} \geq \log^{c_n} n$ for some constant $c_n > 0$. Since the runtime for any algorithm with error probability $p'$ for MIS or MM on a balanced $d$-regular tree with $n'$ vertices is 
$$\Omega(\min\{d, \log_d n', \log_d \log {1/p'}-\bigO(1)\})$$
using \Cref{lem:treeMISMMLowerbound}, the runtime $t$ for our tree $T$ with parameters $d = c_d \cdot \log\log n$, $h = c_h \cdot \log_d\log n$ and $n' \geq \log^{c_n} n$ is 
$$t\in \Omega\left(\min\left\{\log\log n, \frac{\log\log n}{\log\log\log n},\frac{\log\log n}{\log\log\log n}\right\}\right)$$
setting $p'=2/n^c$. Hence, $t \in \Omega\left(\frac{\log\log n}{\log\log\log n}\right)$ and the theorem follows as this implies that \aas for $G$, an algorithm requires at least $t$ rounds to obtain a probabilistic guarantee of $1 - n^{-c}$ for MIS or MM for the giant component of $G$.
\end{proof}
\section{Embedding-Aware Symmetry Breaking (Theorem \ref{thm:geometric})} \label{sec:geometric}
In this section, we assume that each vertex knows its geometric position. So a vertex $v \in V(G)$ can also communicate its geometric coordinate $(r(v), \varphi(v))$ to any neighbour within one round. We prove the following.
\theoremgeometric*
The Maximal Independent Set part we show in \Cref{sec:mis_geom} and the Maximal Matching part in \Cref{sec:mm_geom}. Throughout the section we assume that radial and angular coordinates are unique to vertices; a property that holds almost surely.
To describe our algorithms in both sections, we make use of a tiling which ensures that vertices in the same tile form a clique. Moreover, for any two vertices $u$ and $v$ in the same layer, if the tiles containing $u$ and $v$ are separated by sufficiently many intermediate tiles, then there is no edge between $u$ and $v$.
 This is useful since all vertices in a tile can communicate with each other within one round, finding a local solution that is globally valid if active tiles are far enough apart.

\paragraph{Tiling.} Throughout this section, let $\ell \in [\lceil \log n / \log\log n \rceil]$. Recall the definition of a layer $$\Layer{\ell}:= \B_0(R-\ell)\setminus \B_0(R-\ell-1),$$ and for layer level $\ell$, let
$$N_\ell := 40\cdot \left\lceil \frac{n \cdot \pi e^{C/2}}{20e^{\ell}}\right\rceil,$$
where the constant $40$ is chosen for later convenience.  Then, for any $\ell$ and $i \in [N_\ell]$ we define a \emph{tile} $\tiletile{\ell}{i}$ by the set of points 
\begin{align}\label{eq:tile}
    \tiletile{\ell}{i} := \{x \in \Layer{\ell} : (i\cdot 2\pi)/N_{\ell} \leq \varphi(x) < ((i+1)\cdot 2\pi)/N_{\ell} \}.
\end{align}
For a sketch of the tiling, see also \Cref{fig:onion}a. We remark that $N_\ell$, the number of tiles per layer, is divisible by $40$.
\begin{remark}\label{rem:divisible}
    For any $\ell$ it holds 

    $$
    N_{\ell} \equiv 0 \pmod{40}.
    $$
\end{remark}
We sum up the key properties of our tiling in the following lemma. The first property tells us that vertices in the same tile form a clique, and the second property tells us that vertices in the same layer but with a constant number of tiles between them do not have an edge. Finally, the third property says that if a vertex $u$ has a neighbour $v$ which is located in a tile in a ``smaller'' layer, then $v$ is not a neighbour to another vertex $u'$ in the same layer as $u$, if the number of tiles between $u$ and $u'$ is a large enough constant. 
\begin{lemma}[Tiling properties]\label{lem:tiling}
    The following holds for any $\ell$ of our tiling:

    \begin{enumerate}
  \renewcommand{\labelenumi}{\roman{enumi}.}
  \item\label{item:one} for any $i \in [N_\ell]$, it holds for any $u \in V \cap \tiletile{\ell}{i}$ and $v \in V \cap \tiletile{\ell}{i}$ that $\{u,v\} \in E(G)$;
  \item\label{item:two} for any pair $i, j$ such that $|i-j| \geq 20$, it holds for any $u \in V \cap \tiletile{\ell}{i}$ and $v \in V \cap \tiletile{\ell}{j}$ that $\{u,v\} \not\in E(G)$;
  \item\label{item:three} let $\ell \geq \ell'$ and consider any pair $i,j$ such that $|i-j| \geq 40$. Moreover, let $u \in V \cap \tiletile{\ell}{i}$, $u' \in V \cap \tiletile{\ell}{j}$ and let $v \in N(u) \cap \Layer{\ell'}$. Then, it also holds  $\{u',v\} \not\in E(G)$.
\end{enumerate}
\end{lemma}
\begin{proof}
\Cref{item:one}: W.l.o.g. let $R - \ell - 1 < r(u) \leq r(v) \leq R - \ell$. Then, using that $N_\ell = 40\cdot \left\lceil \frac{n \cdot \pi e^{C/2}}{20e^{\ell}}\right\rceil  \geq \frac{2 n \cdot \pi e^{C/2}}{e^{\ell}}$, it holds for the angular distance between $u, v$ that
\begin{align*}
    \angulardist{u}{v} &\leq 2\pi/ N_\ell
    \leq \frac{e^{\ell}}{ne^{C/2}} 
    \leq \theta_R(R- \ell, R - \ell)
    \leq \theta_R(r(u), r(v)),
\end{align*}
by applying the lower bound of \Cref{lem:max-angle} and  $\theta_R(\cdot, \cdot)$ is monotonically decreasing in both arguments (see e.g. \cite[Remark 4]{km-slcrhg-19}). Hence, $u$ and $v$ have an edge as desired.

\smallskip

\Cref{item:two}: W.l.o.g. let $R - \ell - 1 < r(u) \leq r(v) \leq R - \ell$. Then, using that there are at least $19$ tiles in between $u$ and $v$ and that $N_\ell = 40\cdot \left\lceil \frac{n \cdot \pi e^{C/2}}{20e^{\ell}}\right\rceil  \leq \frac{2 n \cdot \pi e^{C/2}}{e^{\ell}} + 1$, it holds for the angular distance between $u, v$ that
\begin{align*}
    \angulardist{u}{v} &\geq 38\pi/ N_\ell
    \geq \frac{18 e^{\ell}}{ne^{C/2}}
    > \frac{\pi e^{\ell+1}}{ne^{C/2}}
    \geq \theta_R(R- \ell - 1, R - \ell - 1)
    \geq \theta_R(r(u), r(v)),
\end{align*}
by applying the upper bound of \Cref{lem:max-angle} and the monotonicity of $\theta_R(\cdot, \cdot)$. Hence, $u$ and $v$ have no edge as desired.

\smallskip

\Cref{item:three}:
Since $v \in N(u)$, it holds 
$$
\theta_R(r(u), r(v)) \leq \pi e^{(R-r(v) -r(u))/2},
$$
using \Cref{lem:max-angle}. W.l.o.g. let $u$ have angular coordinate $\varphi(u) = 0$ such that 
$$\varphi(v) \leq \frac{\pi \cdot e^{\ell+1-C/2}}{n}, $$
using that $\theta_R(\cdot, \cdot)$ is monotonically decreasing in both arguments and $r(u), r(v) \geq R - \ell - 1$. 

On the other hand, we have 
$$\varphi(u') \geq 78\pi/ N_\ell \geq \frac{38 e^{\ell}}{ne^{C/2}} ,$$
using that $\varphi(u) = 0$ and between $u$ and $u'$ there are at least $39$ tiles as $|i-j|\geq 40$.

Combining these bounds on the angle, we get for the angular distance between $u'$ and $v$ that
$$
\angulardist{u'}{v} = |\varphi(u') - \varphi(v)| \geq \frac{38 e^{\ell}}{ne^{C/2}} - \frac{\pi \cdot e^{\ell+1-C/2}}{n} > \frac{10 e^{\ell}}{ne^{C/2}} > \frac{\pi \cdot e^{\ell+1-C/2}}{n} > \theta_R(r(u'), r(v)),
$$
using the monotonocity of $\theta_R(\cdot, \cdot)$, $r(u'), r(v) \geq R - \ell - 1$ and \Cref{lem:max-angle}. Subsequently, $v \not\in N(u')$ which is what we sought to prove. 
\end{proof}
\subsection{Embedding-Aware Maximal Independent Set (MIS part of Theorem \ref{thm:geometric})}\label{sec:mis_geom}
We now exploit the knowledge of the coordinates, which implies the knowledge of which tile a vertex belongs to, in order to find an independent set in an HRG.
We use the geometry and the coordinates in an algorithm we refer to as \algomis, for which we first give an informal description. 

We iterate through the layers in bottom-up fashion, starting with layer $\Layer{0}$ as follows: we ``activate'' tiles in the current layer such that ``active tiles'' are far enough apart from each other so that vertices in different active tiles do not have an edge. Thus, in parallel, we can then select one vertex in every active tile to participate in the independent set (as long as none of its neighbours already participates in the independent set). Since vertices in the same tile form a clique, this removes all vertices from a tile. We then iterate in this fashion through all tiles of a layer such that all vertices in a layer either participate in the independent set or have a neighbour that is part of the independent set. The key observation here is that we only need to iterate through $\bigO(\log\log n)$ layers, since after iterating through layer $\Layer{0}$, every vertex with radius at most $R - 4\log\log n$ has a neighbour that participates in the independent set (this is addressed in \Cref{lem:onion-mis-1}).   

Indeed, we shall show in this section that this procedure, which is formalised in \Cref{algo:onion-mis}, gives a maximal independent set within $\bigO(\log\log n)$ rounds. 
\begin{algorithm}
\caption{\algomis}
\begin{algorithmic}[1]\label{algo:onion-mis}

\State \textbf{Input:} Graph $G = (V,E)$
\State \textbf{Output:} Independent set $I \subseteq V$

\State $I \gets \emptyset$ \text{// Initialise independent set}

\For{$\ell = 0$ to $\lceil 4\log\log n \rceil$} \text{// Iterate through layers bottom-up}

    \For{$m = 0$ to $39$} \text{// Activate every 40th tile in layer $\ell$}

        \State \textbf{(Executed in parallel at all nodes $u \in V$)}

        \State $U \gets \emptyset$ \text{// Deactivate all previously active vertices}

        \State \textbf{Activation step}
        \ForAll{$u \in V$}
            \If{$u \in V \cap \tiletile{i}{j}$ such that $i \equiv m \ (\mathrm{mod}\ 40)$} \text{// Check if vertex $u$ is in active tile}
                \If{$N(u) \cap I = \emptyset$} \text{// Check if vertex $u$ has no neighbour in the independent set}
                    \State $U \gets U \cup \{u\}$ \text{// Activate vertex $u$} 
                \EndIf
            \EndIf
        \EndFor

        \State \textbf{Communication step}
        \ForAll{$u \in U$}
            \State send $\varphi(u)$ to all $v \in N(u) \cap U$ \text{// Send angular coordinate to all active neighbours}
        \EndFor

        \State \textbf{Selection step}
        \ForAll{$u \in U$}
            \State let $v_1 \prec v_2 \prec \dots$ be $N(u) \cap U$ sorted by $\varphi(v_i)$ \text{// Ordering by angular coordinates}
            \If{$u = v_1$} \text{// Select vertex $u$ with smallest angular coordinate}
                \State $I \gets I \cup \{u\}$ \text{// Add vertex $u$ to independent set}
            \EndIf
        \EndFor

    \EndFor

\EndFor

\State \textbf{Return} $I$

\end{algorithmic}
\end{algorithm}
We first observe that every 40-th tile in layer $\Layer{0}$ contains a vertex that is part of the independent set (if there is at least one vertex in this tile). 
\begin{observation}\label{obs:tile-layer0}
If $|V \cap \tiletile{0}{i}| \geq 1$ such that $i \equiv 0 \pmod{40}$, then $V \cap \tiletile{0}{i}$ contains a vertex that is in the set $I$.
\end{observation}
\begin{proof}
    Throughout the proof, we call a tile that contains at least one active vertex an active tile. Note that the tiles for which we wish to show the desired property are the tiles that are active in \algomis if $\ell, m = 0$, i.e., the first active tiles in the algorithm such that $I = \emptyset$.

    We fix any active tile, i.e., a tile $\tiletile{0}{i}$ such that $i \equiv 0 \pmod{40}$. Since all vertices that share the same tile form a clique by \Cref{lem:tiling}~\Cref{item:one}, in the \emph{Selection step}, every active tile with at least one vertex has exactly one vertex $u$ that tries to be part of the set $I$. So it is sufficient to show that $u$ has no active neighbour in any other active tile. By the \emph{Activation step} and using \Cref{rem:divisible}, it follows that between two active tiles there are at least $39$ non-active tiles between two active tiles. Thus, since any other active tile is in the same layer as the tile $\tiletile{0}{i}$, any vertex $v$ that is included in any other active tile is not a neighbour of $u$ by \Cref{lem:tiling}~\Cref{item:two}. 
\end{proof}

We make use of this observation to show that, after iterating through layer $\Layer{0}$, vertices with radial coordinate larger than $R - 4\log\log n$ have a neighbour that is part of the independent set.
\begin{lemma}[First iteration of \algomis]\label{lem:onion-mis-1}
    When $\ell \geq 1$ in \algomis, any vertex $u \in V$ with radial coordinate $r(u) \geq R - 4\log\log n$ has a neighbour that is included in the set $I$ with probability $1 - n^{-\omega(1)}$. 
\end{lemma}
\begin{proof}
Using Observation \ref{obs:tile-layer0}, it is sufficient to show that for any vertex $u$ with radius $r(u) \geq R - 4\log\log n$, there exists at least one neighbour $v \in N(u) \cap \tiletile{0}{i}$ such that $i \equiv 0 \pmod{40}$.

To show this, fix any vertex $u$ with radius $r(u) \geq R - 4\log\log n$ and note that, since $\theta_R(r(u), R) \in \Omega(\log^2(n)/n)$ (using \Cref{lem:max-angle}) and the angular width of a tile $\tiletile{0}{i}$ is at most $\Theta(1/n)$ (by \Cref{eq:tile}), there are at least $\Omega(\log^2 n)$ tiles $\tiletile{0}{i}$ such that $i \equiv 0 \pmod{40}$ in the neighbourhood disk $\B_u(R)$. Hence, it is sufficient to show that at least one of the $\Omega(\log^2 n)$ tiles contains any vertex. Since tiles do not overlap and $\mu(\tiletile{0}{i}) = \mu(\Layer{0})/{N_0} \in \Omega(1/n)$ (using \Cref{eq:layer_measure} and angle $1/N_0 \in \Theta(1/n)$ of a tile $\tiletile{0}{i}$), we conclude that the expected number of vertices in distinct active tiles (tiles $\tiletile{0}{i}$ tiles such that $i \equiv 0 \pmod{40}$) in the neighbourhood disk $\B_u(R)$ is $\Omega(\log^2 n)$. Since the number of vertices in each tile follow a Poisson-distribution, a Chernoff-bound yields that there is at least one active tile in the neighbourhood radius of $u$ with probability $1 - n^{-\omega(1)}$. Hence,$u$ has a neighbour $v \in N(u) \cap \tiletile{0}{i}$ that participates in the set $I$ with probability $1 - n^{-\omega(1)}$ and a union bound over all vertices gives the desired result. 
\end{proof}
Next, we show that \algomis produces an MIS on the vertices that were not removed by the previous lemma. 
\begin{lemma}[\algomis gives MIS]\label{lem:onion-mis-2}
   After \algomis has terminated, any vertex $u \in V$ with radial coordinate $r(u) \leq R - 4\log\log n$ is either in the set $I$ or has a neighbour $v \in N(u)$ that is in the set $I$.
\end{lemma}
\begin{proof}
    Note that if for all $\ell \in [\lceil 4\log\log n\rceil + 1]$ it holds that after iteration $\ell$ all vertices up to radius $R - \ell$ are either part of the independent set or have a neighbour that is in the independent set, we are done.
    
We prove this by induction. For $\ell = 0$, observe that for any $m$ in~\Cref{algo:onion-mis} that for any pair of active vertices $u,v$, either (a) $u$ and $v$ lie in the same tile, or (b) there are at least $39$ tiles between $u$ and $v$. This follows from the Activation step of \algomis together with \Cref{rem:divisible}. Hence, by \Cref{lem:tiling}, every active vertex belongs to a connected component that forms a clique among active vertices in the same tile. Thus, for any active tile, we can select any vertex $u$ to be included in the $I$, and all other vertices in the same tile will have $u$ as a neighbour that is included in the set $I$. This is achieved by the \emph{Selection step} in \algomis and as such, for $\ell = 0$, all vertices $V \cap \Layer{\ell = 0}$ are either part of the set $I$, or they have a neighbour that is in the set $I$ since every tile gets processed exactly once. Thus, after iteration $\ell = 0$, the set $I$ is a maximal independent set for $G[V \cap \Layer{0}]$ since any tile in $\Layer{0}$ is exactly once active.

For the induction step, fix any $\ell$ and suppose that, by the induction hypothesis, after processing layers $0,\dots,\ell-1$, the set $I$ is a maximal independent set of
$
G\left[V \cap \left(\disk \setminus \B_0(R-\ell-1)\right)\right].
$
Consider iteration $\ell$ of \Cref{algo:onion-mis}. Any vertex $u \in V_\ell$ that already has a neighbour in $I$ is not activated by the algorithm and is therefore already dominated. Hence, it remains to consider the active vertices, namely those in $V_\ell \setminus N(I)$.

On the induced subgraph $G[V_\ell \setminus N(I)]$, the same arguments as in the base case apply. For any fixed value of $m$, active vertices are either contained in the same tile or are separated by at least $39$ tiles. Thus, by \Cref{lem:tiling}, active vertices in different active tiles are non-adjacent, while active vertices within the same tile form a clique. Consequently, the \emph{Selection step} chooses exactly one active vertex from every non-empty active tile, and every other active vertex in that tile has a neighbour that is added to $I$.

Therefore, after iteration $\ell$, every vertex in $V_\ell$ either already had a neighbour in $I$ before the iteration or is itself added to $I$ or has a neighbour that is added to $I$ during the iteration. Hence, after processing every tile in layer $\ell$ exactly once, every vertex up to radius $R-\ell$ is either contained in $I$ or has a neighbour in $I$, completing the induction.
\end{proof}
The following algorithm gives us a bound on the number of rounds for \algomis
\begin{lemma}[\algomis Runtime]\label{lem:onion-mis-run}
    The algorithm \algomis requires $\bigO(\log \log n)$ rounds \whp for \CONGEST.
\end{lemma}
\begin{proof}
This follows from noting that \algomis iterates through $\bigO(\log\log n)$ layers and $m \in \bigO(1)$ in line 4 of \Cref{algo:onion-mis}. Moreover, each iteration there is only one communication step where vertices communicate their angular coordinate in the \emph{Communication step}. Though angular coordinates are a real number, it suffices to send $\bigO(\log n)$ bits \whp by \Cref{lem:truncation} to obtain an ordering that is required for the \emph{Selection step}. Hence \ $\bigO(1)$ communication rounds per iteration are sufficient \whp for \CONGEST. Thus, we obtain that the number of rounds is $\bigO(\log\log n)$ \whp as claimed.
\end{proof}
We finish the section by using the established results to show that \algomis produces a maximal independent set in $\bigO(\log\log n)$ rounds.
\begin{proposition}[Embedding-aware Maximal Independent Set]
    There exists a deterministic $\bigO(\log\log n)$-round \CONGEST algorithm to find a Maximal Independent Set on HRGs \whp if the vertices have access to their geometric coordinates.
\end{proposition}
\begin{proof}
Note that set $I$ fulfils the properties of an MIS and thus, the result that \algomis produces an MIS follows from \Cref{lem:onion-mis-1} and \Cref{lem:onion-mis-2}. The runtime for \CONGEST follows from \Cref{lem:onion-mis-run}.
\end{proof}
\subsection{Embedding-Aware Maximal Matching (MM part of Theorem \ref{thm:geometric})}\label{sec:mm_geom}
In this section, we use geometry to find a maximal matching for a hyperbolic random graph in $\bigO(\log\log\log n)$ rounds. We do this as follows:

First, in a preprocessing step, we match all vertices with degree at least $\approx \log^3(n)$ with a vertex in layer $\Layer{0}$. This is accomplished by only marking vertices with degree $\lceil\log^{3/2} n\rceil$ and letting them draw one random neighbour and mark this incident edge. Thereafter, we consider all vertices with radius at most $R - 6\log\log n$ and show that all vertices such vertices are matched by this preprocessing step.  This is shown in \Cref{lem:match-inner-disk}.

Thereafter, all that is left to do is to match all unmatched vertices which have radius at least $R - 6\log\log n$. We partition the outer disk $\disk \setminus \B_0(R - 6\log\log n)$ into $6\log\log n$ layers. We then apply our tiling with the properties shown in \Cref{lem:tiling} to obtain for each induced subgraph of a layer a maximal matching in constant rounds, which can be done for each layer in parallel (see also \Cref{lem:matching-in-layer}).

Next, we use a divide-and-conquer approach to find a maximal matching for the outer disk. It can be shown that, if we have two annuli $\mathcal{A}_1$ and $\mathcal{A}_2$ for which we have each a maximal matching for the induced subgraphs, then we can extent these two matchings to a maximal matching for the induced subgraph of vertices in $\mathcal{A}_1 \cup \mathcal{A}_2$ in constant rounds. We show this in \Cref{lem:block-merging}.

Using \Cref{lem:matching-in-layer} as a base case, and using that by \Cref{lem:block-merging} we can ``merge'' the matchings of two annuli in constant rounds, we then merge pairs of annuli in parallel in constant rounds. By this, we can double the size of the induced subgraph of an annulus that is matched in constant rounds, and we can do this for all disjoint annuli in parallel. Thus, since the outer disk consists of $\bigO(\log\log n)$ layers, after $\bigO(\log\log\log n)$ merging iterations, each of which takes constant rounds, we have found a maximal matching for the outer disk in $\bigO(\log\log\log n)$ rounds. 

The following shows that we can compute a maximal matching for an induced subgraph of a layer in constant rounds.
\begin{lemma}[Maximal matching within a layer]\label{lem:matching-in-layer}
     Let $G$ be a threshold hyperbolic random graph where vertices are given their geometric coordinates as an input. Let $G_\ell := G[V_\ell]$. Then for all layers we can compute a maximal matching $M_{\ell} \subseteq E(G_\ell)$ for $G_\ell$ in $\bigO(1)$ rounds \whp for \CONGEST.
\end{lemma}
\begin{proof}
    We use the following claim, which says that in constant rounds we can compute a matching such that each tile\footnote{Recall \Cref{eq:tile} for the definition of our tiling.} contains at most one unmatched vertex\footnote{We call a vertex $u$ unmatched if $\exists v \in N(u)$ such that $E(v) \cap M = \emptyset$ and $E(u) \cap M = \emptyset$, i.e., under current matching $M$, it is still possible to add an edge $e$ incident to $u$ such that $M \cup e$ is a valid matching.}.
    \begin{claim}\label{claim:tile-single-vertex}
         For any $\ell$ let $i \in [N_\ell]$ and let $V'\subseteq V$ be a set of unmatched vertices. Moreover, let $G' := G[V' \cap \tiletile{\ell}{i}]$. Then we can compute a matching $M' \subseteq E(G')$ with probability $1 - n^{-10}$ for \CONGEST in $\bigO(1)$ rounds, such that there exists at most one vertex in $V'$ that is not incident to an edge of $M'$.
    \end{claim}
    {
\renewcommand{\qedsymbol}{$\blacksquare$} 
\begin{proof}[Proof of claim] Consider any tile  $\tiletile{\ell}{i}$ and let $U:= V' \cap \tiletile{\ell}{i}$. In constant communication rounds, any vertex $u \in U$ can communicate the tile $\tiletile{\ell}{i}$ based on its radial and angular coordinate to all its neighbours $N(u)$. Thus, after constant rounds, every vertex $u$ has learned all vertices of $U$ (the set of vertices located in the same tile as $u$) since the induced subgraph of a tile forms a clique (\Cref{lem:tiling}). Next, every vertex $u$ sends the first  $\lceil 13 \cdot \log n \rceil$ bits of $\varphi(u)$ to all neighbour. Then, using \Cref{lem:truncation}, it holds with probability $1-n^{-10}$ that every vertex is able to compute an ordering $v_1 \prec v_2 \prec \dots$ which is the set $U$, ordered by angular coordinates $\varphi(v_i)$. Then, consider $u = v_{j}$ in this ordering. If it has an odd index $j$ based on our ordering, we add the edge $\{\{u,v_{j+1}\}\}$ to the matching $M'$. We now distinguish two cases which show that by this procedure, at most one vertex remains unmatched within a tile. 

\textbf{Case 1}~[Number of vertices in a tile $\tiletile{\ell}{i}$ is even]: all vertices can be paired with any neighbour in the same tile $\tiletile{\ell}{i}$ using \Cref{lem:tiling}. By our procedure of pairing odd-even pairs, every vertex is matched. 

\smallskip

\textbf{Case 2}~[Number of vertices in a tile $\tiletile{\ell}{i}$ is odd]: Observe that the vertex $u$ with the largest angular coordinate is the only vertex in $\tiletile{\ell}{i}$ after our pairing procedure that has no matching partner. Thus, there is exactly one vertex which is unmatched.

\smallskip

Since our described procedure requires constant rounds and vertices communicate in any round at most $\bigO(\log n)$ bits, the claim follows from the above two cases and that angular coordinates used are unique with probability $1 - n^{-10}$.
\end{proof}}
Now, applying Claim \ref{claim:tile-single-vertex} to each tile of layer $\ell$ in parallel, after constant rounds, there is at most $1$ unmatched vertex in each tile $\tiletile{\ell}{i}$. Given this property, we iterate through $\Layer{\ell}$ as follows: for $t \in [40]$ we attempt to match any possible unmatched vertex $u \in V \cap \tiletile{\ell}{i}$ where  $i \equiv t \ (\mathrm{mod}\ 40)$, i.e., we ``activate'' every 40-th tile. We do so by first letting every unmatched vertex $v$ send its layer $\ell \in \bigO(\log n)$ to every neighbour $N(v)$ in constant rounds. Thereafter, we activate any vertex $u \in V \cap \tiletile{\ell}{i}$ where $i \equiv t \ (\mathrm{mod}\ 40)$. Let $N'(u)$ be the set of unmatched neighbours of $u$ and $u$ marks the set of unmatched vertices $N'(u) \cap V_{\ell}$. Then $u$ picks a vertex $v$ uniform at random from $N'(u) \cap V_{\ell}$ and \{u,v\} is added to $M_{\ell}$. Since in each iteration $t$, we have for any pair of active vertices $u  \in V \cap \tiletile{\ell}{i}$ and $u' \in V \cap \tiletile{\ell}{j}$ that $|i-j|\geq 40$, it follows by \Cref{lem:tiling} that $N(u) \cap N(u') = \emptyset$. Therefore, no conflict can occur: no unmatched vertex can receive matching requests from two distinct active vertices during the same iteration phase $t$. Since after $40$ iterations we tried to match any unmatched vertex in $\Layer{\ell}$ and each iteration requires $\bigO(1)$ rounds while also $t \in \bigO(1)$, the procedure terminates after constant rounds. Moreover, $M_{\ell}$ is maximal for $G_{\ell}$: when a tile $\tiletile{\ell}{i}$ is processed, the unique unmatched vertex of tile $\tiletile{\ell}{i}$ is matched whenever it has an unmatched neighbour. Therefore, after tile $\tiletile{\ell}{i}$ has been processed, either it has become matched or every neighbour remaining in later tiles is already matched (or non-existent). Hence, using that each tile gets processed exactly once, no edge can remain between two unmatched vertices and we obtain a maximal matching $M_{\ell}$ for $G_{\ell}$ as desired.  
\end{proof}
In the following, let $\block{i}{k} := \bigcup_{\ell=i}^{k+i-1}\Layer{\ell}$. We refer to $\block{i}{k}$ as an \emph{annulus of size $k$} with \emph{starting layer $i$}. That is, an annulus of size $k$ with starting layer $i$ contains layer $i$ up to layer $k-i$ and in total $k$ layers. Note that an annulus of size $1$ is equivalent to a layer.

The following lemma says that if we are given two disjoint annuli for which we have a local maximal matching each, then we can enhance such matchings to a maximal matching for the induced subgraph of the two annuli combined in constant rounds (see also \Cref{fig:annuli}a for a sketch).
\begin{figure}[t]
    \centering \includegraphics[height=0.2\textheight]{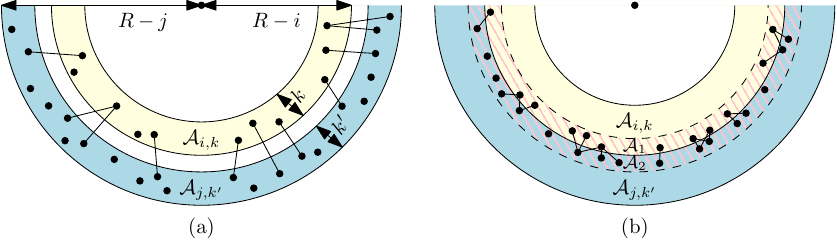}
    \caption{(a) Illustration of unmatched vertices in two annuli (blue and yellow area) and the set of edges available to enhance the matching $M_1 \cup M_2$. (b) The hatched area is the only area where conflicts for the matching might occur when merging two annuli.}
    \label{fig:annuli}
\end{figure}
\begin{lemma}[Merging two annuli]\label{lem:block-merging}
Let $G$ be a threshold hyperbolic random graph where vertices are given their geometric coordinates as an input. Let $i-j \geq k \geq k'$ such that $U_1:= V \cap \block{i}{k}$, $U_2:= V \cap \block{j}{k'}$, $G_1:= G[U_1]$ and $G_2:= G[U_2]$. If we are given matchings $M_1 \subseteq E(G_1)$ and $M_2 \subseteq E(G_2)$ such that $M_1$ is a maximal matching for $G_1$ and $M_2$ a maximal matching for $G_2$, then in \CONGEST, we can compute in $\bigO(1)$ rounds a matching $M \subseteq E(G_1 \cup G_2)\setminus(M_1 \cup M_2)$ such that $M_1 \cup M_2 \cup M$ is a maximal matching for $G[U_1 \cup U_2]$.    
\end{lemma}
\begin{proof}
We consider two cases. First, if the two annuli $\block{i}{k}$ and $\block{j}{k'}$ have distance larger than $4$, i.e., for any pair of points $x \in \block{i}{k}$ and $y \in \block{j}{k'}$ it holds $\dist(x,y) \geq 4$.

\smallskip

\textbf{Case 1}~[For the two annuli it holds $i-j \geq k + 4$]: consider the given matchings $M_1$ and $M_2$. Under these matchings, all unmatched vertices $v \in V \cap \block{j}{k'}$ send to any neighbour $w \in N(v)$ their layer $\ell$ in constant rounds. Then, if an unmatched vertex $u \in  V \cap \block{i}{k}$ has at least one unmatched neighbour in $V \cap \block{j}{k'}$, let $N'(u):= N(u) \cap \block{j}{k'}$ and $u$ picks uniform at random a neighbour $v$ from $N'(u)$ and $\{u,v\}$ is included in our matching $M$ such that $v \in V \cap \block{j}{k'}$. It is left to show that no unmatched vertex $v \in V \cap \block{j}{k'}$ can receive matching requests from two distinct vertices $u, u' \in V \cap \block{i}{k}$ in the same round. To see this, recall that under matching $M_1$, for any unmatched pair $u$ and $u'$, there is no edge $\{u,u'\} \in E$ since $M_1$ is maximal such that two unmatched vertices cannot be adjacent. Hence, it holds $\dist(u,u') \geq R$. Now, for the sake of contradiction, assume that $v \in N(u) \cap N(u')$. By our case assumption, it holds $r(u), r(u') \leq r(v) - 4$. By an application of \Cref{lem:triangle-inequality}, $v$ has distance at most $R$ to both $u$ and $u'$ if and only if $\dist(u,u') \leq R$. A contradiction, and thus, we can obtain the desired matching $M$ in constant rounds.

\smallskip

\textbf{Case 2}~[For the two annuli it holds $i-j < k + 4$]: let $m := \min(i+3, k)$ and $m' := \max(j, k'-3)$. We consider the two "border" annuli $$\mathcal{A}_1 := \bigcup_{\ell = i}^m (\Layer{\ell}) \subseteq \block{i}{k} \text{ and } \mathcal{A}_2 := \bigcup_{\ell = m'}^{k'} (\Layer{\ell}) \subseteq \block{j}{k'}.$$ 

For a visualisation, see the hatched area in \Cref{fig:annuli}b. Note that we can find a maximal matching $M'$ for the induced subgraph $G' = (V \cap (\block{i}{k} \cup (\block{j}{k'} \setminus \mathcal{A}_2)), E \setminus (M_1 \cup M_2))$ in constant rounds by the previous case when $i-j \geq k + 4$. Analogously, we can find a matching $M''$ for the induced subgraph $G'' = (V \cap ( \block{j}{k'} \cup (\block{i}{k} \setminus \mathcal{A}_1)), E \setminus (M_1 \cup M_2 \cup M'))$ in constant rounds. Hence, after constant rounds, we can obtain a matching $M_1 \cup M_2 \cup M' \cup M'' =: \tilde{M}$, such that for a maximal matching $M \supseteq \tilde{M}$ where $M \subseteq E(G_1 \cup G_2)$, all edges $M\setminus \tilde{M}$ are formed among vertex pairs $u,v$ where $u \in V \cap \mathcal{A}_1$ and $v \in V \cap \mathcal{A}_2$. 

We then find such a set of edges $M\setminus \tilde{M}$ in additional constant rounds as follows: we use $m-i+1$ iterations (iterating layer by layer through $\mathcal{A}_1$), where we consider in iteration $t \in [m-i+1]$ all unmatched vertices of the set $V \cap \Layer{m-t}$. That is, we iterate through the layers of $\mathcal{A}_1$ in top-down fashion. By using Claim \ref{claim:tile-single-vertex}, we can assume that each tile in $\Layer{m-t}$ contains at most $1$ unmatched vertex after constant rounds. Then, for $t' \in [40]$, in step $(t, t')$, we activate the set of unmatched vertices where $u \in V \cap \tiletile{\ell}{b}$ fulfils $b \equiv t' \ (\mathrm{mod}\ 40)$. That is, we activate every $40$-th tile. We then match in parallel any such active vertex $u$ with an unmatched vertex $v \in V \cap \mathcal{A}_2$ uniform at random from the set of unmatched neighbours of $u$ in $\mathcal{A}_2$. This can be done in $\bigO(1)$ rounds. Using \Cref{lem:tiling}, it holds for any pair of active vertices $u,u'$ that $N(u) \cap N(u') \cap \mathcal{A}_2 = \emptyset$ and thus, no conflict occurs by this procedure since any active tile contains at most one unmatched vertex. Since $\bigO(40\cdot(m-i+1))$ is constant, this procedure requires $\bigO(1)$ rounds. Using the obtained matching by this procedure in conjunction with our matching $\tilde{M}$ (which we acquired in constant rounds), we obtain our desired matching $M$ in $\bigO(1)$ rounds, which is maximal by similar arguments as used in \Cref{lem:matching-in-layer}: in iteration $(t,t')$, after processing a tile, its unique unmatched vertex has either been matched or has no unmatched neighbour in $\mathcal{A}_2$. Since every tile of $\mathcal{A}_1$
is processed exactly once, no unmatched edge remains between the border annulli $\mathcal{A}_1$ and $\mathcal{A}_2$. Together with the maximality of $\tilde{M}$ on all remaining vertex pairs, this implies that $M$ is a maximal matching of $G[U_1 \cup U_2]$ as desired.
\end{proof}
We now put everything together to obtain a maximal matching in $\bigO(\log\log\log n)$ rounds if the embedding of the graph is given.
\begin{proposition}[Embedding-aware Maximal Matching]
     There exists a $\bigO(\log\log\log n)$-round \CONGEST algorithm to find a Maximal Matching on HRGs \whp, if the vertices have access to their geometric coordinates
\end{proposition}
\begin{proof}
\textbf{Match inner-disk.} By applying \Cref{lem:match-inner-disk}, in constant rounds we can find a matching $M \subseteq E$, such that for any $u \in V \cap \B_0(R - 6\log\log n)$, there exists an edge $E(u) \in M$ \wehp For the rest of the proof we condition on this matching and focus on matching vertices outside $ \B_0(R - 6\log\log n)$ since all vertices in $V \cap \B_0(R - 6\log\log n)$ are matched in constant rounds.

\smallskip

\textbf{Match within layers.} Next, we consider the $\lceil 6\log\log n\rceil$ layers of the outer disk $\disk \setminus \B_0(R - \lceil 6\log\log n\rceil)$. Using \Cref{lem:matching-in-layer}, we compute in parallel, for every layer $\Layer{\ell}$ contained in $\disk \setminus \B_0(R - \lceil 6\log\log n\rceil)$, a maximal matching in $G[V_\ell]$ in constant rounds \whp

\smallskip

\textbf{Match annuli.} Now, we finish the proof by induction.

\smallskip

\textbf{Base case:} For each $i \in [\lceil 3 \log\log n\rceil]$, consider the annuli
$
    \mathcal{A}_{2i,2} = \Layer{2i} \cup \Layer{2i+1}.
$
That is, each annulus consists of two consecutive layers. Using
\Cref{lem:block-merging}, we compute in parallel, for every $i \in [\lceil 3
\log\log n\rceil]$, a maximal matching for the annulus $\mathcal{A}_{2i,2}$ in constant rounds.

\textbf{Induction step:} Consequently, we can repeat this process $t$ times, doubling the annulus size in each step. More formally, in iteration $t \in [\lceil \log_2\log\log n + 3 \rceil] \setminus\{0\}$, for $i \in [\lceil 6\cdot 2^{-t} \log\log n\rceil]$ and $k = 2^t$ we consider the annuli $\block{k \cdot i}{k}$. By our induction hypothesis, for all $i' \in [\lceil 6\cdot 2^{-(t-1)} \log\log n\rceil]$ and $k' = 2^{t-1}$ we have a maximal matching for all the annuli $\block{k' \cdot i'}{k'}$ after $\bigO(t-1)$ rounds. Thus, by an application of \Cref{lem:block-merging}, we obtain, after an additional constant number of rounds, a maximal matching for all annuli $\block{k \cdot i}{k}$ after iteration $t$. That is, we ``merge'' two ``neighbouring'' annuli in constant rounds.

Setting $t= \lceil \log_2\log\log n + 3 \rceil $ we obtain for an annulus of size $k\geq \lceil 6\log\log n\rceil$ a maximal matching in $\bigO(\log\log\log n)$ rounds. Thus, we can find a maximal matching for the outer disk $\disk \setminus \B_0(R - \lceil 6\log\log n\rceil)$ in $\bigO(\log\log\log n)$ rounds. This, in conjunction with \Cref{lem:match-inner-disk}, yields a maximal matching for $G$ in $\bigO(\log\log\log n)$ rounds \whp as claimed.
\end{proof}
    
\printbibliography 

@String{FOCS = {Proceedings of the IEEE Symposium on Foundations of Computer Science (FOCS)}}

@String{PODC = {Proceedings of the ACM Symposium on Principles of Distributed Computing (PODC)}}

@String{SODA = {Proceedings of the SIAM-ACM Symposium on Discrete Algorithms (SODA)}}

@String{STOC = {Proceedings of the ACM Symposium on Theory of Computing (STOC)}}

@String{DISC = {Proceedings of the International Symposium on Distributed Computing (DISC)}}

@String{STACS = {{Proceedings of the International Symposium on Theoretical Aspects of Computer Science (STACS)}}}

@Article{luby86,
  author  = {M. Luby},
  journal = {SIAM Journal on Computing},
  title   = {{A Simple Parallel Algorithm for the Maximal Independent Set Problem}},
  year    = {1986}
}

@article{cl-ccrgg-02,
  author =	 {Chung, Fan and Lu, Linyuan},
  title =	 {{Connected Components in Random Graphs with Given
                  Expected Degree Sequences}},
  journal =	 {Annals of Combinatorics},
  year =	 2002,
  doi =          {10.1007/PL00012580},
}

@book{p-rgg-03,
  author =       {Penrose, Mathew},
  title =        {Random Geometric Graphs},
  year =         2003,
  publisher =    {Oxford University Press},
  isbn =         9780198506263,
}

@article{cl-adrgged-02,
  author =       {Chung, Fan and Lu, Linyuan},
  title =        {The Average Distances in Random Graphs with Given
                  Expected Degrees},
  journal =      {Proceedings of the National Academy of Sciences},
  year =         2002,
  doi =          {10.1073/pnas.252631999}
}

@article{vhhk-s-19,
  author =       {Voitalov, Ivan and van der Hoorn, Pim and van der
                  Hofstad, Remco and Krioukov, Dmitri},
  title =        {Scale-free Networks Well Done},
  journal =      {Physical Review Research},

  year =         2019,
  doi =          {10.1103/PhysRevResearch.1.033034}
}

@article{kpk-h-10,
  author =       {Krioukov, Dmitri and Papadopoulos, Fragkiskos and Kitsak,
                  Maksim and Vahdat, Amin and Bogu\~n\'a, Mari\'an},
  title =        {Hyperbolic Geometry of Complex Networks},
  journal =      {Physical Review E},
  year =         2010,
  doi =          {10.1103/PhysRevE.82.036106},
}

@article{bf-evacaga-22,
 title     = {On the External Validity of Average-case Analyses of Graph Algorithms},
  author    = {Thomas Bläsius and Philipp Fischbeck},
  year      = {2024},
  journal   = {ACM Transactions on Algorithms (TALG)},
  doi       = {10.1145/3633778}
}

@article{bfk-eggihrg-22,
  author =	 {Bläsius, Thomas and Friedrich, Tobias and Katzmann,
                  Maximilian and Meyer, Ulrich and Penschuck, Manuel
                  and Weyand, Christopher},
  title =	 {Efficiently generating geometric inhomogeneous and
                  hyperbolic random graphs},
  journal =	 {Network Science},
  year =	 2022,
  DOI =		 {10.1017/nws.2022.32}
}

@article{bfm-giant-15,
    author = {Bode, Michel and Fountoulakis, N. and Müller, Tobias},
year = {2015},
title = {On the largest component of a hyperbolic model of complex networks},
journal = {Electronic Journal of Combinatorics},
doi = {10.1214/17-AAP1314}
}

@article{fm-giant-18,
 URL = {https://www.jstor.org/stable/26542317},
 author = {Nikolaos Fountoulakis and Tobias Müller},
 journal = {The Annals of Applied Probability},
 title = {Law of large numbers for the largest component in a hyperbolic model of complex networks},
 year = {2018}
}

@inproceedings{gpp-rhg-12,
  author =       {Luca Gugelmann and Konstantinos Panagiotou and Ueli
                  Peter},
  title =        {{Random Hyperbolic Graphs: Degree Sequence and
                  Clustering}},
  booktitle =    {ICALP'12},
  year =         2012,
  doi =          {10.1007/978-3-642-31585-5_51},
}

@article{ms-k-19,
  author =       {Müller, Tobias and Staps, Merlijn},
  title =        {{The Diameter of KPKVB Random Graphs}},
  journal =      {Advances in Applied Probability},
  year =         2019,
  DOI =          {10.1017/apr.2019.23},
}

@article{km-slcrhg-19,
  author =       {Kiwi, Marcos and Mitsche, Dieter},
  title =        {On the Second Largest Component of Random Hyperbolic Graphs},
  journal =      {SIAM Journal on Discrete Mathematics},
  year =         2019,
  doi =          {10.1137/18M121201X},
}

@article{fk-dhrg-18,
  author =       {Tobias Friedrich and Anton Krohmer},
  title =        {{On the Diameter of Hyperbolic Random Graphs}},
  journal =      {SIAM Journal on Discrete Mathematics},
  year =         2018,
  doi =          {10.1137/17M1123961},
}

@article{bfk-chrg-18,
author = {Bl\"{a}sius, Thomas and Friedrich, Tobias and Krohmer, Anton},
title = {Cliques in Hyperbolic Random Graphs},
year = {2018},
journal = {Algorithmica},
doi = {10.1007/s00453-017-0323-3}
}

@phdthesis{katz-diss-23,
  doi = {10.25932/PUBLISHUP-58296},
  url = {https://publishup.uni-potsdam.de/58296},
  author = {Katzmann,  Maximilian},
  title = {About the analysis of algorithms on networks with underlying hyperbolic geometry},
  school = {Universit\"{a}t Potsdam},
  year = {2023},
}

@InProceedings{bks-hudg-23,
  author =	{Bl\"{a}sius, Thomas and Friedrich, Tobias and Katzmann, Maximilian and Stephan, Daniel},
  title =	{{Strongly Hyperbolic Unit Disk Graphs}},
  booktitle =	{STACS'23},
year =	 2023,
  doi =		{10.4230/LIPIcs.STACS.2023.13}
}

@inproceedings{bmrs-stacs-25,
  author       = {Samuel Baguley and
                  Yannic Maus and
                  Janosch Ruff and
                  George Skretas},
  title        = {Hyperbolic Random Graphs: Clique Number and Degeneracy with Implications
                  for Colouring},
  booktitle    = {STACS'25},
  year         = {2025},
  doi          = {10.4230/LIPICS.STACS.2025.13}
}

@inproceedings{bfkrz-esa-2023,
  author       = {Thomas Bl{\"{a}}sius and
                  Tobias Friedrich and
                  Maximilian Katzmann and
                  Janosch Ruff and
                  Ziena Zeif},
 
  title        = {On the Giant Component of Geometric Inhomogeneous Random Graphs},
  booktitle    = {ESA'23},
  year         = {2023},
  doi          = {10.4230/LIPICS.ESA.2023.20}
}

@inproceedings{pkbv-info-2010,
  author       = {Fragkiskos Papadopoulos and
                  Dmitri V. Krioukov and
                  Mari{\'{a}}n Bogu{\~{n}}{\'{a}} and
                  Amin Vahdat},
  title        = {Greedy Forwarding in Dynamic Scale-Free Networks Embedded in Hyperbolic
                  Metric Spaces},
  booktitle    = {INFOCOM'10},
  year         = {2010},
  doi          = {10.1109/INFCOM.2010.5462131}
}

@article{Kiwi2024,
  title = {Cover and hitting times of hyperbolic random graphs},
  DOI = {10.1002/rsa.21249},
  journal = {Random Structures \& Algorithms},
  author = {Kiwi,  Marcos and Schepers,  Markus and Sylvester,  John},
  year = {2024}
}

@article{Newman2003,
  title = {Why social networks are different from other types of networks},
  DOI = {10.1103/physreve.68.036122},
  journal = {Phys. Rev. E},
  author = {Newman,  M. E. J. and Park,  Juyong},
  year = {2003}
}

@article{PhysRevE.74.056114,
  title = {Clustering in complex networks. {I.} General formalism},
  author = {Serrano, M. \'Angeles and Bogu\~n\'a, Mari\'an},
  journal = {Phys. Rev. E},
  year = {2006},
  doi = {10.1103/PhysRevE.74.056114}
}

@article{Watts1998,
  title = {Collective dynamics of ‘small-world’ networks},
  DOI = {10.1038/30918},
  journal = {Nature},
  author = {Watts,  Duncan J. and Strogatz,  Steven H.},
  year = {1998}
}

@article{faloutsos1999power,
  title={On power-law relationships of the internet topology},
  author={Faloutsos, Michalis and Faloutsos, Petros and Faloutsos, Christos},
  journal={ACM SIGCOMM computer communication review},
  year={1999}
}

@inproceedings{DBLP:conf/analco/KiwiM15,
  author       = {Marcos Kiwi and
                  Dieter Mitsche},
  title        = {A Bound for the Diameter of Random Hyperbolic Graphs},
  booktitle    = {ANALCO'15},
  year         = {2015},
  doi          = {10.1137/1.9781611973761.3}
}

@inproceedings{bfk-tw-2016,
  title     = {Hyperbolic Random Graphs: Separators and Treewidth},
  author    = {Thomas Bläsius and Tobias Friedrich and Anton Krohmer},
  year      = {2016},
  booktitle = {ESA},
  doi       = {10.4230/LIPIcs.ESA.2016.15}
}

@article{katzmann-approxvc-2023,
  author       = {Thomas Bl{\"{a}}sius and
                  Tobias Friedrich and
                  Maximilian Katzmann},
  title        = {Efficiently Approximating Vertex Cover on Scale-Free Networks with
                  Underlying Hyperbolic Geometry},
  journal      = {Algorithmica},
  year         = {2023}
}

@article{katzmann-exactvc-2023,
  author       = {Thomas Bl{\"{a}}sius and
                  Philipp Fischbeck and
                  Tobias Friedrich and
                  Maximilian Katzmann},
  title        = {Solving Vertex Cover in Polynomial Time on Hyperbolic Random Graphs},
  journal      = {Theory Comput. Syst.},
  year         = {2023}
}

@article{bffkmm-icalp-2022,
  title     = {Efficient Shortest Paths in Scale-Free Networks with Underlying Hyperbolic Geometry},
  author    = {Thomas Bläsius and Cedric Freiberger and Tobias Friedrich and Maximilian Katzmann and Felix Montenegro-Retana and Marianne Thieffry},
  year      = {2022},
  journal   = {ACM Transactions on Algorithms (TALG)},
  doi       = {10.1145/3516483}
}

@article{Linial-92,
author = {Linial, Nathan},
title = {Locality in Distributed Graph Algorithms},
journal = {SIAM Journal on Computing},
year = {1992},
doi = {10.1137/0221015}
}

@book{Peleg-00,
author = {Peleg, David},
title = {Distributed computing: a locality-sensitive approach},
year = {2000}
}

@article{Naor91,
  author    = {Naor, M.},
  title     = {A Lower Bound on Probabilistic Algorithms for Distributive Ring Coloring},
  journal   = {{SIAM} J. Discrete Math.},
  year      = {1991}
}

@inproceedings{GG24,
  author       = {Mohsen Ghaffari and
                  Christoph Grunau},
  title        = {Near-Optimal Deterministic Network Decomposition and Ruling Set, and
                  Improved {MIS}},
  booktitle    = {FOCS'24},
  year         = {2024},
  doi          = {10.1109/FOCS61266.2024.00007}
}

@InProceedings{CLP18,
 author    = {Yi{-}Jun Chang and
               Wenzheng Li and
               Seth Pettie},
  title     = {An optimal distributed ({\(\Delta\)}+1)-coloring algorithm?},
  booktitle = {STOC'18},
  year      = {2018},
}

@inproceedings{RG20,
	author    = {V{\'{a}}clav Rozho\v{n} and
	Mohsen Ghaffari},
	title     = {Polylogarithmic-time deterministic network decomposition and distributed
	derandomization},
	booktitle = {STOC'20},
	year      = {2020}
}

@phdthesis{Krohmer2016,
  author      = {Anton Krohmer},
  title       = {Structures \& algorithms in hyperbolic random graphs},
url          = {https://publishup.uni-potsdam.de/frontdoor/index/index/docId/39597},
  type        = {doctoralthesis},
  school      = {Universit{\"a}t Potsdam},
  year        = {2016},
}

@article{hrg-spectral,
  title = {Spectral gap of random hyperbolic graphs and related parameters},
  DOI = {10.1214/17-aap1323},
  journal = {The Annals of Applied Probability},
  author = {Kiwi,  Marcos and Mitsche,  Dieter},
  year = {2018}
}

@article{DBLP:journals/im/AbdullahFB17,
  author       = {Mohammed Amin Abdullah and
                  Nikolaos Fountoulakis and
                  Michel Bode},
  title        = {Typical distances in a geometric model for complex networks},
  journal      = {Internet Math.},
  year         = {2017},
  doi          = {10.24166/IM.13.2017}

}

@article{BEPS,
author = {Barenboim, Leonid and Elkin, Michael and Pettie, Seth and Schneider, Johannes},
title = {The Locality of Distributed Symmetry Breaking},
year = {2016},
journal = {Journal of the ACM (JACM)},
doi = {10.1145/2903137}
}

@INPROCEEDINGS{Linial-87,
  author={Linial, Nathan},
  booktitle={FOCS'87}, 
  title={Distributive graph algorithms Global solutions from local data}, 
  year={1987},
  doi={10.1109/SFCS.1987.20}}

@article{Bogu2010,
  title = {Sustaining the Internet with hyperbolic mapping},
  DOI = {10.1038/ncomms1063},
  journal = {Nature Communications},
  author = {Boguñá,  Marián and Papadopoulos,  Fragkiskos and Krioukov,  Dmitri},
  year = {2010}}

@article{Serrano2008,
  title = {Self-Similarity of Complex Networks and Hidden Metric Spaces},
  DOI = {10.1103/physrevlett.100.078701},
  journal = {Physical Review Letters},
  author = {Serrano,  M. Ángeles and Krioukov,  Dmitri and Boguñá,  Marián},
  year = {2008}
}

@Book{barenboimelkin_book,
  author = 	 {Leonid Barenboim and Michael Elkin},
  title = 	 {Distributed Graph Coloring: Fundamentals and Recent Developments},
  publisher = 	 {Morgan \& Claypool Publishers},
  year = 	 2013}

@article{Komjthy2024,
  title = {Polynomial growth in degree-dependent first passage percolation on spatial random graphs},
  url = {http://dx.doi.org/10.1214/24-EJP1216},
  DOI = {10.1214/24-ejp1216},
  journal = {Electronic Journal of Probability},
  author = {Komjáthy,  Júlia and Lapinskas,  John and Lengler,  Johannes and Schaller,  Ulysse},
  year = {2024}
}

@article{Bringmann2019,
  title = {Geometric inhomogeneous random graphs},
  url = {http://dx.doi.org/10.1016/j.tcs.2018.08.014},
  DOI = {10.1016/j.tcs.2018.08.014},
  journal = {Theoretical Computer Science},
  author = {Bringmann,  Karl and Keusch,  Ralph and Lengler,  Johannes},
  year = {2019}
}

@article{yannic-greedy,
title = {Greedy routing and the algorithmic small-world phenomenon},
journal = {Journal of Computer and System Sciences},
year = {2022},
doi = {https://doi.org/10.1016/j.jcss.2021.11.003},
url = {https://www.sciencedirect.com/science/article/pii/S0022000021001112},
author = {Karl Bringmann and Ralph Keusch and Johannes Lengler and Yannic Maus and Anisur R. Molla}
}

@article{johannes-2026,
  doi = {10.48550/ARXIV.2410.22186},
  url = {https://arxiv.org/abs/2410.22186},
  author = {Cerf,  Sacha and Dayan,  Benjamin and De Ambroggio,  Umberto and Kaufmann,  Marc and Lengler,  Johannes and Schaller,  Ulysse},
  title = {Balanced Bidirectional Breadth-First Search on Scale-Free Networks},
  journal = {arXiv},
  year = {2024}
}

@inproceedings{mr-soda-26,
  author       = {Yannic Maus and
                  Janosch Ruff},
  title        = {On Distributed Colouring of Hyperbolic Random Graphs},
  booktitle    = {SODA'26},
  year         = {2026},
  url          = {https://doi.org/10.1137/1.9781611978971.91},
  doi          = {10.1137/1.9781611978971.91}
}

@article{Komjthy2020,
  title = {Explosion in weighted hyperbolic random graphs and geometric inhomogeneous random graphs},
  url = {http://dx.doi.org/10.1016/j.spa.2019.04.014},
  DOI = {10.1016/j.spa.2019.04.014},
  journal = {Stochastic Processes and their Applications},
  author = {Komjáthy,  Júlia and Lodewijks,  Bas},
  year = {2020}
}

@article{cliques-scale-free-2024,
  title     = {Maximal cliques in scale-free random graphs},
  author    = {Thomas Bläsius and Maximilian Katzmann and Clara Stegehuis},
  year      = {2024},
  journal   = {Network Science},
  doi       = {10.1017/nws.2024.13}
}

@inproceedings{fggkr-soda-23,
  author       = {Salwa Faour and
                  Mohsen Ghaffari and
                  Christoph Grunau and
                  Fabian Kuhn and
                  V{\'{a}}clav Rozhon},
  title        = {Local Distributed Rounding: Generalized to MIS, Matching, Set Cover,
                  and Beyond},
  booktitle    = {SODA'23},
  year         = {2023},
  url          = {https://doi.org/10.1137/1.9781611977554.ch168},
  doi          = {10.1137/1.9781611977554.CH168}
}

@inproceedings{ghaffari-soda-16,
  author       = {Mohsen Ghaffari},
  title        = {An Improved Distributed Algorithm for Maximal Independent Set},
  booktitle    = {SODA'16},
  year         = {2016},
  url          = {https://doi.org/10.1137/1.9781611974331.ch20},
  doi          = {10.1137/1.9781611974331.CH20}
}

@article{ks-stoc-26,
  author       = {Seri Khoury and
                  Aaron Schild},
  title        = {Breaking Barriers for Distributed {MIS} by Faster Degree Reduction},
  journal      = {STOC'26},
  year         = {2026},
  url          = {https://doi.org/10.1145/3798129.3800816},
  doi          = {10.1145/3798129.3800816},
}

@article{khoury2025round,
  author       = {Seri Khoury and
                  Aaron Schild},
  title        = {Round Elimination via Self-Reduction: Closing Gaps for Distributed
                  Maximal Matching},
  journal    = {FOCS'25},
  year         = {2025},
  url          = {https://doi.org/10.1109/FOCS63196.2025.00120},
  doi          = {10.1109/FOCS63196.2025.00120}
}

@article{kmw-jacm-2016,
  author       = {Fabian Kuhn and
                  Thomas Moscibroda and
                  Roger Wattenhofer},
  title        = {Local Computation: Lower and Upper Bounds},
  journal      = {J. {ACM}},
  year         = {2016},
  url          = {https://doi.org/10.1145/2742012},
  doi          = {10.1145/2742012}
}

@article{bbhors-jacm-2019,
author = {Balliu, Alkida and Brandt, Sebastian and Hirvonen, Juho and Olivetti, Dennis and Rabie, Mika\"{e}l and Suomela, Jukka},
title = {Lower Bounds for Maximal Matchings and Maximal Independent Sets},
year = {2021},
url = {https://doi.org/10.1145/3461458},
doi = {10.1145/3461458},
journal = {J. ACM},
}

@article{bgko-podc-23,
  author       = {Alkida Balliu and
                  Mohsen Ghaffari and
                  Fabian Kuhn and
                  Dennis Olivetti},
  title        = {Node and edge averaged complexities of local graph problems},
  journal      = {Distributed Comput.},
  year         = {2023},
  url          = {https://doi.org/10.1007/s00446-023-00453-1},
  doi          = {10.1007/S00446-023-00453-1}
}

@book{koonin2006power,
  title = {Power Laws,  Scale-Free Networks and Genome Biology},
  url = {http://dx.doi.org/10.1007/0-387-33916-7},
  DOI = {10.1007/0-387-33916-7},
  author = {Koonin,  Eugene V. and Wolf,  Yuri I. and Karev,  Georgy P.},
    publisher = {Springer US},
  year = {2006}
}

@article{preferal-attatch,
author = {Albert-László Barabási  and Réka Albert },
title = {Emergence of Scaling in Random Networks},
journal = {Science},
year = {1999},
doi = {10.1126/science.286.5439.509}}

@InProceedings{kostas-diameter,
  author =	{Benjert, Zylan and Lakis, Kostas and Lengler, Johannes and Ravi, Raghu Raman},
  title =	{{The Diameter of (Threshold) Geometric Inhomogeneous Random Graphs}},
  booktitle =	{STACS'26},
  year =	{2026}
}

@inproceedings{kostas-rumour,
  author       = {Marc Kaufmann and
                  Kostas Lakis and
                  Johannes Lengler and
                  Raghu Raman Ravi and
                  Ulysse Schaller and
                  Konstantin Sturm},
  title        = {Rumour Spreading Depends on the Latent Geometry and Degree Distribution
                  in Social Network Models},
  booktitle    = {SODA'26},
  year         = {2026},
  url          = {https://doi.org/10.1137/1.9781611978971.226},
  doi          = {10.1137/1.9781611978971.226}
}

@article{clara-21,
  author       = {Riccardo Michielan and
                  Clara Stegehuis},
  title        = {Cliques in geometric inhomogeneous random graphs},
  journal      = {J. Complex Networks},
  year         = {2021},
  url          = {https://doi.org/10.1093/comnet/cnac002},
  doi          = {10.1093/COMNET/CNAC002}
}

@article{Bet_Michielan_Stegehuis_2025, title={Localized geometry detection in scale-free random graphs}, DOI={10.1017/jpr.2025.10038}, journal={Journal of Applied Probability}, author={Bet, Gianmarco and Michielan, Riccardo and Stegehuis, Clara}, year={2025}}

@article{Jorritsma2025,
  title = {Cluster-size decay in supercritical kernel-based spatial random graphs},
  url = {http://dx.doi.org/10.1214/24-AOP1742},
  DOI = {10.1214/24-aop1742},
  journal = {The Annals of Probability},
  author = {Jorritsma,  Joost and Komjáthy,  Júlia and Mitsche,  Dieter},
  year = {2025}
}

@article{Molla2019,
  title = {Optimal deterministic distributed algorithms for maximal independent set in geometric graphs},
  DOI = {10.1016/j.jpdc.2019.05.012},
  journal = {Journal of Parallel and Distributed Computing},
  author = {Molla,  Anisur Rahaman and Pandit,  Supantha and Roy,  Sasanka},
  year = {2019}
}

@inproceedings{kmnw-disc-05,
author = {Kuhn, Fabian and Moscibroda, Thomas and Nieberg, Tim and Wattenhofer, Roger},
title = {Fast Deterministic Distributed Maximal Independent Set Computation on Growth-Bounded Graphs},
booktitle = {DISC'05},
year = {2005},
url = {https://www.microsoft.com/en-us/research/publication/fast-deterministic-distributed-maximal-independent-set-computation-growth-bounded-graphs/}
}

@article{Schneider2010,
  title = {An optimal maximal independent set algorithm for bounded-independence graphs},
  url = {http://dx.doi.org/10.1007/s00446-010-0097-1},
  DOI = {10.1007/s00446-010-0097-1},
  journal = {Distributed Computing},
  author = {Schneider,  Johannes and Wattenhofer,  Roger},
  year = {2010}
}

@article{Alon1986,
  title = {A fast and simple randomized parallel algorithm for the maximal independent set problem},
  url = {http://dx.doi.org/10.1016/0196-6774(86)90019-2},
  DOI = {10.1016/0196-6774(86)90019-2},
  journal = {Journal of Algorithms},
  author = {Alon,  Noga and Babai,  László and Itai,  Alon},
  year = {1986},
}

@inproceedings{thomas-socg-25,
  author       = {Thomas Bl{\"{a}}sius and
                  Jean{-}Pierre von der Heydt and
                  S{\'{a}}ndor Kisfaludi{-}Bak and
                  Marcus Wilhelm and
                  Geert van Wordragen},
  title        = {Structure and Independence in Hyperbolic Uniform Disk Graphs},
  booktitle    = {SoCG'25},
  year         = {2025},
  url          = {https://doi.org/10.4230/LIPIcs.SoCG.2025.21},
  doi          = {10.4230/LIPICS.SOCG.2025.21}
}

@article{fischer2020improved,
  title={Improved deterministic distributed matching via rounding},
  author={Fischer, Manuela},
  journal={Distributed Computing},
  year={2020}
}

@article{Balliu2026New,
  author       = {Alkida Balliu and
                  Filippo Casagrande and
                  Francesco d'Amore and
                  Dennis Olivetti},
  title        = {New Hardness Results for the {LOCAL} Model via a Simple Self-Reduction},
  journal      = {PODC'26},
  year         = {2026},
  doi          = {10.48550/ARXIV.2510.19972},
}

@article{Barenboim2009,
  title = {Sublogarithmic distributed MIS algorithm for sparse graphs using Nash-Williams decomposition},
  url = {http://dx.doi.org/10.1007/s00446-009-0088-2},
  DOI = {10.1007/s00446-009-0088-2},
  journal = {Distributed Computing},
  author = {Barenboim,  Leonid and Elkin,  Michael},
  year = {2009}
}

@inproceedings{ghaffari2023faster,
  title={Faster deterministic distributed MIS and approximate matching},
  author={Ghaffari, Mohsen and Grunau, Christoph},
  booktitle={STOC'23},
  year={2023}
}

@InProceedings{ghaffariDISC19,
  author =	{Ghaffari, Mohsen and Portmann, Julian},
  title =	{{Improved Network Decompositions Using Small Messages with Applications on MIS, Neighborhood Covers, and Beyond}},
  booktitle =	{DISC'19},
  year =	{2019},
  URL =		{https://drops.dagstuhl.de/entities/document/10.4230/LIPIcs.DISC.2019.18},
  doi =		{10.4230/LIPIcs.DISC.2019.18}
}

@InProceedings{gghir-soda-23,
author = {Mohsen Ghaffari and Christoph Grunau and Bernhard Haeupler and Saeed Ilchi and Václav Rozhoň},
title = {Improved Distributed Network Decomposition, Hitting Sets, and Spanners, via Derandomization},
  booktitle =	{SODA'23},
  year =	{2023},
doi = {10.1137/1.9781611977554.ch97},
URL = {https://epubs.siam.org/doi/abs/10.1137/1.9781611977554.ch97}
}

@inproceedings{lenzen2011mis,
  title={MIS on trees},
  author={Lenzen, Christoph and Wattenhofer, Roger},
  booktitle={PODC'11},
  year={2011}
}

@article{Mtivier2010,
  title = {An optimal bit complexity randomized distributed MIS algorithm},
  url = {http://dx.doi.org/10.1007/s00446-010-0121-5},
  DOI = {10.1007/s00446-010-0121-5},
  journal = {Distributed Computing},
  author = {Métivier,  Y. and Robson,  J. M. and Saheb-Djahromi,  N. and Zemmari,  A.},
  year = {2010}
}

@inproceedings{Coupette2021,
  title = {A Breezing Proof of the KMW Bound},
  url = {http://dx.doi.org/10.1137/1.9781611976496.21},
  DOI = {10.1137/1.9781611976496.21},
  booktitle = {SOSA'21},
  author = {Coupette,  Corinna and Lenzen,  Christoph},
  year = {2021}
}

@inproceedings{Kisfaludi-Bak-SODA20,
  author       = {S{\'{a}}ndor Kisfaludi{-}Bak},
  title        = {Hyperbolic intersection graphs and (quasi)-polynomial time},
  booktitle    = {SODA'20},
  year         = {2020},
  url          = {https://doi.org/10.1137/1.9781611975994.100},
  doi          = {10.1137/1.9781611975994.100}
}

@inproceedings{bdhm-socg26,
  author       = {Thomas Bl{\"{a}}sius and
                  Emil Dohse and
                  Deborah Haun and
                  Laura Merker},
  title        = {Product Structure and Treewidth of Hyperbolic Uniform Disk Graphs},
  booktitle    = {SoCG'26},
  year         = {2026},
  url          = {https://doi.org/10.4230/LIPIcs.SoCG.2026.18},
  doi          = {10.4230/LIPICS.SOCG.2026.18}
}

@article{HarrisSS18,
  author       = {David G. Harris and
                  Johannes Schneider and
                  Hsin{-}Hao Su},
  title        = {Distributed ({\(\Delta\)} +1)-Coloring in Sublogarithmic Rounds},
  journal      = {J. {ACM}},
  year         = {2018},
  url          = {https://doi.org/10.1145/3178120}
}

@article{CKP19,
  author       = {Yi{-}Jun Chang and
                  Tsvi Kopelowitz and
                  Seth Pettie},
  title        = {An Exponential Separation between Randomized and Deterministic Complexity
                  in the {LOCAL} Model},
  journal      = {{SIAM} J. Comput.},
  year         = {2019},
  url          = {https://doi.org/10.1137/17M1117537},
  doi          = {10.1137/17M1117537}

}

@misc{GHMN26,
      title={Robust Shattering Arguments}, 
      author={Mohsen Ghaffari and Magnús M. Halldórsson and Yannic Maus and Alexandre Nolin},
      year={2026},
      eprint={2606.27847},
      archivePrefix={arXiv},
      url={https://arxiv.org/abs/2606.27847}, 
}

@inproceedings{G19,
  author       = {Mohsen Ghaffari},
  title        = {Distributed Maximal Independent Set using Small Messages},
  booktitle    = {SODA'19},
  year         = {2019},
  url          = {https://doi.org/10.1137/1.9781611975482.50},
  doi          = {10.1137/1.9781611975482.50}
}

@inproceedings{PR16,
  author       = {Sriram V. Pemmaraju and
                  Talal Riaz},
  title        = {Using Read-k Inequalities to Analyze a Distributed {MIS} Algorithm},
  booktitle    = {OPODIS'16},
  year         = {2016},
  url          = {https://doi.org/10.4230/LIPIcs.OPODIS.2016.9},
  doi          = {10.4230/LIPICS.OPODIS.2016.9}
}

@inproceedings{quantum-stoc,
  author       = {Alkida Balliu and
                  Sebastian Brandt and
                  Xavier Coiteux{-}Roy and
                  Francesco d'Amore and
                  Massimo Equi and
                  Fran{\c{c}}ois Le Gall and
                  Henrik Lievonen and
                  Augusto Modanese and
                  Dennis Olivetti and
                  Marc{-}Olivier Renou and
                  Jukka Suomela and
                  Lucas Tendick and
                  Isadora Veeren},
  title        = {Distributed Quantum Advantage for Local Problems},
  booktitle    = {STOC'25},
  year         = {2025}
}

@inproceedings{BBKO22,
  author       = {Alkida Balliu and
                  Sebastian Brandt and
                  Fabian Kuhn and
                  Dennis Olivetti},
  title        = {Distributed {\(\Delta\)}-coloring plays hide-and-seek},
  booktitle    = {STOC'22},
  year         = {2022},
  url          = {https://doi.org/10.1145/3519935.3520027},
  doi          = {10.1145/3519935.3520027}
}

@inbook{quantum-soda,
author = {Alkida Balliu and Filippo Casagrande and Francesco d’Amore and Massimo Equi and Barbara Keller and Henrik Lievonen and Dennis Olivetti and Gustav Schmid and Jukka Suomela},
title = {Distributed Quantum Advantage in Locally Checkable Labeling Problems},
booktitle = {SODA'26},
doi = {10.1137/1.9781611978971.49},
URL = {https://epubs.siam.org/doi/abs/10.1137/1.9781611978971.49}
}

@inproceedings{BO20,
  author       = {Sebastian Brandt and
                  Dennis Olivetti},
  title        = {Truly Tight-in-{\(\Delta\)} Bounds for Bipartite Maximal Matching
                  and Variants},
  booktitle    = {PODC'20},
  year         = {2020},
  url          = {https://doi.org/10.1145/3382734.3405745},
  doi          = {10.1145/3382734.3405745}
}

\newpage

\appendix

\section{Luby's Algorithm Retains a Polynomial Degree After Constant Rounds}\label{sec:hrg-aint-nuthin-to-fuck-with}
In this section, we show that a standard Luby algorithm requires more than constant rounds so that the degree of every remaining vertex is at most $n^{o(1)}$ (see \Cref{luby-aint-shatter-my-hrg} for a formal statement). Throughout the section, we call a vertex $v \in V$ a leaf if $\deg(v) = 1$. We use the following fact of HRGs (\cite[Lemma 10]{mr-soda-26}).
\begin{lemma}[Layer-leaves]\label{lem:leaves}
Let $G\sim\hrg$ be a threshold hyperbolic random graph and let $u \in V(G)$ be a vertex in $G$ with radius $r \in R - \omega(1)$. Moreover, let $N_0 := \Layer{0} \cap N(u)$ and let $L_0(u) := \{v \in N_0(u) : \deg(v) = 1 \}$. Then, with non-vanishing probability, $|L_0(u)| \in \Theta(n\cdot e^{-r/2})$.
\end{lemma}
Let $k, d \in \mathbb{Z}^+$ and let $W = \{w_1, w_2, \dots, w_{k}\}$ be a set of vertices. We say that $W$ is a \emph{similar degree path of length $k$} if the induced subgraph $G[W]$ is a path and if the degree of any vertex $w \in W$ is $\deg(w) \in \Theta(d)$. Moreover, we say that a vertex $u \in V\setminus W$ with degree $\deg(u) = d$ has a \emph{similar degree path $W$ of length $k$}, if $G[W\cup \{u\}]$ is a path. The following lemma says that most vertices have a similar degree path, where the upper bound on $r$ guarantees expected degree at least polylogarithmic, while the lower bound ensures that the sectors contain polylogarithmically many vertices.
\begin{lemma}[Similar degree path]\label{lem:similar-degree-path}
Let $G\sim\hrg$ be a threshold hyperbolic random graph and let $u \in V(G)$ be a vertex in $G$ with radius $\frac{\log n - \log\log n}{\alpha} \leq r \leq 2\log n - \frac{2\log\log n}{1-\alpha}$. Then, for any constant $c \in \bigO(1)$, $u$ has a similar degree path $W$ of length $c$ and degree $\Theta(n\cdot e^{-r/2})$ \wehp In particular, the similar degree path $W$ is contained in the area $\{x \in \disk : r - 1/1000 \leq r(x) \leq r  \text{ and } \varphi(u) \leq \varphi(x) \leq \varphi(u) + 6c\cdot (1 - 1/1000)e^{R/2-r}\}$.
\end{lemma}
\begin{proof}
Chosen with hindsight, let $\psi := 2(1 - 1/1000)e^{R/2-r}$ and, w.l.o.g., let $\varphi(u) = 0$. Then consider the sector $\Phi := \{x \in \disk : 0 \leq \varphi(x) \leq 3c \cdot\psi\}$, i.e., the sector that lies counter-clockwise to vertex $u$, where one ray of $\Phi$ intersects $u$. We partition $\Phi$ into $3c$ equally sized sub-sectors, such that for all $i \in [3c+1]\setminus\{0\}$, sector $\Phi_i$ has angle $\psi$. Now, let $\mathcal{S}_i := \Phi_i \cap (\B_0(r)\setminus\B_0(r-1/1000))$. We show that for $\mathcal{S}:= \bigcup_{i=1}^{3c}$, there exists a subset of vertices $W \subseteq V \cap \mathcal{S}$, s.t. $G[W]$ induces the desired similar degree path of $u$ \wehp 
First, we note that by \Cref{lem:vertex-degree}, every vertex $w \in W\cup \{u\}$ has degree $\E{\deg(w)} \in \Theta(n\cdot e^{-r/2})$. Since $r \leq  2\log n - \frac{2\log\log n}{1-\alpha}$, and for every $w \in W$ it holds $r - 1/1000 \leq r(w) \leq r$ a Chernoff and union bound reveals for all $w \in W$ that $\deg(w) \in \Theta(\deg(u))$ with probability $1 - n^{-\omega(1)}$
Next, we observe that the expected number of vertices in $\mathcal{S}_i$ is 
$$
\E{|V \cap \mathcal{S}_i|} = n \mu(\mathcal{S}_i) = \frac{n\psi}{2\pi}\mu(\B_0(r)\setminus\B_0(r-1/1000)) = \Theta(1)n\cdot e^{R/2-r}\cdot e^{-\alpha(R-r)}, 
$$
by our choice of $\psi$ and \Cref{lem:measure-inner-disk}. Using that hypothesis that $r \leq 2\log n - \frac{2\log\log n}{1-\alpha}$ and $R = 2\log n + C$ we then have $\E{|V \cap \mathcal{S}_i|} \in \Omega(\log^2 n)$. Using Poisson distribution for each $\mathcal{S}_i$ and a union bound for $3c$ areas $\mathcal{S}_i$, we have that each $\mathcal{S}_i$ contains a vertex with probability $1 - n^{-\omega(1)}$.
We proceed by showing that, given that each $\mathcal{S}_i$ contains at least one vertex, $u$ has a similar degree path $W$ of length $c$, which finishes the proof. To this end, we use \cite[Lemma 3.1]{gpp-rhg-12} and by our angular width $\psi$, it follows for any $x \in \mathcal{S}_{i}$ and $y \in \mathcal{S}_{j}$ that $\dist(x,y) \leq R$ if $|i - j| \leq 3$ while $\dist(x,y) > R$ if $|i - j| \geq 5$ since $r - 1/1000 \leq r(x), r(y) \leq r$. Consequently, if for every $\mathcal{S}_i$ where $i$ is divisible by $3$ we include exactly one vertex $w \in V \cap \mathcal{S}_i$ in the set $W$, this ensures that $G[W]$ induces a path of length $c$. Moreover,
$u$ has an edge to the unique selected vertex $w\in V\cap\mathcal S_3$, but not to any $w \in V \cap \mathcal{S}_i$ if $i\geq 6$. It follows that $W$ is a similar degree path of length $c$ of $u$. Since every $\mathcal{S}_i$ contains at least one vertex \wehp and all vertices $w \in W$ have degree $\deg(w) \in \Theta(\deg(u))$ \wehp, a union bound over both events finishes the proof, where the 'in particular' statement directly follows from our construction of $W$.
\end{proof}
We now show that an HRG has many vertices that have both polynomially many leaves and a similar degree path of constant length. This structure will be useful to show that Luby does not shatter an HRG in constant rounds.
\begin{lemma}[Luby Obstruction]\label{lem:luby-obstruction}
    Let $G \sim \hrg$ be a threshold hyperbolic random graph and let $\eps < 1 - 1/(2\alpha)$ be a constant larger than $0$. Then, \aas there exist a set of vertices $U(\eps) \subseteq V$ of size $|U(\eps)| \in \Omega\left(n^{1 - 1/2\alpha -\eps}\right)$, where a vertex $u \in U(\eps)$ has the following properties.
    \begin{enumerate}
           \item\label{item:property1}\textbf{Degree of $u$}: The degree of $u$ is $\deg(u) \in \Theta(n^{\eps / 2})$.
        \item\label{item:property2}\textbf{Leaves of $u$}: Vertex $u$ has $\Theta(n^{\eps / 2})$ neighbours that are leaves.
        \item\label{item:property3}\textbf{Similar degree path of $u$}: For any constant $t \in \bigO(1)$, $u$ has a similar degree path $W(u)$ of length at least $2t$.
    \end{enumerate}
     Moreover, for any pair $u, u' \in U(\eps)$ it holds that for any pair $v \in L(u) \cup W(u) \cup \{u\}$ and $v' \in L(u') \cup W(u') \cup \{u'\}$ that $\{v,v'\} \not\in E(G)$.
\end{lemma}
\begin{proof}
    We partition the disk $\disk$ into $\left\lfloor\frac{n\cdot 2\pi}{n^{1/(2\alpha) + \eps}} \right\rfloor =:  k$ sectors, (where $k \geq n^{1 - 1/2\alpha -\eps}\in n^{\Omega(1)}$ due to our choice of $\eps < 1 - 1/(2\alpha)$), such that for $i \in [k]$, a sector $\Phi_i$ has angle $\phi \in \Theta(n^{1/(2\alpha) + \eps}/n)$. We consider one such sector $\Phi_i$ and let $X_i$ be the indicator random variable that is $1$ if $\Phi_i$ contains a vertex $u$ with properties \Cref{item:property1}, \Cref{item:property2} and \Cref{item:property3}. We first show that $\Prob{X_i =1} >0$. Thereafter, we show for $X:= \sum_{i=1}^k X_i$ that $X \in \Omega(k)$ \aas what proves our desired statement.
    \paragraph{\Cref{item:property1}:} Without loss of generality, let $\Phi_i$ have bisector $0$. Then, for $\hat{\ell} := \lceil \eps \cdot \log n \rceil$, let $\mathcal{S} := \{x \in \Layer{\hat{\ell}}: - n^{\alpha \cdot \eps }/n < \varphi(x) < n^{\alpha \cdot \eps }/n\} \subset \Phi_i$. Hence, using \Cref{eq:layer_measure}, we obtain
    $$
    \E{|V\cap \mathcal{S}|} = n\cdot\mu(\mathcal{S}) = \Theta(1)\cdot n^{\alpha\cdot \eps}\cdot e^{-\alpha\hat{\ell}} \in \Theta(1).
    $$
    Thus, since $|V\cap \mathcal{S}|$ follows a Poisson-distribution, we get
    \begin{align}\label{eq:Prob-event1}
        \Prob{\mathcal{E}} >0, 
    \qquad \mathcal{E} := \{|V \cap \mathcal{S}|=1\}. 
    \end{align}
    Next, condition on event $\mathcal{E}$ and let let $u$ denote the unique vertex in $V\cap\mathcal S$. Then it follows
    \begin{align}\label{eq:Prob-event2}
        \Prob{\mathcal{E}_1 | \mathcal{E}} \in 1 - n^{-\omega(1)}, 
    \qquad \mathcal{E}_1 := \{\deg(u) \in \Theta(n^{\eps / 2})\},
    \end{align}
by \Cref{eq:layer_expected_degree} and a Chernoff-bound since $u \in V_{\hat{\ell}}$ where $\hat{\ell} = \lceil \eps \cdot \log n \rceil$. 

\paragraph{\Cref{item:property2}:} Let $|L_0(u)| := |\{v \in N_0(u) : \deg(v) = 1 \}|$ as in \Cref{lem:leaves}. Then, using \Cref{lem:leaves} it follows
   \begin{align}\label{eq:Prob-event3}
        \Prob{\mathcal{E}_2 | \mathcal{E}} >0, 
    \qquad \mathcal{E}_2 := \{|L_0(u)| \in \Theta(n^{\eps / 2})\},
    \end{align}
since $u$ is in layer $\hat{\ell} = \lceil\eps \log n\rceil$ and thus, $r(u) \leq R - \lceil\eps \log n \rceil$.

 \paragraph{\Cref{item:property3}:} Applying \Cref{lem:similar-degree-path}, it follows that 
   \begin{align}\label{eq:Prob-event4}
        \Prob{\mathcal{E}_3 | \mathcal{E}} \in 1 - n^{-\omega(1)}, 
    \qquad \mathcal{E}_3 := \{|W(u)| \geq 2t\},
    \end{align}
    and we remark that for event $\mathcal{E}_3$, the vertices of $W(u)$ are within an area $\mathcal{S}' := \{x \in \disk : R- \hat{\ell} - (1 + 1/1000) \leq r(x) \leq R - \hat{\ell} \text{ and } \varphi(u) \leq \varphi(x) \leq \varphi(u) + 6t\cdot (1 - 1/1000)e^{-R/2 +\hat{\ell} + 1 + 1/1000}\}$ by the 'in particular' statement in \Cref{lem:similar-degree-path}. Thus, $\mathcal{S} \subset \Phi_i$, since $\Phi_i$ has an angle of $\Theta(n^{1/(2\alpha) + \eps}/n)$ while $\mathcal{S}'$ spans an angle $\Theta(n^{\eps}/n) \in o(n^{1/(2\alpha) + \eps}/n)$, using that $R = 2\log n +C$, $t \in \bigO(1)$ and $\hat{\ell} = \lceil \eps \log n \rceil$. Our desired 'moreover' statement then follows since for any pair of sectors $\Phi_i$ and $\Phi_j$ where $\mathcal{E} \cap \mathcal{E}_1 \cap \mathcal{E}_2 \cap \mathcal{E}_3$ occur, we consider $u \in V \cap \Phi_i$ and $u' \in V \cap \Phi_j$ to be the vertices occurring due to event $\mathcal{E}$ and it holds for any pair $v \in L(u) \cup W(u) \cup \{u\}$ and $v' \in L(u') \cup W(u') \cup \{u'\}$ that $\{v,v'\} \not\in E(G)$. This is since every object associated with $u$ that is contained in $\Phi_i$, and distinct sector $\Phi_j$ are separated by an angular distance exceeding the maximal adjacency angle of \Cref{lem:max-angle} for any pair of vertices $v \in L(u) \cup W(u) \cup \{u\}$ and $v' \in L(u') \cup W(u') \cup \{u'\}$; there are no edges between distinct structures belonging to different sectors. Subsequently, if $X \in \Theta(k)$ \aas this finishes the proof.

    To obtain this result, note first that $$\Prob{X_i = 1} = \Prob{\mathcal{E} \cap \mathcal{E}_1 \cap \mathcal{E}_2 \cap \mathcal{E}_3}.$$
    
    Then, using \Cref{eq:Prob-event1}, \Cref{eq:Prob-event2}, \Cref{eq:Prob-event3} and \Cref{eq:Prob-event4} in conjunction with a union bound and conditional probabilities, we have that 
$$ \Prob{\mathcal{E} \cap \mathcal{E}_1 \cap \mathcal{E}_2 \cap \mathcal{E}_3} \geq \left(1 - \Prob{\mathcal{E}_1^C | \mathcal{E}} - \Prob{\mathcal{E}_2^C | \mathcal{E}} -\Prob{\mathcal{E}_3^C | \mathcal{E}}\right)\Prob{\mathcal{E}} > 0,$$
and we conclude that $\E{X} \in \Theta(k) \in n^{\Omega(1)}$.  

We wrap up the proof as follows: we show for any $i \neq j$ that $X_i$ and $X_j$ are independent so that a Chernoff bound yields $X \in \Omega(k)$. To this end, we show that events $\mathcal{E}, \mathcal{E}_1, \mathcal{E}_2,$ and $\mathcal{E}_3$ in a sector $\Phi_i$ only depends on the randomness of the Poisson point process in $\Phi_i$.

\paragraph{Event $\mathcal{E}$:} For event $\mathcal{E}$ this is immediate since $\mathcal{S} \subset \Phi_i$.

\paragraph{Event $\mathcal{E}_1$:} For event $\mathcal{E}_1$, note that by \Cref{eq:empty-inner-disk} there is no vertex in $\B_0(r^*)$ \aas where $r^* = R - \log n/\alpha - \log\log n /2$. Thus, since $u \in V_{\hat{\ell}}$ and $\hat{\ell} = \lceil \eps \log n \rceil$, the sector that we need to reveal to obtain the degree of $u$ has angle at most 
$$n^{\alpha \cdot \eps }/n + 2\theta_R(r^*, R - \hat{\ell} -2) \in \bigO\left(e^{(\hat{\ell} - r^* )/2}\right) \in o(\phi) \text{ \aas,}$$
using $\phi \in \Theta(n^{1/(2\alpha) + \eps}/n)$ and \Cref{lem:max-angle} in conjunction with \Cref{rmk:theta-monotonicity}. That is, a smaller angle than the angle spanned by $\Phi_i$.

\paragraph{Event $\mathcal{E}_2$:} By a similar argument we get for $\mathcal{E}_2$, using that a leaf $v \in L_0(u)$ has radius $r(v) \geq R - 1$ and thus $\angulardist{u}{v} \leq \theta_R(R- \hat{\ell} - 2, r(v)) \in \bigO(n^{-\eps/2})$, \aas that we do not require to reveal a sector with an angle that is larger than $\bigO(n^{-\eps/2}) + \theta_R(R-1, r^*) \in o(\phi)$ using \Cref{lem:max-angle}.

\paragraph{Event $\mathcal{E}_3$:} For event $\mathcal{E}_3$, recall that we remarked that the vertices of $W(u)$ are within an area $$\mathcal{S}' := \{x \in \disk : R- \hat{\ell} - (1 + 1/1000) \leq r(x) \leq R - \hat{\ell} \text{ and } \varphi(u) \leq \varphi(x) \leq \varphi(u) + 6t\cdot (1 - 1/1000)e^{-R/2 +\hat{\ell} + 1 + 1/1000}\}.$$ Thus, \aas, we need to reveal a sector of angle at most $ \bigO\left(e^{-R/2 +\hat{\ell}}\right) + \theta_R(r^*, R - \hat{\ell} -2)$. Using \Cref{lem:max-angle}, $\hat{\ell} = \lceil \eps \log n \rceil$ and $r^* = R - \log n/\alpha - \log\log n /2$, this angle is again $o(\phi)$ and thus, also completely contained within sector $\Phi_i$. 

\smallskip

Finally, using that $\E{X} \in \Theta(k) \in n^{\Omega(1)}$, it follows by a Chernoff bound
$$
\Prob{\B_0(r^*) \cap V = \emptyset} \Prob{X \in \Omega(k) | \B_0(r^*) \cap V = \emptyset} \in (1-o(1)) \cdot (1-e^{-\Omega(k)}).
$$
That is, $|U(\eps)| \in \Omega(k) \in \Omega\left(n^{1 - 1/2\alpha -\eps}\right)$ \aas as desired.
\end{proof}
We now use \Cref{lem:luby-obstruction} to show that for any constant $t \in \bigO(1)$ number of iterations of Luby's algorithm, there remains a vertex $u$ with polynomial degree \aas To formalise this, we write $V_{(t)}$ for the set of vertices that are not removed by Luby after iteration $t$, let $G_t = [V_t]$ and $\Delta(G_t) := \max_{v \in V_t}\deg_t(v)$ where for $v \in V_t$ we used $\deg_t(v):= |\{u \in N(v) : u \in V_t\}|$.
\begin{proposition}\label{luby-aint-shatter-my-hrg}  Let $G \sim \hrg$ be a threshold hyperbolic random graph and let $t \in \bigO(1)$ be any fixed constant. Then $\Delta(G_t) \in n^{\Omega(1)}$ \aas   
\end{proposition}
\begin{proof}
With hindsight, we set $\eps := \frac{1 - \frac{1}{2\alpha}}{2t(t+2)}$ and consider the set of vertices $U(\eps)$ of \Cref{lem:luby-obstruction}. We show that after $t$ rounds there exists \aas a vertex $u \in U(\eps)$ such that $\deg_t(u) \in n^{\Omega(1)}$.

To this end, for $i \in [t]$, let $L_{(i)}(u) := \{v \in N(u) \cap V_{(i)} : \deg(v) = 1\}$ (the ``remaining leaves'' of $u$ after round $i$), $W_{(i)}(u):=
\left(\bigcup_{j=1}^{2(t-i)}\{w_j\}\right)\cap V_{(i)}$\footnote{Recall that $W(u) := \{w_1, w_2, \dots w_k\}$ is the similar degree path of $u$ such that for $j \in [k]$ it holds $\{w_{j}, w_{j+1}\} \in E$.} (the ``remaining'' similar degree path of $u$ after round $i$) and we define the random variable
\begin{align}
    K_{(i)} := |\{u \in U(\eps) \cap V_{(i)} : |L_{(i)}(u)| \in n^{\Omega(1)} \text{ and } |W_{(i)}(u)| \geq 2(t-i)\}|.  
\end{align}
Note that if $K_{(t)} \geq 1$ \aas this is sufficient to prove our desired statement and that $K_{(0)} \in \Omega\left(n^{1 - 1/2\alpha - \frac{1 - \frac{1}{2\alpha}}{2t(t+2)}}\right)$ \aas due to \Cref{lem:luby-obstruction} and our choice of $\eps$. In order to prove $K_{(t)} \geq 1$ \aas, we will use the following claim, which tells us that if $K_{(i)}$ is polynomial in $n$, then so is $K_{(i+1)}$ \aas given that $i < t$.
\begin{claim}\label{claim:appendix}
Let $\eps := \frac{1 - \frac{1}{2\alpha}}{2t(t+2)}$. Then, for $i \in [t]$ it holds
$$
\Pro{K_{(i+1)} \in \Omega\left(n^{1 - 1/2\alpha - \eps -(i+1)\eps(2t+1)} \right) | K_{(i)} \in  \Omega\left(n^{1 - 1/2\alpha - \eps -i\eps(2t+1)} \right)} \in 1 - o(1).
$$
\end{claim}
{
\renewcommand{\qedsymbol}{$\blacksquare$} 
\begin{proof}[Proof of claim] 
Let $U_{(i)} = \{u \in U(\eps) \cap V_{(i)} : |L_{(i)}(u)| \in n^{\Omega(1)} \text{ and } |W_{(i)}(u)| \geq 2(t-1)\}$, i.e., $|U_{(i)}| =  K_{(i)}$. We analyse one round of Luby on $U_{(i)}$. To this end, we fix a vertex $u \in U_{(i)}$ and consider the following events:
\begin{enumerate}
    \item[($\mathcal{E}_u$)] Vertex $u$ draws a random number for Luby that lies in the open interval $(1- \frac{2t + 1}{n^\eps}, 1- \frac{2t}{n^\eps})$.
    
    \item[($\mathcal{E}_{\text{path}}$)] For $j \in [2(t-i)+1 ]\setminus\{0\}$, vertex $w_j \in W(u)$\footnote{Where $w_j \in W(u)$ is the $j$-th vertex of the similar degree path $W(u)$ such that $\{u, w_1\} \in E$ and  $\{w_j, w_{j+1}\} \in E$.} draws a random number for Luby that lies in the open interval $(1- \frac{2t + 1 -i}{n^\eps}, 1- \frac{2t- i}{n^\eps})$. 
    
\item[($\mathcal{E}_{\text{leaves}}$)] Any vertex $v \in N((W(u) \cup \{u\}) \setminus (W(u)) \cup \{u\})$  draws a random number for Luby that lies in the open interval $(0, 1- \frac{2t + 1}{n^\eps})$.   
\end{enumerate}
Observe that if event $\mathcal{E}_u \cap \mathcal{E}_{\text{path}} \cap \mathcal{E}_{\text{leaves}} =: \mathcal{E}$ occurs, then $u \in U_{(i+1)}$. To see this, note that all leaves of $u$ draw a smaller number than $u$ and $u$ is not removed as it is the largest number among neighbours, except for $w_1 \in N(u)$, which does not join the independent set since $w_2 \in N(w_1)$ draws a larger number than $w_1$. This ``chain'' of events continues for the entire similar degree path $W_{(i)}(u)$ by the intersection of the two events $\mathcal{E}_{\text{path}}$ and $\mathcal{E}_{\text{leaves}}$. Hence, only the "terminal" vertex $w_{2(t-i)}$ of $W_{(i)}(u)$ that might join the independent set, removing all its neighbours including the neighbour $w_{2(t-i-1)}$ but no other vertex of the similar degree path $W_{(i)}(u)$. Since $u$ has polynomially many leaves, all necessary conditions for $u \in U_{(i+1)}$ are fulfilled. We continue by lower bounding the probability of event $\mathcal{E}$.

\paragraph{Event $\mathcal{E}_u$:} Since the range of the open interval is $n^{-\eps}$ we have
\begin{align}\label{eq:randomdraw1}
    \Pro{\mathcal{E}_u} = n^{-\eps}.
\end{align}
\paragraph{Event $\mathcal{E}_{\text{path}}$:} Similar the event $\mathcal{E}_u$, the range of the interval of each random number is $n^{-\eps}$. Since each draw is independent and we draw at most $2t$ random numbers, we obtain
\begin{align}\label{eq:randomdraw2}
    \Pro{\mathcal{E}_{\text{path}}} \geq n^{-2t\eps}.
\end{align}
\paragraph{Event $\mathcal{E}_{\text{leaves}}$:} Since for any $w \in W(u) \cup \{u\}$ we have $\deg(w) \in \bigO(n^{\eps/2})$, we can upper bound $m := |N((W(u) \cup \{u\}) \setminus (W(u)) \cup \{u\})| \in \bigO(n^{\eps/2})$ using that $|W(u)|$ is constant since $t \in \bigO(1)$. Applying this with the fact that the desired event for every vertex $v$ of drawing a random number within the interval $(0, 1- \frac{2t + 1}{n^\eps})$ has probability $1- \frac{2t + 1}{n^\eps}$, we obtain via independence among the events for each $v$
\begin{align}\label{eq:randomdraw3}
    \Pro{\mathcal{E}_{\text{leaves}}} \geq \left(1- \frac{2t + 1}{n^\eps}\right)^{m} \in \Omega(1),
\end{align}
where we used that $t$ is constant and $m \in \bigO(n^{\eps/2})$.

Putting \Cref{eq:randomdraw1}, \Cref{eq:randomdraw2} and \Cref{eq:randomdraw3} together and using independence we then obtain
\begin{align}\label{eq:final-prob}
  \Pro{\mathcal{E}} = \Pro{\mathcal{E}_u \cap \mathcal{E}_{\text{path}} \cap \mathcal{E}_{\text{leaves}}} \in \Omega(n^{-(2t +1)\eps}).  
\end{align}
To finish the proof of our claim, let $X_u$ be the indicator random variable that event $\mathcal{E}$ occurs for vertex $u \in U_{(i)}$ and define the random variable $X := \sum_{u\in U_{(i)}} X_u \geq K(i+1)$. We then get by linearity of expectation $\Exp{X} \geq \Pro{\mathcal{E}}\cdot |U_{(i)}| \in \Omega\left(n^{1 - 1/2\alpha - \eps -(i+1)\eps(2t+1)} \right)$ by the assumption that $K_{(i)} \in  \Omega\left(n^{1 - 1/2\alpha - \eps -i\eps(2t+1)} \right)$ and \Cref{eq:final-prob}. Since the random variables $X_u$ are independent by the 'moreover' statement of \Cref{lem:luby-obstruction}, a Chernoff bound gives the desired result that $K_{(i+1)} \in \Omega\left(n^{1 - 1/2\alpha - \eps -(i+1)\eps(2t+1)} \right)$ \aas
\end{proof}}
Now, using that $K_{(0)} \in \Omega\left(n^{1 - 1/2\alpha - \frac{1 - \frac{1}{2\alpha}}{2t(t+2)}}\right)$ \aas by \Cref{lem:luby-obstruction}, and applying Claim~\ref{claim:appendix}iteratively for $t \in \bigO(1)$ rounds, a union bound over the $t$ induction steps yields
\begin{align*}
\Pro{K_{(t)} \in \Omega\left(n^{1 - 1/2\alpha - \eps -(t+1)\eps(2t+1)} \right)} &\geq 1 - \Pro{K_{(0)} \not \in \Omega\left(n^{1 - 1/2\alpha - \frac{1 - \frac{1}{2\alpha}}{2t(t+2)}}\right)}\\ 
&- \Pro{\bigcup_{i=0}^{t-1} K_{(i+1)} \not \in \Omega\left(n^{1 - 1/2\alpha - \eps -(i+1)\eps(2t+1)} \right) | K_{(i)}\in  \Omega\left(n^{1 - 1/2\alpha - \eps -i\eps(2t+1)} \right)}\\ &\in 1- o(1).
\end{align*}
This finishes the proof since by our choice $\eps := \frac{1 - \frac{1}{2\alpha}}{2t(t+2)} > 0$, $K_{(t)} \in n^{\Omega(1)}$ \aas
\end{proof}
\section{Concentration Bounds}
\begin{lemma}[Chernoff bound]\label{lem:Chernoff}
For $i \in [k]$, let $X_i \in \{0,1\}$ be independent random variables and $X = \sum_i X_i$. Then
\begin{align*}
    \Pro{X\geq \frac{3}{2}\Exp{X}} \leq \exp{\left(-\frac{\Exp{X}}{12}\right)} \text{ and }
    \Pro{X\leq \frac{1}{2}\Exp{X}} \leq \exp{\left(- \frac{\Exp{X}}{8}\right)}.
\end{align*}
\end{lemma}
The following is a Chernoff\footnote{For convenience, we refer to both as a Chernoff bound whenever using either of the two.} type deviation bound (see e.g. \cite[Lemma 6]{Kiwi2024}) which is necessary for the distribution of vertices as it follows a Poisson point distribution.
\begin{lemma}[Poisson Chernoff bound]\label{lem:Poisson-Chernoff}
    Let $X$ have a Poisson distribution with mean $\E{X}$. Then for $a \geq e^{3/2}$,
    \begin{align*}
        \Prob{X\leq \frac{1}{2}\E{X}} \leq \exp{(- \E{X}/8)} \text{ and }
    \Prob{X\geq a\cdot\E{X}} \leq \exp{(-a\cdot \E{X} /2)}.
    \end{align*}
\end{lemma}
\section{Round Elimination Lower Bounds with General Error Probabilities}\label{sec:sevenOne}
We extend Theorem~7.1 from \cite{BBKO22} that can provide lower bounds for randomised algorithms in the \LOCAL model for problems that meet the prerequisites of the theorem. We emphasise that we do not claim any major contribution here; we only change their theorem to work with a general error probability instead of a hard-coded error probability of $1/n$. Thus, we refer to \cite{BBKO22} for more background and formal definitions, and here we focus on the differences in the proof. In their work, graph problems are modelled using the black-white formalism, where a locally checkable problem $\Pi = (\Sigma, \mathcal{N}, \mathcal{E})$ is defined by a label space $\Sigma$, node (white) constraints $\mathcal{N}$, and hyperedge (black) constraints $\mathcal{E}$~\cite{BBKO22}.
The goal in some computational models, like the \LOCAL model, is to compute a labelling that satisfies all constraints. MIS and MM can both easily be modelled in this formalism~\Cite{BBKO22,BO20}. 
 Within this framework, \textit{Round Elimination (RE)} acts as a mechanical operator, denoted by $\overline{\mathcal{R}}(\mathcal{R}(\Pi))$, which transforms a given problem into a new variant that is exactly one round easier to solve~\cite{BBKO22}. In a nutshell, if one starts with a problem $\Pi_0$ and applies this operator $t$ times and obtains a problem $\Pi_t$ that cannot be solved in $0$ rounds, one obtains a $t$ round lower bound for $\Pi_0$. However, as analysed in \cite{BBKO22}, the allowed error probabilities also increase with each step of applying the operator and also the label space required to describe the problems increases doubly exponentially with each naïve application of the operator. Thus, to prevent the state space from expanding uncontrollably with each iteration, a \textit{relaxation} is applied. A relaxation is a mapping to a structurally simpler problem $\Pi'$ where any valid solution for $\Pi$ remains a valid solution for $\Pi'$~\cite{BBKO22}. The \textit{label complexity function} $f(\Delta)$ serves as a strict upper bound on the size of the label alphabet ($|\Sigma| \le f(\Delta)$) across all intermediate problems and their relaxations throughout a sequence~\cite{BBKO22}. The main challenge is finding these relaxations such that the label complexity remains small and such that the sequence is long. For MIS and MM, one can limit the label complexity to $2^{\Delta+1}$ and for MM to $5$ while providing sequences of length $\Theta(\Delta)$~\cite{BBKO22,BO20} (see \Cref{lem:treeMISMMLowerbound}). The following theorem implies that such \emph{round elimination lower bound sequences} provide lower bounds for solving the respective problem in the  randomised \LOCAL model with a certain error probability $p$.
\begin{theorem}[Extension of Theorem~7.1~in~\cite{BBKO22}]
\label{thm:sevenOne}
Let $\Pi_0 \rightarrow \Pi_1 \rightarrow \cdots \rightarrow \Pi_t$ be a sequence of problems.
Assume that, for all $0 \le i < t$, and for some function $f$, the following holds:
\begin{itemize}
    \item There exists a problem $\Pi_i'$ that is a relaxation of $\mathcal{R}(\Pi_i)$;

    \item $\Pi_{i+1}$ is a relaxation of $\overline{\mathcal{R}}(\Pi_i')$;

    \item The number of labels of $\Pi_i$, and the ones of $\Pi_i'$, are upper bounded by $f(\Delta)$.
\end{itemize}
Also, assume that $\Pi_t$ has at most $f(\Delta)$ labels and is not $0$-round solvable in the deterministic port numbering model, even if the port numbering assignment satisfies some local constraints $C^{\mathrm{port}}$.

Then, $\Pi_0$ requires
\[
\Omega\left(\min\left\{t,\log_{\Delta}n, \log_{\Delta}\log 1/p-\log_{\Delta}\log f(\Delta)\right\}\right)
\]
rounds in the randomised \LOCAL model with local error probability $p$, even if the port numbering satisfies some local constraints $C^{\mathrm{port}}$.
\end{theorem}
\begin{proof}
The proof is almost verbatim along the proof of the almost identical claim in \cite[Theorem 7.1]{BBKO22}. Thus, we focus on the changes to their proof that are limited to the following two aspects: 
\begin{itemize}
    \item \textbf{Change~1:} In their proof, the error probability is hardcoded to $p=1/n$. This is done in \cite[A.5]{BBKO22}. Instead, doing similar calculations with a general parameter $p$ for the error probability leads to the following claim:
    
    \smallskip
    
    \emph{\textbf{Modified A.5:} Let $\Pi_0 \rightarrow \Pi_1 \rightarrow \cdots \rightarrow \Pi_t$ be a sequence of problems satisfying the conditions
of \Cref{thm:sevenOne}. Any randomised algorithm in the port numbering model\footnote{See \cite{BBKO22} for the precise definition of the port numbering model.} running in strictly less than $\min\{t, \frac{1}{10}
(\log_{\Delta}(\log(1/p)) -
\log_{\Delta} \log f(\Delta))\}$ rounds must fail with probability at least $p$.}
\smallskip
\begin{proof}
    Applying \cite[Lemma A.4]{BBKO22}, after $t$ rounds the error probability of any randomised algorithm is 
    $$
    p_{\text{Error}} \geq \frac{1}{{f(\Delta)^{\Delta}}^{10 t}} \geq \frac{1}{{f(\Delta)^{\log (1/ p) / \log f(\Delta)}}} \geq p,
    $$
where in the second step we used our hypothesis $t < \frac{1}{10}
(\log_{\Delta}( \log( 1/p)) - \log_{\Delta} \log f(\Delta))$. So we obtain $p_{\text{Error}} \geq p$ as desired.
\end{proof}    
\item \textbf{Change~2:} In the proof of \cite[A.6]{BBKO22}, the previous claim is lifted to the \LOCAL model using an indistinguishability argument. This argument only works if nodes cannot see the whole graph, requiring to cap the lower bound at $\Omega(\log_{\Delta}n)$, which is the diameter of an $n$-node $\Delta$-regular tree. In their work,  this does not appear in their theorem as the cap is always larger than the bound $\log_{\Delta}( \log( 1/p)) - \log_{\Delta} \log f(\Delta))$ imposed when used with a hardcoded $p=1/n$. Thus, for our theorem, the lower bound cannot exceed $\log_{\Delta}n$, to which we adjust by adding it as a third argument in the min function. 

\smallskip

\emph{\textbf{Modified A.6:} Let $\Pi_0 \rightarrow \Pi_1 \rightarrow \cdots \rightarrow \Pi_t$ be a sequence of problems satisfying the conditions of\Cref{thm:sevenOne}. Any randomised algorithm running in the \LOCAL model, in $\Delta$-regular balanced trees of $n$ nodes, that fails with probability at most $p$, requires $\Omega\left(\min\left\{t,\log_{\Delta}(\log( 1/p)) -\log_{\Delta}\log f(\Delta), \log_{\Delta} n \right\}\right)$ rounds.}
\end{itemize}

\smallskip

Note that this is exactly our desired statement \Cref{thm:sevenOne}.
\end{proof}
\section{Truncating Angular Coordinates}
By the following lemma, $\bigO(\log n)$ bits per angular coordinate suffice so that each angular coordinate is unique \whp
\begin{lemma}\label{lem:truncation}
    Let $G$ be a threshold hyperbolic random graph and for any pair of vertices $u,v \in V(G)$, let $\varphi_b(u)$ and $\varphi_b(v)$ be the first $b$ random bits of the angular coordinates $\varphi(u)$ and $\varphi(v)$ respectively. Then, for any pair $u,v$, any constant $c > 0$ and $b = \lceil c\cdot \log n \rceil$, it holds $\varphi_b(u) \neq \varphi_b(v) $ with probability $1 - \bigO(n^{-c+5/2})$.
\end{lemma}
\begin{proof}
    Throughout the proof, we work with the binomial model of hyperbolic random graphs, in which the number of vertices is fixed to be $n$. At the end of the proof, the result is transferred to the Poisson model with $\mathbb{E}[N]=n$, incurring only a minor additional error probability (see also \cite[\S 3.3.4, Lemma 3.9]{katz-diss-23} for a detailed discussion).

Let $b = \lceil c\cdot \log n \rceil$ as given by our hypothesis and let $\mathcal{E}$ be the event that there exists a pair of vertices $u,v \in [n]$ where  $\varphi_b(u) = \varphi_b(v)$. Since the angular coordinates are sampled independently and uniformly from $[0,2\pi)$, the first $b$ bits of every angular coordinate form an independent uniformly distributed bit string in $\{0,1\}^b$. Hence, for every fixed pair of distinct vertices $u,v \in [n]$ we have $\Prob{\varphi_b(u)=\varphi_b(v)} = 2^{-b}$. Subsequently, it holds for our "collision" event $\mathcal{E}$ by a union bound
$$
\Prob{\mathcal{E}} \leq {n\choose 2}\cdot 2^{-b}\in\bigO(n^{-c+2}),
$$
where in the last step we plugged in $b = \lceil c\cdot \log n \rceil$. We then obtain the desired probability bound by the transfer from the binomial model to Poisson~\cite[Lemma 3.9]{katz-diss-23} which yields $\Prob{\mathcal{E}}\cdot\bigO(n^{1/2}) \in \bigO(n^{-c+5/2})$.
\end{proof}
\end{document}